\documentclass[11pt]{article}
\usepackage{latexsym}
\usepackage{amssymb}
\usepackage{graphicx}
\usepackage{enumerate}
\usepackage{amssymb}
\usepackage{amsmath}
\usepackage{color}
\usepackage{tikz}
\usepackage{hyperref}
\hypersetup{hypertex=true,
	colorlinks=true,
	linkcolor=blue,
	anchorcolor=blue,
	citecolor=blue}
\usepackage{amsthm}
\usepackage{cleveref}
\crefformat{section}{\S#2#1#3} 
\crefformat{subsection}{\S#2#1#3}
\crefformat{subsubsection}{\S#2#1#3}
\def\Z{\mathbb{Z}}
\def\p{\partial}
\def\C{\mathbb{C}}
\def\o{\omega}
\def\sm{{\mathrm{reg}}}
\def\tE{{\widetilde{E}}}
\def\bz{{\bar{z}}}
\def\deep{{\delta_{\epsilon}}}
\def\Conf{{\mathrm{Conf}}}
\def\bp{\bar{\p}}
\def\bw{{\bar{w}}}
\def\tz{{\tilde{z}}}
\def\tw{{\tilde{w}}}
\def\btz{{\bar{\tz}}}
\def\a{\alpha}
\def\bfe{{\mathbf{e}}}
\def\bfw{{\mathbf{w}}}

\def\bfi{{\mathbf{i}}}
\def\half{{\frac{1}{2}}}
\def\bfk{{\mathbf{k}}}
\def\lan{{\langle}}

\def\bfl{{\mathbf{l}}}
\def\b{\beta}
\def\btw{{\bar{\tw}}}
\def\Op{\mathrm{Op}}
\def\End{\mathrm{End}}
\def\ep{\epsilon}

\numberwithin{equation}{section}

\theoremstyle{plain}
\newtheorem{thm}{Theorem}[section]
\newtheorem{lem}[thm]{{{L}}emma}
\newtheorem{cor}[thm]{{C}orollary}
\newtheorem{prop}[thm]{{P}roposition}

\newtheorem{conj}[thm]{{C}onjecture}
\newtheorem{rem}[thm]{Remark}

\newtheorem{defn}[thm]{{D}efinition}
\def\R{\mathbb{R}}
\def\Xint#1{\mathchoice
	{\XXint\displaystyle\textstyle{#1}}%
	{\XXint\textstyle\scriptstyle{#1}}%
	{\XXint\scriptstyle\scriptscriptstyle{#1}}%
	{\XXint\scriptscriptstyle\scriptscriptstyle{#1}}%
	\!\int}
\def\XXint#1#2#3{{\setbox0=\hbox{$#1{#2#3}{\int}$}
		\vcenter{\hbox{$#2#3$}}\kern-.5\wd0}}

\def\dashint{\Xint-}
\newcommand{\be}{\begin{equation}}
	\newcommand{\ee}{\end{equation}}
\newcommand{\ba}{\begin{aligned}}
	\newcommand{\ea}{\end{aligned}}
\begin{document}
	\title{Feynman Graph Integrals on K\"ahler Manifolds}
	\author{Minghao Wang\footnote{School of Mathematical Sciences, East China Normal University, 200241, Shanghai, China, wmh18@tsinghua.org.cn}
		\and
		Junrong Yan\footnote{Department of Mathematics, Northeastern University, 02115, Boston, USA, j.yan@northeastern.edu} 
	}
	\maketitle
	\begin{abstract}
In this paper, we establish the convergence of Feynman graph integrals on closed real-analytic K\"ahler manifolds and uncover the structural mechanism underlying this convergence. The key insight is that, using Getzler's rescaling technique, the graph integrands extend canonically to the Fulton--MacPherson compactification of configuration spaces as forms with divisorial-type singularities. This allows the Feynman graph integrals to be rigorously defined as Cauchy principal value integrals. As an application, these integrals provide a mathematically rigorous construction of the higher-genus B-model invariants on Calabi–Yau threefolds in the sense of Bershadsky–Cecotti–Ooguri–Vafa (BCOV).
\end{abstract}
	\tableofcontents

	\section{Introduction}
Feynman graph integrals lie at the core of perturbative quantum field theory and have driven developments across geometry, topology, algebra, and analysis in  mathematics. 
While their algebraic and combinatorial structures are well understood, providing a mathematically rigorous meaning remains subtle. 
Understanding the convergence of Feynman integrals in certain quantum field theories has far-reaching applications across mathematics and physics. 
Notable examples include the configuration space integrals for topological quantum field theories developed by S. Axelrod and I. Singer \cite{axelrod1994chern,Axelrod1991ChernSimonsPT} and by M. Kontsevich \cite{kontsevich1994feynman}, which have proved exceptionally powerful. 
These integrals offer modern approaches to knot invariants \cite{kontsevich1994feynman}, 3-manifold invariants \cite{bott2018integral,bott1999integral}, smooth structures on fiber bundles \cite{kontsevich1994feynman}, operad theory \cite{getzler1994operads,kontsevich1999operads}, deformation quantization of Poisson manifolds \cite{kontsevich2003deformation}, and numerous other applications.

In contrast, a systematic theory of Feynman graph integrals in holomorphic quantum field theories has emerged only in the past decade \cite{Costello2015QuantizationOO,Li:2011mi,li2021regularized,li2023regularized,Li:2016gcb,budzik2023feynman,Williams:2018ows,wang2024factorization,wang2025feynman}. 
A major technical challenge is that these integrals are typically not absolutely convergent (see \cite[Example 2.1.1]{wang2025feynman}). 
To overcome this, the first author recently introduced a heat-kernel regularization and established the finiteness of holomorphic field theories on affine spaces \cite{wang2025feynman}. For a comprehensive account of heat-kernel regularization , we refer the reader to \cite{costellorenormalization}.

This work offers an alternative approach to \cite{wang2025feynman}. We show that the Feynman graph integrand extends canonically to the Fulton–MacPherson compactification of configuration spaces with mild singularities, allowing the corresponding graph integrals to be defined rigorously as Cauchy principal value integrals (c.f. \cite{herrera1971residues}). The resulting convergence theorem extends from affine spaces to arbitrary closed \emph{real-analytic K\"ahler manifolds}, thereby broadening the scope of holomorphic quantum field theories.

Similar to topological quantum field theories, the Feynman graph integrals of holomorphic field theories have numerous applications.  Below, we highlight a few:
	\begin{enumerate}[(1)]
		\item Constructions of invariants as in  \Cref{BCOV invariant}.
		\item  Construction of factorization algebras, as in \cite{wang2024factorization}; see also \cite{costello_gwilliam_2016, costello_gwilliam_2021} for foundational discussions on factorization algebras.
		\item Constructions of higher chiral algebras as in \cite{minghaoBrianZhengping2025}.
	\end{enumerate}
	Moreover, the results established in this paper have potential applications in the following areas:
	\begin{enumerate}[(1)]
    	\item \textbf{Mirror Symmetry.} In mirror symmetry, the Gromov–Witten invariants of Calabi–Yau manifolds are conjecturally identified with the Feynman graph integrals arising from the Kodaira–Spencer gravity theory (also known as the BCOV theory, see \cite{bershadsky1994kodaira,costello2012quantum}) on the mirror manifolds. Our results define the higher-genus B-model invariants and provide a rigorous mathematical formulation of the BCOV mirror symmetry conjecture (see also \cref{section of BCOV conjecture}).
		\item \textbf{Gauge field theories on $\mathbb{R}^{4}$.} In the well-known Penrose–Ward correspondence (see \cite{Penrose:1976js,WARD197781}), it is conjectured that perturbative self-dual Yang-Mills theory should correspond to a perturbative holomorphic quantum field theory on the twistor space. We expect that our theorems offer new insights toward understanding this correspondence. For example, our results should imply the finiteness of Feynman graph integrals in self-dual Yang-Mills theory.
		\item \textbf{Holomorphic ``knot" invariants.} Feynman graph integrals arising from topological field theories have proven useful in the study of knot invariants. Our theorems have potential applications in the study of embeddings of Riemann surfaces, viewed as holomorphic ``knots", into complex manifolds. For the study of holomorphic linking numbers in the literature, see \cite{Atiyah:1981ey,Frenkel:2005qk}.

	\end{enumerate}
	
	\subsection{Main results}\label{section of main results}
	For simplicity, and to avoid introducing excessive technical notions at the outset, we present here a special case of our main results (Theorems~\ref{Thm-Fey-Graph-Int} and~\ref{defining function and Feynman graph integral}), which provides a rigorous mathematical construction of the higher-genus B-model in the framework of Bershadsky--Cecotti--Ooguri--Vafa (BCOV).

    Let $M$ be a closed Calabi--Yau threefold equipped with a Calabi--Yau metric and a holomorphic volume form $\Omega$.  
Consider the bundle
\[
\widetilde{E} = \wedge^\ast T^{1,0}M \otimes \wedge^\ast (T^{0,1}M)^\ast,
\]
whose space of sections is
\[
\Gamma(\widetilde{E}) = \Omega^{0,\ast}\big(M, \wedge^\ast T^{1,0}M\big).
\]

Contraction with $\Omega$ gives an isomorphism
\[
(-)^\vee_\Omega : \Omega^{0,j}(M,\wedge^i T^{1,0}M) \;\longrightarrow\; \Omega^{3-i,j}(M).
\]
We view the operators $\partial$ and $\bar\partial$ on $\Omega^{\ast,\ast}(M)$ as operators on $\Omega^{0,\ast}(M,\wedge^\bullet T^{1,0}M)$ via this identification.

Consider the trace map
\be\label{defn of trace}
\operatorname{tr} : \Omega^{0,j}(M,\wedge^i T^{1,0}M) \to \Omega^{6-i,j}(M), \quad
\alpha \mapsto (\alpha^\vee_\Omega)\wedge\Omega.
\ee
Let $P$ be a distribution-valued section of the bundle $\tE\boxtimes\tE$ over $M\times M$, such that the following equations hold (see Definition \ref{ordinary propagator} and \Cref{BCOV's pairing} for details):
	\[
	\begin{cases}
		\left(\bar{\partial} \otimes \operatorname{id}+\operatorname{id} \otimes \bar{\partial}\right) P=\delta-\mathcal{H},\\
		P\in \mathrm{Im}(\bar{\partial}^* \otimes \mathrm{id}),
	\end{cases}
	\]
	where $\delta$ and $\mathcal{H}$ are the delta-distribution and the integral kernel of the harmonic projection operator, respectively.

The \emph{BCOV propagator} is given by
\[
P_{\mathrm{BCOV}} := (\partial\otimes \mathrm{id})\,P.
\]

Let $\tilde{\rho}:M\times M\rightarrow\mathbb{R}_{\geq 0}$ be a non-negative function, such that:
	\begin{enumerate}[(1)]
		\item $\tilde{\rho}^{2}$ is smooth and $\tilde{\rho}^{-1}(0)=\triangle$, where $\triangle:=\{(p,p)\in M\times M : p\in M\}$.
		\item There exists an open neighborhood $U\subset M\times M$ of $\triangle$ such that \[\tilde{\rho}^{2}|_{U}=\rho^{2}|_{U},\]
		where $\rho$ is the distance function on $M$.
	\end{enumerate}
	Then we have
	\begin{thm}\label{main result1}
		Let $\vec{\Gamma}$ be a directed graph, whose vertex set $\vec{\Gamma}_{0}$ and edge set $\vec{\Gamma}_{1}$ are ordered sets. Then
		\be\label{convergence}
		\lim_{\epsilon\rightarrow0}\int_{M^{\vec{\Gamma}_{0}}}\chi_{\epsilon}(p_{1},\cdots,p_{|\vec{\Gamma}_{0}|})\left(\bigotimes_{v\in \vec{\Gamma}_0}\mathrm{tr}_v\right)
\left(\prod_{e\in\vec{\Gamma}_1}
P_{\mathrm{BCOV}}(p_{t(e)},p_{h(e)})\right)
		\ee
		exists, where $t,h:\vec{\Gamma}_{1}\rightarrow\vec{\Gamma}_{0}$ denotes the tail and head of a given directed edge, and $|\vec{\Gamma}_{0}|$ is the number of vertices. Here $\chi_{\epsilon}$ is the function on $M^{|\vec{\Gamma}_{0}|}$ defined by
		\[
		\chi_{\epsilon}(p_{1},\cdots,p_{|\vec{\Gamma}_{0}|})=
		\begin{cases}
			0,&\text{if }\prod_{i<j\in\vec{\Gamma}_{0}}\tilde{\rho}^2(p_i,p_j)\leq \epsilon,\\
			1,&\text{if }\prod_{i<j\in\vec{\Gamma}_{0}}\tilde{\rho}^2(p_i,p_j)> \epsilon,
		\end{cases}
		\]
        and for a finite set $A$, $|A|$ denotes number of elements in $A$.
        
		Moreover, the integral does not depend on the choice of $\tilde{\rho}$.
	\end{thm}

 The convergence of the limits in \eqref{convergence} arises from a structural property of the Feynman graph integrand, described in \Cref{main result 0} below. Let $\widetilde{Conf}_{\vec{\Gamma}_{0}}(M)$ denote the Fulton–MacPherson compactification (see \Cref{Fulton-MacPherson}) of the configuration space of $|\vec{\Gamma}_0|$ points in $M$. We prove the following:

\begin{thm}\label{main result 0} Under the assmuption of \Cref{main result1},
the Feynman graph integrand
\[
\left(\bigotimes_{v\in \vec{\Gamma}_0}\mathrm{tr}_v\right)
\left(\prod_{e\in\vec{\Gamma}_1}
P_{\mathrm{BCOV}}(p_{t(e)},p_{h(e)})\right)
\]
extends canonically to a form on $\widetilde{Conf}_{\vec{\Gamma}_{0}}(M)$
with divisorial-type singularities (\Cref{divisorial type singularities}) .
\end{thm}

	Using Cauchy principal value integrals in Appendix \ref{section Cauchy principal value} and \Cref{main result 0}, we can reformulate Theorem \ref{main result1} as follows:
	\begin{thm}\label{main result2}
		Let $\vec{\Gamma}$ be a directed graph, whose vertex set $\vec{\Gamma}_{0}$ and edge set $\vec{\Gamma}_{1}$ are ordered sets. Then
		\begin{align*}
			&\lim_{\epsilon\rightarrow0}\int_{M^{\vec{\Gamma}_{0}}}\chi_{\epsilon}(p_{1},\dots,p_{|\vec{\Gamma}_{0}|})\left(\bigotimes_{v\in \vec{\Gamma}_0}\mathrm{tr}_v\right)
\left(\prod_{e\in\vec{\Gamma}_1}
P_{\mathrm{BCOV}}(p_{t(e)},p_{h(e)})\right)\\
			&=\dashint_{\widetilde{Conf}_{\vec{\Gamma}_{0}}(M)}\left(\bigotimes_{v\in \vec{\Gamma}_0}\mathrm{tr}_v\right)
\left(\prod_{e\in\vec{\Gamma}_1}
P_{\mathrm{BCOV}}(p_{t(e)},p_{h(e)})\right),
		\end{align*}
		 where the integral on the right hand side is in the sense of Cauchy principal value integral (see Appendix \ref{section Cauchy principal value}).
	\end{thm}
	
	\Cref{main result1}, \Cref{main result 0} and \Cref{main result2} are special cases of our main results Theorem \ref{Thm-Fey-Graph-Int} and \Cref{defining function and Feynman graph integral}, where we consider general closed real analytic K\"ahler manifolds and Hermitian holomorphic vector bundles with holomorphic pairings.
	
	Notably, the proof of our main theorems draws on Getzler’s rescaling technique (see \cite{berline2003heat, getzler1986short}), originally developed for the Atiyah–Singer local index theorem, highlighting a fruitful interplay between index-theoretic methods and Feynman graph integrals.
	In future work, we plan to adapt the machinery employed in the proof of the family version of the Atiyah–Singer local index theorem to the study of Feynman graph integrals.

\subsection{BCOV's formulation of Mirror symmetry conjecture}\label{section of BCOV conjecture}
In this section, under the same setting as in \cref{section of main results}, we apply our convergence theorems to briefly describe the higher-genus B-model invariants in the BCOV framework and to formulate their mirror symmetry conjecture.

It follows from \Cref{main result2} and the graded symmetric property of $P_{\mathrm{BCOV}}$ that
\begin{cor}
Let $\Gamma$ be a trivalent undirected graph.  
Then the integral
\[
W_{\mathrm{BCOV}}(\Gamma)
= 
\dashint_{\widetilde{\mathrm{Conf}}_\Gamma(M)}
\left(\bigotimes_{v\in \Gamma_0}\mathrm{tr}_v\right)
\left(\prod_{e\in\Gamma_1}
P_{\mathrm{BCOV}}(p_{t(e)},p_{h(e)})\right)
\]
is well-defined in the sense of Cauchy principal value.
\end{cor}

\begin{defn} \label{BCOV invariant}
Let $\mathrm{Graph}_3$ denote the set of connected trivalent undirected graphs.
For $g\ge2$, define
\[
F_g^{\mathrm{BCOV}}(M)
= 
\sum_{\substack{\Gamma\in \mathrm{Graph}_3\\ g(\Gamma)=g}}
\frac{1}{|\mathrm{Aut}(\Gamma)|}\, W_{\mathrm{BCOV}}(\Gamma),
\qquad 
g(\Gamma)=|\Gamma_1|-|\Gamma_0|+1.
\]
\end{defn}

Now we can formulate BCOV conjecture mathematical rigorously as follows:
\begin{conj}[Bershadsky–Cecotti–Ooguri–Vafa]
For $g\ge2$, $F_g^{\mathrm{BCOV}}(M)$ defines the genus-$g$ B-model invariant, i.e.
\begin{enumerate}
\item $F_g^{\mathrm{BCOV}}(M)$ is independent of the K\"ahler class of $M$;
\item if $M$ and $M^{\vee}$ are mirror manifolds, then the holomorphic limit (see \cite{kanazawa2015lectures}) of $F_g^{\mathrm{BCOV}}(M)$ coincides with the genus-$g$ Gromov–Witten invariants of the mirror $M^\vee$.
\end{enumerate}
\end{conj}

\begin{rem}
For the elliptic curve $E_\tau$, S.~Li and K.~Costello \cite{li2011bcov,costello2012quantum} constructed a variant of $F_g^{\mathrm{BCOV}}(E_\tau)$ with punctures, whose large complex structure limit as $\bar\tau\to\infty$ (with $\tau$ fixed) recovers the punctured Gromov–Witten invariants.
\end{rem}

	\subsection{Organizations}
	
	In \Cref{sec2}, we introduce the notion of the propagator. Assuming \Cref{propagator has divisorial type singularities}, that is, that the propagator has the desired singularities, we define Feynman graph integrals using Cauchy principal value integrals. This construction forms the core of our main results, Theorem~\ref{Thm-Fey-Graph-Int} and Theorem~\ref{defining function and Feynman graph integral}.

	Sections \Cref{Section 3} and \Cref{Section 4} are devoted to the proof of \Cref{propagator has divisorial type singularities}. In \Cref{Section 3}, we introduce the notion of regular expressions, which can be viewed as the analog of divisorial type singularities in the context of heat kernels. In \Cref{subsection regular expressions and divisorial type singularities}, we establish the connection between regular expressions of the heat kernel and the divisorial type singularities of the propagator, as stated in \Cref{thm41}. This reduces the proof of \Cref{propagator has divisorial type singularities} to verifying that the heat kernel admits a regular expression, namely \Cref{main0}. 
	
	Finally, in \Cref{Section 4}, we prove \Cref{main0} using Getzler’s rescaling technique.

	\subsection*{Acknowledgments}
	We thank Keyou Zeng, Zhengping Gui, Si Li, Ezra Getzler, Maciej Szczesny, Brian Williams, and Kevin Costello for valuable discussions. Part of this work was completed during the first author’s visit to the Institut Mittag-Leffler, whose hospitality he gratefully acknowledges. Part of this work was also carried out while the first author was a postdoctoral associate in the Department of Mathematics and Statistics at Boston University. He gratefully acknowledges the department for its support and stimulating research environment.

	\section{Feynman Graph Integrals on K\"ahler Manifolds}\label{sec2}
	
	In this section, we introduce the Feynman graph integrals from holomorphic quantum field theories. 
	\subsection{Propagators}\label{subsection propagators}
	In this subsection, we introduce the notion of the \emph{propagator} in holomorphic quantum field theory, and establish its existence (\Cref{existence of propagator}) and uniqueness (\Cref{uniqueness of propagator}).

	We begin by recalling some basic facts from Hodge theory on K\"ahler manifolds.
	
	Let $M$ be a closed K\"ahler manifold, and let $E \to M$ be a Hermitian holomorphic vector bundle.
	
	Consider the twisted bundle $\widetilde{E} = E \otimes \Lambda^\bullet T_{0,1}^* M$. The space of smooth sections $\Gamma(\widetilde{E})$ forms the Dolbeault complex of $E$, equipped with the Dolbeault operator $\bar{\partial}_{\widetilde{E}}$ and its formal adjoint $\bar{\partial}_{\widetilde{E}}^*$ with respect to the natural Hermitian inner product. The associated Laplacian is given by
	\[
	\Delta_{\widetilde{E}} = \bar{\partial}_{\widetilde{E}} \circ \bar{\partial}_{\widetilde{E}}^* + \bar{\partial}_{\widetilde{E}}^* \circ \bar{\partial}_{\widetilde{E}}.
	\]
	
	Similarly, the dual bundle $E^{*}$ is a Hermitian holomorphic vector bundle. Define $\widetilde{E}' = E^{*} \otimes K_M \otimes \Lambda^\bullet T_{0,1}^* M$, where $K_M$ denotes the canonical bundle of $M$. The space $\Gamma(\widetilde{E}')$ carries natural operators $\bar{\partial}_{\widetilde{E}'}$, $\bar{\partial}_{\widetilde{E}'}^*$, and the Laplacian $\Delta_{\widetilde{E}'}$.
	
	Moreover, there is a natural pairing between $\Gamma(\widetilde{E}')$ and $\Gamma(\widetilde{E})$ given by
	\be\label{natural pairing}
	( s_{1} \otimes u_1,\, s_{2} \otimes u_2 ) = \int_M \langle s_1, s_2 \rangle\, u_1 \wedge u_2,
	\ee
	where $s_1 \in \Gamma(E^{*})$, $s_2 \in \Gamma(E)$, $u_1 \in \Gamma(K_M \otimes \Lambda^\bullet T_{0,1}^* M)$, and $u_2 \in \Gamma(\Lambda^\bullet T_{0,1}^* M)$, and $\lan s_1, s_2 \rangle$ denotes the natural pairing between $E$ and $E^{*}$.
	\begin{rem}
		In this paper, $ \Gamma(E) $ denotes the space of smooth sections of a vector bundle $ E $, rather than holomorphic sections.  When necessary, we also write $ \Gamma(E) $  as $ \Gamma(M;E) $ to emphasize the base manifold $ M $.
		
	\end{rem}

	\begin{lem}\label{dual operator properties}
		Let $\alpha\in \Gamma(\widetilde{E}')$, $\beta \in \Gamma(\widetilde{E})$. We denote the degrees of the differential forms $\alpha$ and $\beta$ by $|\alpha|$ and $|\beta|$, respectively. Then we have
		\begin{enumerate}[(1)]
			\item \[
			(\bar{\partial}_{\widetilde{E}'}\alpha,\beta)=(-1)^{|\alpha|+1}(\alpha,\bar{\partial}_{\widetilde{E}}\beta).
			\]
			\item \[
			(\bar{\partial}_{\widetilde{E}'}^{*}\alpha,\beta)=(-1)^{|\alpha|}(\alpha,\bar{\partial}_{\widetilde{E}}^{*}\beta).
			\]
		\end{enumerate}
	\end{lem}
	\begin{proof}
		The first assertion follows from Stokes' theorem.

		To prove the second, recall the identity:
		\be\label{identity-p E star}
		\bar{\partial}_{\widetilde{E}}^{*} = -\bar{*}_{E^*} \circ \bar{\partial}_{\widetilde{E}'} \circ \bar{*}_E,
		\ee
		where $ \bar{*}_E : \widetilde{E} \to \widetilde{E}' $ and $ \bar{*}_{E^*} : \widetilde{E}' \to \widetilde{E} $ denote the Hodge star operators, see \cite[Definition 4.1.6]{huybrechts2005complex}. 
		Let $(\cdot,\cdot)_{\widetilde{E}}$ and $(\cdot,\cdot)_{\widetilde{E}'}$ denote the Hermitian inner products on $\Gamma(\widetilde{E})$ and $\Gamma(\widetilde{E}')$ respectively.
		Then for differential forms $\alpha\in\Gamma(\tE'),\beta\in\Gamma(\tE)$,
		\begin{alignat*}{2}
			&\quad( \alpha, \bar{\partial}_{\widetilde{E}}^* \beta )
			= -( \alpha, \bar{*}_{E^*} \bar{\partial}_{\widetilde{E}'} \bar{*}_E \beta )  = -(\alpha, \bar{\partial}_{\widetilde{E}'} \bar{*}_E \beta)_{\widetilde{E}'} 
			\\
			&= -\overline{(\bar{\partial}_{\widetilde{E}'} \bar{*}_E \beta, \alpha)_{\widetilde{E}'}}
			= -\overline{( \bar{\partial}_{\widetilde{E}'} \bar{*}_E \beta, \bar{*}_{E^*} \alpha )}
			\\
			&= (-1)^{|\beta|} \overline{( \bar{*}_E \beta, \bar{\partial}_{\widetilde{E}} \bar{*}_{E^*} \alpha )}
			= (-1)^{|\beta| + |\alpha| + 1} \overline{( \bar{*}_E \beta, 
				\bar{*}_{E^*} \bar{*}_E \bar{\partial}_{\widetilde{E}} \bar{*}_{E^*} \alpha )}
			\\
			&= (-1)^{|\beta| + |\alpha|} \overline{( \bar{*}_E \beta, 
				\bar{*}_{E^*} \bar{\partial}_{\widetilde{E}'}^* \alpha )}
			= (-1)^{|\beta| + |\alpha|} (\bar{\partial}_{\widetilde{E}'}^* \alpha, 
			\bar{*}_E \beta)_{\widetilde{E}} \\
			&= (-1)^{|\beta| + |\alpha|} ( \bar{\partial}_{\widetilde{E}'}^* \alpha, 
			\bar{*}_{E^*} \bar{*}_E \beta )
			= (-1)^{|\alpha|} ( \bar{\partial}_{\widetilde{E}'}^* \alpha, \beta ).
		\end{alignat*}
		This completes the proof.
	\end{proof}
	
	Consider the linear space of continuous linear functionals on $\Gamma(\widetilde{E}')$, denoted by $\mathcal{D}(\widetilde{E})$. Using the pairing \eqref{natural pairing}, we have an embedding $\Gamma(\widetilde{E}) \rightarrow \mathcal{D}(\widetilde{E})$. We call $\mathcal{D}(\widetilde{E})$ the space of distribution-valued sections on $\widetilde{E}$. 
	Inspired by Lemma \ref{dual operator properties}, we extend the action of $\bar{\partial}_{\widetilde{E}}$ and $\bar{\partial}_{\widetilde{E}}^*$ to $\mathcal{D}(\widetilde{E})$ by setting, for $\alpha \in \Gamma(\widetilde{E}')$ and $\beta \in \mathcal{D}(\widetilde{E})$,
	\[
	\begin{cases}
		( \alpha, \bar{\partial}_{\widetilde{E}} \beta ) := (-1)^{|\alpha| + 1} ( \bar{\partial}_{\widetilde{E}'} \alpha, \beta ), \\
		( \alpha, \bar{\partial}_{\widetilde{E}}^* \beta ) := (-1)^{|\alpha|} ( \bar{\partial}_{\widetilde{E}'}^* \alpha, \beta ).
	\end{cases}
	\]
	
	By Hodge theory, we have the decomposition
	\[
	\Gamma(\widetilde{E}) \cong \mathrm{Ker}(\Delta_{\widetilde{E}}) \oplus \mathrm{Im}(\Delta_{\widetilde{E}}).
	\]
	Let $\mathcal{H}'$ be the harmonic projection operator, defined as the composition
	\[
	\Gamma(\widetilde{E}) \rightarrow \mathrm{Ker}(\Delta_{\widetilde{E}}) \rightarrow \Gamma(\widetilde{E}).
	\]
	Consider the distribution-valued section $\mathcal{H}''$ of $\widetilde{E} \boxtimes \widetilde{E}'$ given by
	\[
	( \alpha \boxtimes \beta, \mathcal{H}'' ) = ( \alpha, \mathcal{H}' \beta ),
	\]
	where $\alpha \in \Gamma(\widetilde{E}')$ and $\beta \in \Gamma(\widetilde{E})$. Here $\boxtimes$ is the exterior tensor product of vector bundles.
	
	\begin{lem}
		The distribution-valued section $\mathcal{H}''$ turns out to be a smooth section of $\widetilde{E} \boxtimes \widetilde{E}'$.
	\end{lem}
	\begin{proof}
		We note $\mathcal{H}''$ satisfies the following equations
		\[
		\begin{cases}
			(\bar{\partial}_{\widetilde{E}}\otimes \mathrm{id}+\mathrm{id}\otimes\bar{\partial}_{\widetilde{E}'})\mathcal{H}''=0\\
			(\bar{\partial}_{\widetilde{E}}^{*}\otimes \mathrm{id}+\mathrm{id}\otimes\bar{\partial}_{\widetilde{E}'}^{*})\mathcal{H}''=0.
		\end{cases}
		\]
		Therefore, $\mathcal{H}''$ is a harmonic section of $\widetilde{E}\boxtimes\widetilde{E}'$. Since the Laplacian operator is an elliptic operator, we conclude that $\mathcal{H}''$ is a smooth section of $\widetilde{E}\boxtimes\widetilde{E}'$ by the regularity of elliptic differential equations (see \cite{hormander1983analysis}).
	\end{proof}
	Now, we assume that there is a non-degenerate holomorphic bundle map
	\[
	\omega : E \otimes E \to K_M.
	\]
	The non-degeneracy induces two isomorphisms:
	\[
	\begin{cases}
		\omega_{L}:\alpha\in \Gamma(E) \rightarrow\omega (\alpha\otimes-)\in \Gamma(E^{*}\otimes K_{M}) \\
		\omega_{R}:\alpha \in \Gamma(E)\rightarrow\omega (-\otimes\alpha)\in \Gamma(E^{*}\otimes K_{M}).
	\end{cases}
	\]
	We denote the inverses of the induced isomorphisms by $\omega^{-1}_{L}$ and $\omega^{-1}_{R}$.
	\begin{defn}\label{harmonic delta distribution}
		Let $M$ be a closed K\"ahler manifold, $E \to M$ be a Hermitian holomorphic vector bundle,  $
		\omega : E \otimes E \to K_M,
		$ be a non-degenerate bundle map. The \textbf{Harmonic projection section} $\mathcal{H}\in\Gamma(\widetilde{E}\boxtimes\widetilde{E})$ is defined by
		\[
		\mathcal{H}:=(\mathrm{id}\otimes \omega^{-1}_{L})\mathcal{H}''.
		\] 
		The \textbf{delta distribution section} $\delta\in \mathcal{D}(\widetilde{E}\boxtimes\widetilde{E})$ is defined by
		\[
		\delta(\alpha\boxtimes\beta)=\big(\alpha,\omega^{-1}_{R}(\beta)\big),
		\]
		where $\alpha,\beta\in\Gamma(\widetilde{E}')$.
	\end{defn}

	Now, we can define propagators.
	\begin{defn}\label{ordinary propagator}
		Let $(M,E,\omega)$ be a triple as in \Cref{harmonic delta distribution}. A \textbf{propagator} is a distributional section of $\widetilde{E}\boxtimes\widetilde{E}$, such that 
		$$
		\left\{\begin{array}{l}
			\left(\bar{\partial}_{\widetilde{E}} \otimes \operatorname{id}+\operatorname{id} \otimes \bar{\partial}_{\widetilde{E}}\right) P=\delta-\mathcal{H} \\
			P\in \mathrm{Im}(\bar{\partial}_{\widetilde{E}}^* \otimes \mathrm{id}) . \\
		\end{array}\right.
		$$
	\end{defn}

\begin{rem}\label{BCOV's pairing}
   In the framework of BCOV theory, \( E = \wedge^{\bullet} T^{1,0}M \). The pairing \(\omega\) is   
\[
\omega(\alpha, \beta) = \operatorname{tr}(\alpha \wedge \beta), \qquad 
\alpha, \beta \in \Gamma(M, \wedge^{\bullet} T^{1,0}M),
\]
where \(\operatorname{tr}\) denotes the trace map introduced in \eqref{defn of trace}.

\end{rem}

	The propagator exists and is unique. Let's first prove the uniqueness.
	\begin{prop}\label{uniqueness of propagator}
		Let $(M,E,\omega)$ be a triple as in \Cref{harmonic delta distribution}. If $P\in \mathcal{D}(\widetilde{E}\boxtimes\widetilde{E})$ and $P'\in \mathcal{D}(\widetilde{E}\boxtimes\widetilde{E})$ are propagators, then $P=P'$.
	\end{prop}
	\begin{proof}
		To prove $P = P'$, it suffices to show that, for any $\alpha, \beta \in \Gamma(\widetilde{E}')$,
		\[
		( P - P',\, \alpha \boxtimes \beta ) = 0.
		\]
		We begin by observing that
		\begin{equation}
			\left\{
			\begin{aligned}
				&\big(\bar{\partial}_{\widetilde{E}} \otimes \operatorname{id} + \operatorname{id} \otimes \bar{\partial}_{\widetilde{E}} \big)\big(P - P'\big) = 0, \\
				&P - P' \in \operatorname{Im}(\bar{\partial}_{\widetilde{E}}^* \otimes \operatorname{id}).
			\end{aligned}
			\right.
		\end{equation}
		Hence, there exists a distributional section $G \in \mathcal{D}(\widetilde{E} \boxtimes \widetilde{E})$ such that
		\[
		P - P' = (\bar{\partial}_{\widetilde{E}}^* \otimes \mathrm{id}) G.
		\]
		
		By Hodge decomposition on $\Gamma(\widetilde{E}')$, we may write
		\[
		\alpha = \alpha_H + (\bar{\partial}_{\widetilde{E}'} \bar{\partial}_{\widetilde{E}'}^* + \bar{\partial}_{\widetilde{E}'}^* \bar{\partial}_{\widetilde{E}'}) \alpha',
		\]
		where $\alpha_H$ is harmonic. Then we compute:
		\begin{align*}
			&\quad( P - P',\, \alpha \boxtimes \beta ) 
			= \big( P - P',\, (\alpha_H + \bar{\partial}_{\widetilde{E}'}^* \bar{\partial}_{\widetilde{E}'} \alpha') \boxtimes \beta \big) 
			+ \big( P - P',\, (\bar{\partial}_{\widetilde{E}'} \bar{\partial}_{\widetilde{E}'}^* \alpha') \boxtimes \beta \big)\\
			&= \big( (\bar{\partial}_{\widetilde{E}}^* \otimes \mathrm{id}) G,\, (\alpha_H + \bar{\partial}_{\widetilde{E}'}^* \bar{\partial}_{\widetilde{E}'} \alpha') \boxtimes \beta \big) 
			\pm \big( (\bar{\partial}_{\widetilde{E}} \otimes \mathrm{id})(P - P'),\, (\bar{\partial}_{\widetilde{E}'}^* \alpha') \boxtimes \beta \big) \\
			&= \pm \Big( G,\, \big(\bar{\partial}_{\widetilde{E}'}^* (\alpha_H + \bar{\partial}_{\widetilde{E}'}^* \bar{\partial}_{\widetilde{E}'} \alpha') \boxtimes \beta\big) \Big) 
			\pm \big( (\mathrm{id} \otimes \bar{\partial}_{\widetilde{E}})(P - P'),\, (\bar{\partial}_{\widetilde{E}'}^* \alpha') \boxtimes \beta \big) \\
			&= 0 \pm \big( (\bar{\partial}_{\widetilde{E}}^* \otimes \mathrm{id})(\mathrm{id} \otimes \bar{\partial}_{\widetilde{E}}) G,\, (\bar{\partial}_{\widetilde{E}'}^* \alpha') \boxtimes \beta \big)=\pm\big( (\mathrm{id} \otimes \bar{\partial}_{\widetilde{E}}) G,\, (\bar{\partial}_{\widetilde{E}'}^* \bar{\partial}_{\widetilde{E}'}^* \alpha') \boxtimes \beta \big) = 0.
		\end{align*}
		
		This completes the proof.
		
	\end{proof}
	
	Now, we construct the propagator using the heat kernel.
	
	\begin{lem}
		Let $(M,E,\omega)$ be a triple as in \Cref{harmonic delta distribution}. Then there exists a unique solution $H_t \in C^{\infty}((0,\infty); \Gamma(\widetilde{E} \boxtimes \widetilde{E}))$ satisfying:
		\begin{equation} \label{heateq1}
			\begin{cases}
				\displaystyle \p_t H_t = -(\Delta_{\widetilde{E}} \otimes \mathrm{id}) H_t, \\
				\displaystyle \lim_{t \to 0} H_t = \delta.
			\end{cases}
		\end{equation}
		Moreover, we have $\displaystyle \lim_{t \to \infty} H_t = \mathcal{H}$ and
		\be\label{vanish25}
		(\bar{\partial}_{\widetilde{E}} \otimes \mathrm{id} + \mathrm{id} \otimes \bar{\partial}_{\widetilde{E}}) H_t = 0.
		\ee
		Here $\delta$ and $\mathcal{H}$ are induced by the pairing $\o$ as in \Cref{harmonic delta distribution}.
	\end{lem}
	
	\begin{proof}
		The existence and uniqueness of the heat kernel follow from \cite{berline2003heat}. The convergence $\displaystyle \lim_{t \to \infty} H_t = \mathcal{H}$ could be proved as in \cite[Section 2.6]{berline2003heat}.
		
		To show that 
		\[
		(\bar{\partial}_{\widetilde{E}} \otimes \mathrm{id} + \mathrm{id} \otimes \bar{\partial}_{\widetilde{E}}) H_t = 0,
		\]
		we note that 
		\[
		(\bar{\partial}_{\widetilde{E}} \otimes \mathrm{id} + \mathrm{id} \otimes \bar{\partial}_{\widetilde{E}}) H_t
		\]
		satisfies the heat equation:
		\[
		\begin{cases}
			\displaystyle \p_t \left( (\bar{\partial}_{\widetilde{E}} \otimes \mathrm{id} + \mathrm{id} \otimes \bar{\partial}_{\widetilde{E}}) H_t \right) 
			= -(\Delta_{\widetilde{E}} \otimes \mathrm{id}) \left( (\bar{\partial}_{\widetilde{E}} \otimes \mathrm{id} + \mathrm{id} \otimes \bar{\partial}_{\widetilde{E}}) H_t \right), \\
			\displaystyle \lim_{t \to 0} (\bar{\partial}_{\widetilde{E}} \otimes \mathrm{id} + \mathrm{id} \otimes \bar{\partial}_{\widetilde{E}}) H_t = 0.
		\end{cases}
		\]
		By uniqueness of the solution to the heat equation, it follows that
		\[
		(\bar{\partial}_{\widetilde{E}} \otimes \mathrm{id} + \mathrm{id} \otimes \bar{\partial}_{\widetilde{E}}) H_t = 0.
		\]
	\end{proof}
	
	The following proposition proves the existence of the propagator:
	
	\begin{prop}\label{existence of propagator}
		Let $(M,E,\omega)$ be a triple as in \Cref{harmonic delta distribution}. Then
		\[
		P := \int_{0}^{\infty} dt \wedge \left( \bar{\partial}_{\widetilde{E}}^* \otimes \mathrm{id} \right) H_t \in \mathcal{D}(\widetilde{E} \boxtimes \widetilde{E})
		\]
		is a propagator.
		
		Moreover, $P|_{M \times M \setminus \triangle}$ is smooth, where $\triangle := \{ (p, p) \in M \times M : p \in M \}$.
	\end{prop}
	
	\begin{proof}
		The proof is similar to that of \cite[Theorem 2.38]{berline2003heat}, and we outline the main steps.
		
		For $\epsilon, L \in (0, \infty)$, set
		\[
		G_{\epsilon,L} := \int_{\epsilon}^{L} dt \wedge H_t \in C^{\infty} \big( (0, \infty); \Gamma(\widetilde{E} \boxtimes \widetilde{E}) \big).
		\]
		Then by \eqref{vanish25}:
		\begin{equation}\label{regularized Green's function}
			\begin{cases}
				(\Delta_{\widetilde{E}} \otimes \mathrm{id}) G_{\epsilon,L}
				= \int_{\epsilon}^{L} dt \wedge (\Delta_{\widetilde{E}} \otimes \mathrm{id}) H_t 
				= H_\epsilon - H_L, \\
				(\bar{\partial}_{\widetilde{E}} \otimes \mathrm{id} + \mathrm{id} \otimes \bar{\partial}_{\widetilde{E}}) G_{\epsilon,L}
				= \int_{\epsilon}^{L} dt \wedge (\bar{\partial}_{\widetilde{E}} \otimes \mathrm{id} + \mathrm{id} \otimes \bar{\partial}_{\widetilde{E}}) H_t 
				= 0.
			\end{cases}
		\end{equation}
		Applying $(\ref{regularized Green's function})$, we obtain
		\[
		\begin{aligned}
			& (\bar{\partial}_{\widetilde{E}} \otimes \mathrm{id} + \mathrm{id} \otimes \bar{\partial}_{\widetilde{E}}) 
			\int_{\epsilon}^{L} dt \wedge \left( \bar{\partial}_{\widetilde{E}}^* \otimes \mathrm{id} \right) H_t \\
			&= (\bar{\partial}_{\widetilde{E}} \otimes \mathrm{id} + \mathrm{id} \otimes \bar{\partial}_{\widetilde{E}}) 
			\left( \bar{\partial}_{\widetilde{E}}^* \otimes \mathrm{id} \right) G_{\epsilon,L} \\
			&= (\Delta_{\widetilde{E}} \otimes \mathrm{id}) G_{\epsilon,L} = H_\epsilon - H_L.
		\end{aligned}
		\]
		Moreover,
		\[
		\int_{\epsilon}^{L} dt \wedge \left( \bar{\partial}_{\widetilde{E}}^* \otimes \mathrm{id} \right) H_t 
		= \left( \bar{\partial}_{\widetilde{E}}^* \otimes \mathrm{id} \right) G_{\epsilon,L}
		\in \mathrm{Im} \left( \bar{\partial}_{\widetilde{E}}^* \otimes \mathrm{id} \right).
		\]
		One can show that  
		\[
		\lim_{\substack{\epsilon \to 0 \\ L \to \infty}} \int_{\epsilon}^{L} dt \wedge \left( \bar{\partial}_{\widetilde{E}}^* \otimes \mathrm{id} \right) H_t
		\]  
		exists (this is the nontrivial part of the argument), and
		since $\bar{\partial}_{\widetilde{E}}^* \otimes \mathrm{id}$ is a closed operator, we obtain:
		\[
		\begin{cases}
			(\bar{\partial}_{\widetilde{E}} \otimes \mathrm{id} + \mathrm{id} \otimes \bar{\partial}_{\widetilde{E}}) P 
			= \lim_{\epsilon \to 0} H_\epsilon - \lim_{L \to \infty} H_L 
			= \delta - \mathcal{H}, \\
			P \in \mathrm{Im}(\bar{\partial}_{\widetilde{E}}^* \otimes \mathrm{id}).
		\end{cases}
		\]
		Hence $P$ is a propagator.
		
		To prove $P|_{M \times M \setminus \triangle}$ is smooth, observe:
		\[
		(\Delta_{\widetilde{E}} \otimes \mathrm{id}) \lim_{\epsilon \to 0} G_{\epsilon,L} = \delta - H_L.
		\]
		Since the singularities of the right-hand side are supported on the diagonal $\triangle$, it follows from elliptic regularity (cf. \cite{hormander1983analysis}) that 
		$\lim_{\epsilon \to 0} G_{\epsilon,L}$ is smooth off $\triangle$. Hence, so is
		\[
		\int_0^L dt \wedge \left( \bar{\partial}_{\widetilde{E}}^* \otimes \mathrm{id} \right) H_t.
		\]
		
		By \cite[Proposition 2.37]{berline2003heat}, the integral 
		\[
		\int_L^{\infty} dt \wedge \left( \bar{\partial}_{\widetilde{E}}^* \otimes \mathrm{id} \right) H_t
		\]
		is smooth globally. Therefore, \[
		P = \int_0^L dt \wedge \left( \bar{\partial}_{\widetilde{E}}^* \otimes \mathrm{id} \right) H_t 
		+ \int_L^{\infty} dt \wedge \left( \bar{\partial}_{\widetilde{E}}^* \otimes \mathrm{id} \right) H_t
		\] is smooth away from the diagonal.
		
	\end{proof}

	The mildness of the singularities of $P$ is crucial for defining Feynman graph integrals using the Cauchy principal value. In fact, the singularities of $P$ are fully determined by the singularities of the heat kernel $H_t$ at $t = 0$:
	
	Let $\epsilon \in (0,\infty)$, and choose a cutoff function $\eta \in C_c^\infty(\mathbb{R})$ such that $\eta(t) \equiv 1$ for $|t| \leq \epsilon$ and $\eta(t) = 0$ for $|t| \geq 2\epsilon$.
	
	\begin{prop}\label{propagator singularities and heat}
		The integral
		\[
		\int_{0}^{\infty} (1 - \eta(t))\, dt \wedge \left( \bar{\partial}_{\widetilde{E}}^* \otimes \mathrm{id} \right) H_t
		\]
		defines a smooth section over $M \times M$. In particular, the singularities of $P$ coincide with those of
		\[
		\int_{0}^{\infty} \eta(t)\, dt \wedge \left( \bar{\partial}_{\widetilde{E}}^* \otimes \mathrm{id} \right) H_t.
		\]
	\end{prop}
	
	\begin{proof}
		This follows directly from \cite[Proposition 2.37]{berline2003heat}.
	\end{proof}
	
	\Cref{propagator singularities and heat} implies that to analyze the singularities of $P$, it suffices to understand the short-time asymptotics of the heat kernel $H_t$ near $t = 0$.
	
	\vspace{0.5em}
	We now introduce a useful notion:
	
	\begin{defn}\label{propagator}
		\textbf{The propagator in Schwinger space} is defined by
		\[
		P_t := -dt \wedge \big( \bar{\partial}_{\widetilde{E}}^* \otimes \mathrm{id} \big) H_t + H_t \in \Omega^*\Big( (0,+\infty); \Gamma(\widetilde{E} \boxtimes \widetilde{E}) \Big),
		\]
		where $\Omega^*\left( (0,+\infty); \Gamma(\widetilde{E} \boxtimes \widetilde{E}) \right)$ denotes the space of $\Gamma(\widetilde{E} \boxtimes \widetilde{E})$-valued smooth forms on $(0,\infty)$.
	\end{defn}
	
	We observe that
	\[
	\int_{0}^{\infty} \eta(t) P_t = -\int_{0}^{\infty} \eta(t)\, dt \wedge \left( \bar{\partial}_{\widetilde{E}}^* \otimes \mathrm{id} \right) H_t.
	\]
	Hence, by \Cref{propagator singularities and heat}, the singularities of $P$ are fully determined by the singularities of $P_t$ at $t = 0$.

	\subsection{Feynman graph integrals}
	In this subsection, we introduce Feynman graph integrals for holomorphic field theories and establish their convergence in the sense of the Cauchy principal value. For the reader’s convenience, relevant definitions and properties of the Cauchy principal value are collected in Appendix \ref{section Cauchy principal value}.
	
	Let $\triangle$ denote the diagonal of $M \times M$. The blow-up of $M \times M$ along the diagonal is denoted by $\mathrm{Bl}_{\triangle}(M \times M)$, with the canonical projection map 
	\[
	p: \mathrm{Bl}_{\triangle}(M \times M) \rightarrow M \times M,
	\]
	such that $p^{-1}(\triangle)$ is the exceptional divisor.
	
	The propagator $P$ can be regarded as a bundle-valued smooth differential form on 
	\[
	M \times M \setminus \triangle \cong \mathrm{Bl}_{\triangle}(M \times M) \setminus p^{-1}(\triangle).
	\]
	The theorem below shows that the propagator has singularities of the desired type.
	\begin{thm}\label{propagator has divisorial type singularities}
		Let $M$ be a closed real analytic K\"ahler manifold, $E \to M$ a real analytic Hermitian holomorphic vector bundle, and let
		\[
		\omega : E \otimes E \to K_M
		\]
		be a non-degenerate bundle map. Then:
		\begin{enumerate}[(1)]
			\item $P$ has divisorial type singularities along $p^{-1}(\triangle)$ (see Definition \ref{divisorial type singularities}). Moreover, all holomorphic derivatives of $P$ also have divisorial type singularities along $p^{-1}(\triangle)$.
			
			\item The pullback
			\[
			\triangle^*P := \lim_{\epsilon \rightarrow 0} \triangle^* \left( \int_{\epsilon}^{\infty} dt \wedge \big( \bar{\partial}_{\widetilde{E}}^* \otimes \mathrm{id} \big) H_t \right)
			\]
			defines a smooth bundle-valued differential form on $M$, where $\triangle^*$ denotes the pullback of bundle-valued differential forms along the diagonal embedding
			\[
			\triangle: M \rightarrow M \times M.
			\]
			Moreover, the pullbacks of all holomorphic derivatives of $P$ along $\triangle$ are also smooth bundle-valued differential forms.
		\end{enumerate}
	\end{thm}
	
	\Cref{propagator has divisorial type singularities} is the most technically involved result in this paper.
	The key ingredients for its proof are presented in \cref{Section 3} and \cref{Section 4}, while the complete proof is given in \cref{proof of technical theorem}. In this section, we assume the validity of \Cref{propagator has divisorial type singularities} and use it to construct Feynman graph integrals.
	
	\begin{rem}
		The analyticity assumption in \Cref{propagator has divisorial type singularities} is essential for our proof. It remains an interesting question whether this theorem holds without it.
	\end{rem}
	
	Before introducing Feynman graph integrals, we recall the Fulton–MacPherson compactification of configuration spaces.
	
	\begin{defn}
		Let $\vec{I}$ be an ordered finite set and $M$ a closed complex manifold. The configuration space of $M$ with respect to $\vec{I}$ is defined by
		\[
		\Conf_{\vec{I}}(M) = \big\{ (p_1, \dots, p_{|\vec{I}|}) \in M^{|\vec{I}|}:\, p_i \neq p_j \text{ for } i \neq j \big\},
		\]
		where $|\vec{I}|$ denotes the number of elements in $\vec{I}$.
	\end{defn}
	
	For any subset $\vec{J} \subset \vec{I}$, define
	\[
	\triangle_{\vec{J} \subset \vec{I}} := \big\{ (p_1, \dots, p_{|\vec{I}|}) \in M^{|\vec{I}|} \,:\, p_i = p_j \text{ for all } i, j \in \vec{J}\, \big\}.
	\]
	Denote the blow-up of $M^{|\vec{I}|}$ along $\triangle_{\vec{J} \subset \vec{I}}$ by $\mathrm{Bl}_{\triangle_{\vec{J} \subset \vec{I}}}(M^{|\vec{I}|})$. We then have a natural map
	\[
	i: \Conf_{\vec{I}}(M) \rightarrow \prod_{\vec{J} \subset \vec{I}} \mathrm{Bl}_{\triangle_{\vec{J} \subset \vec{I}}}(M^{|\vec{I}|}).
	\]
	
	\begin{defn}\label{Fulton-MacPherson}
		The \textbf{Fulton–MacPherson compactification} of $\Conf_{\vec{I}}(M)$ is defined as the closure of $\mathrm{Im}(i)$ in $\prod_{\vec{J} \subset \vec{I}} \mathrm{Bl}_{\triangle_{\vec{J} \subset \vec{I}}}(M^{|\vec{I}|})$. We denote it by $\widetilde{\Conf}_{\vec{I}}(M)$.
	\end{defn}
	
	\begin{prop}
		Let $\vec{I}$ be an ordered finite set. Then $\widetilde{\Conf}_{\vec{I}}(M)$ is a closed complex manifold. Moreover, the complement $\widetilde{\Conf}_{\vec{I}}(M) \setminus \Conf_{\vec{I}}(M)$ is a simple normal crossing divisor.
	\end{prop}
	\begin{proof}
		See \cite{17c23791-d52f-3aa0-a5a9-5abc4d54669e}.
	\end{proof}
	
	\begin{prop}
		Let $\vec{J} \subset \vec{I}$ be an ordered finite set. Then there exists a holomorphic map
		\[
		\pi_{\vec{J} \subset \vec{I}}: \widetilde{\Conf}_{\vec{I}}(M) \rightarrow \widetilde{\Conf}_{\vec{J}}(M)
		\]
		extending the natural forgetful map from $\Conf_{\vec{I}}(M)$ to $\Conf_{\vec{J}}(M)$. Moreover,
		\[
		\pi^{-1}_{\vec{J} \subset \vec{I}} \left( \widetilde{\Conf}_{\vec{J}}(M) \setminus \Conf_{\vec{J}}(M) \right) \subset \widetilde{\Conf}_{\vec{I}}(M) \setminus \Conf_{\vec{I}}(M).
		\]
	\end{prop}
	\begin{proof}
		See \cite{17c23791-d52f-3aa0-a5a9-5abc4d54669e}.
	\end{proof}
	
	Now, let's introduce the notion of Feynman graph integrals. 
	
	Let $\{e_1,\ldots,e_r\}$ be a local holomorphic frame of the holomorphic vector bundle $E \to M$, where $r = \mathrm{rank}(E)$.  
	Then any section $\varphi \in \Gamma(E)$ can be locally expressed as  
	\begin{equation}\label{dec-E}
		\varphi = \sum_{j=1}^r \varphi^j e_j,
	\end{equation}  
	where the coefficients $\varphi^j$ are smooth functions defined on some open subset of $M$.
	
	\def\bfj{{\mathbf{j}}}
	
	Let $\mathbf{j} = (j_1, \ldots, j_n) \in \mathbb{Z}_{\geq 0}^n$. Given a local holomorphic coordinate system $(z_1, \ldots, z_n)$ on $M$, we define the holomorphic differential operator  
	\begin{equation}\label{multi-derivative-defn}
		\partial^{\mathbf{j}} := \frac{\partial^{|\mathbf{j}|}}{\partial z_1^{j_1} \cdots \partial z_n^{j_n}}, \quad \text{where } |\mathbf{j}| := j_1 + \cdots + j_n.
	\end{equation}
	
	\begin{defn}
		A \textbf{degree-$k$ holomorphic Lagrangian density} $I$ is a linear map from $\Gamma(E)^{\otimes k}$ to $\Gamma(K_M)$, which is locally of the form  
		\be\label{defn-I}
		\begin{aligned}
			I(\varphi_1, \ldots, \varphi_k) = \sum c_{i_1 \ldots i_k\, \mathbf{j}_1 \ldots \mathbf{j}_k} \, d^n \mathbf{z} \, 
			\wedge\partial^{\mathbf{j}_1}(\varphi_1^{i_1}) \wedge
			\partial^{\mathbf{j}_2}(\varphi_2^{i_2}) \wedge \cdots\wedge 
			\partial^{\mathbf{j}_k}(\varphi_k^{i_k}),
		\end{aligned}
		\ee
		where each $\varphi_l$ is a local section of $E$ expressed as in \eqref{dec-E}, $\varphi_l^{i_l}$ denotes the $i_l$-th component function in this decomposition, $\partial^{\mathbf{j}_l}$ is a holomorphic differential operator of the form \eqref{multi-derivative-defn} with $\mathbf{j}_l \in \mathbb{Z}_{\geq 0}^n$, and the coefficients $c_{i_1 \ldots i_k\, \mathbf{j}_1 \ldots \mathbf{j}_k}$ are holomorphic functions. Lastly, $d^n\mathbf{z}=dz_1\wedge\cdots\wedge dz_n.$
		
		The density $I$ naturally extends to $\Gamma(\widetilde{E}^{\otimes k})$ in a way compatible with \eqref{defn-I}, and we shall continue to denote it by $I$.
		
	\end{defn}
	\begin{rem}
		For a coordinate-free and equivalent definition of holomorphic Lagrangian densities, see \cite[Definition 2.15]{Williams:2018ows}.
	\end{rem}
	
	\begin{defn}
		A directed graph $\vec{\Gamma}$ consists of the following data:
		\begin{enumerate}[(a)]
			\item An ordered set of vertices $\vec{\Gamma}_0$ and an ordered set of edges $\vec{\Gamma}_1$.
			
			\item Two maps between sets
			\[ t, h : \vec{\Gamma}_1 \rightarrow \vec{\Gamma}_0 \]
			which are the assignments of tail and head to each directed edge.
			\item For each vertex $v \in \vec{\Gamma}_0$ of degree $\mathrm{deg}(v)\in\mathbb{Z}_{\geq 0}$, we assign a $\mathrm{deg}(v)$-degree holomorphic Lagrangian density $I_v$.
			
		\end{enumerate}
	\end{defn}
	
	We will use $| \vec{\Gamma}_0 |$ and $| \vec{\Gamma}_1 |$ to denote the number of vertices
	and edges respectively.
	
	Given a directed graph $\vec{\Gamma}$, we have
	\[
	\otimes_{e \in \vec{\Gamma}_1} P \in \mathcal{D}\left((M \times M)^{|\vec{\Gamma}_{1}|};\boxtimes_{e \in \vec{\Gamma}_{1}}(\widetilde{E} \boxtimes \widetilde{E})\right),
	\]
	and
	\[
	\otimes_{v \in \vec{\Gamma}_0} I_v : \Gamma\Big(\prod_{v \in \vec{\Gamma}_0} M^{\deg(v)}; \boxtimes_{v \in \vec{\Gamma}_0} \widetilde{E}^{\boxtimes \deg(v)}\Big) \to \Gamma\left(\boxtimes_{v \in \vec{\Gamma}_0}(K_M \otimes \Lambda^\bullet T_{0,1}^* M)\right).
	\]
	There exist natural isomorphisms
	\[
	\tau^{\vec{\Gamma}} : (M \times M)^{|\vec{\Gamma}_1|} \cong \prod_{v \in \vec{\Gamma}_0} M^{\deg(v)},
	\]
	and
	\[
	\tau^{\vec{\Gamma}}_* : \mathcal{D}\Big((M \times M)^{|\vec{\Gamma}_{1}|};\boxtimes_{e \in \vec{\Gamma}_1} \widetilde{E} \boxtimes \widetilde{E}\Big) \cong \mathcal{D}\Big(\prod_{v \in \vec{\Gamma}_0} M^{\deg(v)};\boxtimes_{v \in \vec{\Gamma}_0} \widetilde{E}^{\boxtimes \deg(v)}\Big),
	\]
	which are determined by the structure of the directed graph $\vec{\Gamma}$. Hence, $\tau^{\vec{\Gamma}}_*\left(\otimes_{e \in \vec{\Gamma}_1} P\right)$ is an element of
	\[
	\mathcal{D}\Big(\prod_{v \in \vec{\Gamma}_0} M^{\deg(v)}; \boxtimes_{v \in \vec{\Gamma}_0} \widetilde{E}^{\boxtimes \deg(v)}\Big),
	\]
	with singularities along the diagonals. 
	Therefore, by Theorem~\ref{propagator has divisorial type singularities},
	\[
	\left((\otimes_{v \in \vec{\Gamma}_0} I_v) \circ \tau^{\vec{\Gamma}}_* \right)\left(\otimes_{e \in \vec{\Gamma}_1} P\right)
	\]
	is well-defined as a differential form on $\Conf_{\vec{\Gamma}_0}(M)$. We denote it by
	\[
	\left(\otimes_{v \in \vec{\Gamma}_0} I_v, \otimes_{e \in \vec{\Gamma}_1} P\right)_{\vec{\Gamma}} \in \Gamma\left(\Conf_{\vec{\Gamma}_0}(M); \boxtimes_{v \in \vec{\Gamma}_0}(K_M \otimes \Lambda^\bullet T_{0,1}^* M)\right).
	\]
	
	\begin{thm}\label{Thm-Fey-Graph-Int}
		Let $(M,E,\omega)$ be a triple as in \Cref{propagator has divisorial type singularities}.
		Let $\vec{\Gamma}$ be a directed graph. Then
		\[
		\left(\otimes_{v \in \vec{\Gamma}_0} I_v, \otimes_{e \in \vec{\Gamma}_1} P\right)_{\vec{\Gamma}}
		\]
		is a differential form on $\widetilde{\Conf}_{\vec{\Gamma}_0}(M)$ with divisorial type singularities along $\widetilde{\Conf}_{\vec{\Gamma}_0}(M) \setminus \Conf_{\vec{\Gamma}_0}(M)$ (see Definition~\ref{divisorial type singularities}). In particular,
		\[
		\dashint_{\widetilde{\Conf}_{\vec{\Gamma}_0}(M)} \left(\otimes_{v \in \vec{\Gamma}_0} I_v, \otimes_{e \in \vec{\Gamma}_1} P\right)_{\vec{\Gamma}}
		\]
		is well-defined.
	\end{thm}
	
	\begin{proof}\label{finiteness of Feynman graph integrals}
		
		By Theorem~\ref{propagator has divisorial type singularities} and Proposition~\ref{pull back of divisorial type singularities}, we conclude that
		\[
		\left(\otimes_{v \in \vec{\Gamma}_0} I_v, \otimes_{e \in \vec{\Gamma}_1} P\right)_{\vec{\Gamma}}
		\]
		has divisorial type singularities. By the construction described in Appendix~\ref{section Cauchy principal value}, the integral is well-defined in the sense of Cauchy principal value.
	\end{proof}

	\begin{defn}\label{Fey-Graph-Int}
		Let $(M,E,\omega)$ be triples as in \Cref{propagator has divisorial type singularities}.
		Let $\vec{\Gamma}$ be a directed graph. The {\textbf{Feynman graph integral}} on
		$M$ is defined by:
		\[
		\dashint_{\widetilde{\Conf}_{\vec{\Gamma}_{0}}(M)}\left(\otimes_{v \in \vec{\Gamma}_0} I_v, \otimes_{e \in \vec{\Gamma}_1} P\right)_{\vec{\Gamma}}.
		\]
	\end{defn}
	\begin{rem}
		When $M$ is an elliptic curve and $E$ is the trivial vector bundle, the Feynman graph integrals can be computed explicitly, see \cite{yang2024feynman}.
	\end{rem}
	In the remainder of this subsection, we introduce an alternative formulation of \Cref{Thm-Fey-Graph-Int}.
	
	\begin{defn}
		Let $M$ be a closed real analytic K\"ahler manifold. A \textbf{fake distance function} $\tilde{\rho}:M\times M\rightarrow\mathbb{R}$ is a non-negative function satisfying:
		\begin{enumerate}
			\item $\tilde{\rho}^{2}$ is smooth and $\tilde{\rho}^{-1}(0)=\triangle$, where $\triangle:=\{(p,p)\in M\times M : p\in M\}$.
			\item There exists an open neighborhood $U\subset M\times M$ of $\triangle$ such that \[\tilde{\rho}^{2}|_{U}=\rho^{2}|_{U},\]
			where $\rho$ is the distance function on $M$.
		\end{enumerate}
	\end{defn}
	
	\begin{lem}\label{fake distance function is defining function}
		Let $M$ be a closed real analytic K\"ahler manifold, and let $\tilde{\rho}$ be a fake distance function. Then 
		\[
		p^{*}\tilde{\rho}^{2}:\mathrm{Bl}_{\triangle}(M \times M)\rightarrow\mathbb{R}
		\]
		is a smooth defining function (see \ref{continuous defining function}) for the divisor $p^{-1}(\triangle)$, where $p:\mathrm{Bl}_{\triangle}(M \times M)\rightarrow M\times M$ is the canonical blow-up map along $\triangle:=\{(p,p)\in M\times M : p\in M\}$.
	\end{lem}
	
	\begin{proof}
		Since being a smooth defining function is a local property, we may reduce to a local neighborhood in $\C^n$. It suffices to prove the statement over an open set $U \subset \C^n \times \C^n$ containing $(0,0)$, where we may assume $\tilde{\rho} = \rho$.
		
		We identify $\mathrm{Bl}_{\triangle}(\C^n \times \C^n)$ with the closure of the image of the natural embedding
		\[
		i:\C^n \times \C^n \setminus \triangle \hookrightarrow \left( \C^n \times \C^n \right) \times \left( \C^n \times \C^n \setminus \triangle \right)/\C^*, \quad (\mathbf{z},\mathbf{w}) \mapsto (\mathbf{z},\mathbf{w},\mathbf{z},\mathbf{w}),
		\]
		where the $\C^*$-action is given by $\lambda \cdot (\mathbf{z},\mathbf{w}) = (\lambda \mathbf{z}, \lambda \mathbf{w})$ for $\lambda \in \C^*$.
		
		\def\P{\mathbb{P}}
		This yields an identification of $(\C^n \times \C^n \setminus \triangle)/\C^*$ with $\C^n \times \P^{n-1}$ via
		\[
		(\mathbf{z},\mathbf{w}) \mapsto (\mathbf{w},\; [z_1 - w_1 : \cdots : z_n - w_n]).
		\]
		Under this identification, we can write
		\begin{align*}
			&\quad\mathrm{Bl}_{\triangle}(\C^n \times \C^n)
			\\&\cong \{(\mathbf{z},\mathbf{w},[\lambda_1 : \cdots : \lambda_n]) \in \C^n \times \C^n \times \P^{n-1} :\mathbf{z} - \mathbf{w} = k \cdot \boldsymbol{\lambda} \text{ for some } k \in \C\},
		\end{align*}
		where $\boldsymbol{\lambda} = (\lambda_1,\dots,\lambda_n)\in\C^n$.
		
		Let $\{U_i\}_{i=1}^n$ be the standard affine cover of $\P^{n-1}$, corresponding to $\lambda_i \neq 0$. Over $U_1$, we can use local coordinates
		\[
		(z_1 - w_1, \lambda_2, \ldots, \lambda_n, w_1, \ldots, w_n).
		\]
		
		Using the expansion from~\eqref{distance function on arbitrary}, we compute:
		\begin{align}
			\frac{\rho^2(\mathbf{z},\bar{\mathbf{z}},\mathbf{w},\bar{\mathbf{w}})}{|z_1 - w_1|^2}
			&= 1 + \sum_{i=2}^{n} \lambda_i \bar{\lambda}_i \label{first term distance} \\
			&\quad + \Big(f_{1\bar{1}} + \sum_{i=2}^{n} \lambda_i f_{i\bar{1}} + \sum_{i=2}^{n} \bar{\lambda}_i f_{1\bar{i}} + \sum_{i,j=2}^{n} \lambda_i \bar{\lambda}_j f_{i\bar{j}}\Big), \label{second term distance}
		\end{align}
		where $f_{i\bar{j}}(0,0,0,0) = 0$.
		
		By shrinking $U$ if necessary, we may ensure $|f_{i\bar{j}}| < \min\left\{\frac{1}{3}, \frac{1}{3(n-1)}, \frac{1}{3(n-1)^2} \right\}$. Then the absolute value of the term~\eqref{first term distance} dominates that of~\eqref{second term distance}, so the full expression is a positive smooth function on $p^{-1}(U) \cap U_1$. Thus, $\rho^2$ is a smooth defining function on this chart. A similar argument works on other charts $U_i$, completing the proof.
	\end{proof}
	
	\begin{thm}\label{defining function and Feynman graph integral}
		Let $(M,E,\omega)$ be triples as in \Cref{propagator has divisorial type singularities}.
		Let $\vec{\Gamma}$ be a directed graph, and $\tilde{\rho}$ a fake distance function. Then we have
		\[
		\lim_{\epsilon\rightarrow0}\int_{\prod_{i<j\in \vec{\Gamma}_{0}}\tilde{\rho}^{2}_{ij}>\epsilon}\left(\otimes_{v \in \vec{\Gamma}_0} I_v, \otimes_{e \in \vec{\Gamma}_1} P\right)_{\vec{\Gamma}}=\dashint_{\widetilde{\Conf}_{\vec{\Gamma}_{0}}(M)}\left(\otimes_{v \in \vec{\Gamma}_0} I_v, \otimes_{e \in \vec{\Gamma}_1} P\right)_{\vec{\Gamma}},
		\]
		where $\tilde{\rho}^{2}_{ij}=\pi_{\{i<j\}\subset\vec{\Gamma}_{0}}^{*}(\tilde{\rho}^{2})$.
	\end{thm}
	
	\begin{proof}
		Consider the following natural map:
		\[
		\prod_{i<j\in\vec{\Gamma}_{0}}\pi_{\{i<j\}\subset\vec{\Gamma}_{0}}:\widetilde{\Conf}_{\vec{\Gamma}_{0}}(M)\rightarrow \prod_{i<j\in\vec{\Gamma}_{0}}\widetilde{\Conf}_{\{i<j\}}(M).
		\]
		We note that $\widetilde{\Conf}_{\{i<j\}}(M)\cong \mathrm{Bl}_{\triangle}(M\times M)$, so the function $\prod_{i<j\in \vec{\Gamma}_{0}}\tilde{\rho}^{2}_{ij}$ is a smooth defining function of the divisor
		\[
		\bigcup_{i<j\in \vec{\Gamma}_{0}}p_{ij}^{-1}(\triangle)\subset \prod_{i<j\in\vec{\Gamma}_{0}}\widetilde{\Conf}_{\{i<j\}}(M),
		\]
		where $p_{ij}$ is the composition
		\[
		\prod_{i'<j'\in\vec{\Gamma}_{0}}\widetilde{\Conf}_{\{i'<j'\}}(M)\rightarrow \widetilde{\Conf}_{\{i<j\}}(M)\rightarrow M\times M.
		\]
		Since
		\[
		\Big(\prod_{i<j\in\vec{\Gamma}_{0}}\pi_{\{i<j\}\subset\vec{\Gamma}_{0}}\Big)^{-1}\Big(\bigcup_{i<j\in \vec{\Gamma}_{0}}p_{ij}^{-1}(\triangle)\Big)=\widetilde{\Conf}_{\vec{\Gamma}_{0}}(M)\setminus \Conf_{\vec{\Gamma}_{0}}(M),
		\]
		it follows from Proposition \ref{pull back of defining function} that $\prod_{i<j\in \vec{\Gamma}_{0}}\tilde{\rho}^{2}_{ij}$ is a smooth defining function of 
		\[
		\widetilde{\Conf}_{\vec{\Gamma}_{0}}(M)\setminus \Conf_{\vec{\Gamma}_{0}}(M).
		\]
		The theorem then follows from Proposition \ref{prop-a18}.
	\end{proof}
	
	\begin{rem}
		Theorem~\ref{defining function and Feynman graph integral} provides a definition of Feynman graph integrals that does not rely on the Fulton–MacPherson compactification. This shows that the notion of Feynman graph integral introduced in this paper is intrinsic.
		
	\end{rem}
	
	\def\bfz{{\mathbf{z}}}
	\def\bfw{{\mathbf{w}}}
	\section{Regular Expressions}\label{Section 3}
	
	In this section, we introduce a structure we call regular expressions, which serves as an analog of divisorial type singularities in the context of heat kernels. More precisely, in \Cref{subsection regular expression}, we define regular expressions and establish their basic properties. In \Cref{subsection regular expressions and divisorial type singularities}, we show the connection between regular expressions and divisorial type singularities.

	\subsection{Regular expressions and its filtration}\label{subsection regular expression}
	
    	\begin{defn}
    		Let $U\subset\mathbb{C}^{n}$ be a open subset, $E$ is a trivial holomorphic vector bundle on $\mathbb{C}^{n}$, and $\widetilde{E}=E\otimes \Lambda^\bullet T_{0,1}^*\mathbb{C}^{n}$. The linear space of \textbf{regular expressions} $S_{\sm}(U\times U)\subset\Omega^{*}\Big((0,\infty);\Gamma\big((\widetilde{E}\boxtimes \widetilde{E})|_{U\times U}\big)\Big)$ is generated by expressions of the form
    		\[
    		a(\mathbf{y}, d\mathbf{y})\, b,
    		\]
    		where $a(\mathbf{y}, d\mathbf{y})$ is a polynomial in the variables
    		\[
    		\mathbf{y} = (y_1, y_2, \ldots, y_n), \quad \text{and} \quad d\mathbf{y} = (dy_1, dy_2, \ldots, dy_n),
    		\]
    		with
    		\[
    		y_i = \frac{\overline{z_i} - \overline{w_i}}{t}.
    		\]
    		Here, $d$ denotes the de Rham differential on $(0,\infty) \times U \times U$, so
    		\[
    		dy_i = d\left( \frac{\overline{z_i} - \overline{w_i}}{t} \right) = \frac{d(\overline{z_i} - \overline{w_i})}{t} - \frac{(\overline{z_i} - \overline{w_i})\, dt}{t^2}.
    		\]
    		The factor $b \in \Omega^*\Big((0,\infty); \Gamma\big((\widetilde{E} \boxtimes \widetilde{E})|_{U \times U}\big)\Big)$ is assumed to extend smoothly to a differential form in $\Omega^*\Big([0,+\infty); \Gamma\big((\widetilde{E} \boxtimes \widetilde{E})|_{U \times U}\big)\Big)$.
    	\end{defn}

	\begin{prop}\label{indep-holo-coord}
		Let $f:U\subset\mathbb{C}^{n}\rightarrow W\subset \mathbb{C}^{n}$ be a biholomorphism, then $f^{*}(S_{\sm}(W\times W))=S_{\sm}(U\times U)$.
	\end{prop}
	\begin{proof}
		Let $(\mathbf{z},\mathbf{w})\in U\times U$, we have
		\[
		f_i(\mathbf{z}) = f_i(\mathbf{w}) + \sum_{j=1}^n f_{ij}(\mathbf{z}, \mathbf{w}) (z_j - w_j),
		\]
		for some holomorphic functions $f_{ij}$.
		
		Now let $y_i' := \frac{\bar{f}_i(\mathbf{z}) - \bar{f}_i(\mathbf{w}) }{t}$. Then 
		\[
		y_i' 
		= \frac{\bar{f}_i(\mathbf{z}) - \bar{f}_i(\mathbf{w})}{t}= \sum_j \overline{f_{ij}( \mathbf{z},\mathbf{w})} \cdot \frac{\bar{z}_j - \bar{w}_j}{t}
		= \sum_j \overline{f_{ij}( \mathbf{z},\mathbf{w})} \cdot y_j \in S_{\sm}(U\times U)
		\]
		and
		\[
		dy_i' = \sum_j \left( d\overline{f_{ij}(\mathbf{z} , \mathbf{w})} \right) y_j
		+ \overline{f_{ij}(\mathbf{z} ,\mathbf{w})} \, dy_j \in S_{\sm}(U\times U),
		\]
		where we recall that $d$ denotes the de Rham differential on $(0,\infty) \times M \times M$. This proves our claim.
	\end{proof}
	\begin{rem}
		Proposition \ref{indep-holo-coord} shows that the definition of regular expressions is independent of the choice of coordinates. This allows us to define the concept globally on an complex manifold, see \Cref{global-regular-expression}.
	\end{rem}

	In general, it is difficult to check whether a given element in $\Omega^{*}\big((0,\infty);\Gamma(\widetilde{E}\boxtimes \widetilde{E})\big)$ is a regular expression. To obtain more characterizations of regular expressions, we restrict ourselves to analytic expressions.
	\begin{defn}
		Let $U$ be an open subset of $\mathbb{C}^{n}$. For any $r>0$, we define
		\[
		N(U,r)=\Big\{(\mathbf{z},\mathbf{w})\in\mathbb{C}^{n}\times\mathbb{C}^{n}:\mathbf{w}\in U,|\mathbf{z}-\mathbf{w}|<r\Big\}.
		\]
	\end{defn}
	
	On $N(U,r)$, we introduce the coordinates $$(\tilde{\mathbf{z}},\mathbf{w})=(\mathbf{z}-\mathbf{w},\mathbf{w}).$$ We will always use these coordinates for $N(U,r)$ in this paper.
	\begin{defn}
		Let $p\in\mathbb{C}^n$, $U\subset\mathbb{C}^{n}$ be a neighborhood of $p$,  $E$ is a trivial holomorphic vector bundle on $\mathbb{C}^{n}$, and $\widetilde{E}=E\otimes \Lambda^\bullet T_{0,1}^*\mathbb{C}^{n}$. For $r>0$, the linear space of \textbf{analytic expressions} $A\big(N(U,r);p\big)\subset\Omega^{*}\big((0,\infty);\Gamma((\widetilde{E}\boxtimes \widetilde{E})|_{N(U,r)})\big)$ is generated by expressions of the form
		\[
		a\, t^m+a't^{m'}dt,
		\]
		where $m,m'\in\mathbb{Z}$, and $a,a'$ are convergent power series at $(p,p)$ valued in $(\widetilde{E}\boxtimes \widetilde{E})|_{N(U,r)}$. Similarly, the linear space of \textbf{regular analytic expressions} $$A_{\sm}(N(U,r);p)\subset\Omega^{*}((0,\infty);\Gamma((\widetilde{E}\boxtimes \widetilde{E}|_{N(U,r)})))$$ is generated by expressions of the form 
		\[
		a(\mathbf{y}, d\mathbf{y})\, b\,t^m+a'(\mathbf{y}, d\mathbf{y})\, b'\,t^{m'}dt,
		\]
		where $m,m'\in\mathbb{Z}_{\geq 0}$, $a(\mathbf{y}, d\mathbf{y}),a'(\mathbf{y}, d\mathbf{y})$ are polynomials in the variables
		\[
		\mathbf{y} = (y_1, y_2, \ldots, y_n), \quad \text{and} \quad d\mathbf{y} = (dy_1, dy_2, \ldots, dy_n),
		\]
		with
		\[
		y_i = \frac{\bar{\tilde{z}}_{i}}{t}=\frac{\overline{z_i} - \overline{w_i}}{t}.
		\] The factors $b,b'$ are convergent power series at $(p,p)$ of $(\widetilde{E}\boxtimes \widetilde{E})|_{N(U,r)}$.
	\end{defn}

	Given $A(N(U,r);p)$, where $p\in U\subset\mathbb{C}^{n}$ and $r>0$, we will define a scaling action—namely, a version of Getzler's rescaling—on $ A(N(U,r);p) $. For $ \epsilon \in (0,1] $, the scaling operator $ \delta_\epsilon $ is defined by
	\be\label{rescaling}
	\delta_\epsilon\big(a(\tilde{\bfz}, {\bar{\tilde{\bfz}}}, t, \bfw, \bar{\bfw}, d {\bar{\tilde{\bfz}}}, d \bar{\bfw}, dt)\big)=a (\tilde{\bfz}, \epsilon {\bar{\tilde{\bfz}}}, \epsilon t, \bfw, \bar{\bfw}, \epsilon d {\bar{\tilde{\bfz}}}, d \bar{\bfw}, \epsilon dt),
	\ee
	for any $a\in A\big(N(U,r);p\big)$.
	
	This scaling action defines a filtration on $A(N(U,r);p)$:
	
	\begin{defn}
		For $m\in \mathbb{Z}$, we define
		\[
		F_mA\big(N(U,r);p\big)=\Big\{a\in A\big(N(U,r);p\big)\;:\;\lim_{\epsilon\rightarrow 0}\epsilon^{-m}\delta_\epsilon(a) \text{ exists}\Big\}.
		\]
		This gives a descending filtration:
		\[
		\cdots\supset F_mA\big(N(U,r);p\big)\supset F_{m+1}A\big(N(U,r);p\big)\supset\cdots.
		\]
	\end{defn}
	
	We now state two propositions regarding the filtration $F_{*}A\big(N(U,r);p\big)$.
	
	\begin{prop}\label{limit 0 check}
		Let $a\in A\big(N(U,r);p\big)$. Then $a\in F_mA\big(N(U,r);p\big)$ if and only if 
		\[
		\lim_{\epsilon\rightarrow 0}\epsilon^{-m+1}\delta_\epsilon(a)=0.
		\]
	\end{prop}
	\begin{proof}
		This follows from the fact that $\delta_\epsilon a$ admits a real Laurent expansion in $\epsilon$ of the form
		\[
		\delta_\epsilon a = \sum_{i=-l}^{\infty} a_i \epsilon^i,
		\]
		for some $l \in \mathbb{Z}_{\geq 0}$, with $a_i \in A\big(N(U,r);p\big)$. The limit $\epsilon^{-m}\delta_\epsilon(a)$ exists if and only if all terms with $i < m$ vanish, which is equivalent to saying $\epsilon^{-m+1}\delta_\epsilon(a) \to 0$ as $\epsilon \to 0$.
	\end{proof}
	
	\begin{prop}\label{independent of coordinates for filtration}
		Let $f:U\subset \mathbb{C}^{n}\rightarrow V\subset \mathbb{C}^{n}$ be a holomorphic map with $p\in U$ and $r,r'>0$. Suppose that $f\big(N(U,r)\big)\subset N(V,r')$ and that 
		\[
		f^{*}\Big(A\big(N(V,r');f(p)\big)\Big)\subset A\big(N(U,r);p\big).
		\]
		Then for any $m\in \mathbb{Z}$, we have 
		\[
		f^{*}\big(F_{m}A\big(N(V,r');f(p)\big)\big) \subset F_{m}A\big(N(U,r);p\big).
		\]
	\end{prop}
	\begin{proof}
		Let $(\tilde{\bfz},\bfw)$ and $(\tilde{\bfz}',\bfw')$ denote the coordinates on $N(U,r)$ and $N(V,r')$, respectively. For $\epsilon>0$, denote by $\delta_\epsilon$ and $\delta'_\epsilon$ the scaling operators on $N(U,r)$ and $N(V,r')$. Then
		\[
		\begin{cases}
			\delta'_\epsilon\bar{\tilde{z}}'_i = \epsilon\big(\bar{f}_i(\bfz)-\bar{f}_i(\bfw)\big) = \delta_\epsilon\bar{\tilde{z}}'_i + O(\epsilon),\\
			\delta'_\epsilon t = \delta_\epsilon t = \epsilon t.
		\end{cases}
		\]
		
		For $a'\in A(N(V,r');f(p))$, we can write
		\[
		a' = \sum_{i=l_1}^{l_2} t^i \Big( \sum_{\substack{\mathbf{j}\in\mathbb{Z}_{\geq0}^{n} \\ \mathbf{k}\in \{0,1\}^{n}}} b^i_{\mathbf{j}\mathbf{k}} \bar{\tilde{\bfz}}'^{\mathbf{j}} d\bar{\tilde{\bfz}}'^{\mathbf{k}} + \sum_{\substack{\mathbf{j}\in\mathbb{Z}_{\geq0}^{n} \\ \mathbf{k}\in \{0,1\}^{n}}} c^i_{\mathbf{j}\mathbf{k}} \bar{\tilde{\bfz}}'^{\mathbf{j}} d\bar{\tilde{\bfz}}'^{\mathbf{k}} dt \Big),
		\]
		where $b^i_{\mathbf{j}\mathbf{k}}$ and $c^i_{\mathbf{j}\mathbf{k}}$ are analytic sections of $(\widetilde{E}\boxtimes \widetilde{E})|_{N(V,r')}$ that are independent of $\bar{\tilde{\bfz}}$ and $d\bar{\tilde{\bfz}}$. Here for $\bfi\in\Z^n_{\geq0}$, $\bfz^\bfi:=z_1^{i_1}\cdots z_n^{i_n},$ and for $\bfi=(i_1,\cdots,i_n)\in\{0,1\}^n$ such that $i_{j_1}=\cdots i_{j_l}=1$, $d z^\bfi:=dz_{j_1}\wedge\cdots\wedge dz_{j_l}.$
		
		If $a'\in F_mA(N(V,r');f(p))$, then
		\[
		\begin{cases}
			b^i_{\mathbf{j}\mathbf{k}} = 0 & \text{if } |\mathbf{j}| + |\mathbf{k}| + i < m,\\
			c^i_{\mathbf{j}\mathbf{k}} = 0 & \text{if } |\mathbf{j}| + |\mathbf{k}| + i + 1 < m.
		\end{cases}
		\]
		Here for $\bfj=(j_1,\cdots,j_n)\in\Z_{\geq 0}^n \text{ or }\{0,1\}^n, |\bfj|:=j_1+\cdots+ j_n$.
		
		Then we have 
		\[
		\epsilon^{-m+1} \delta_\epsilon f^*(a') = f^*(\epsilon^{-m+1} \delta'_\epsilon a') + O(\epsilon).
		\]
		By Proposition \ref{limit 0 check}, we conclude that $f^*(a') \in F_m A(N(U,r);p)$.
	\end{proof}
	
	\begin{rem}
		Proposition \ref{independent of coordinates for filtration} shows that the filtration $F_*A(N(U,r);p)$ is independent of the choice of coordinates. This allows us to define the filtration globally on complex manifold. The idea of introducing scaling actions and filtrations is motivated by Getzler's rescaling, which plays a crucial role in the study of heat kernel expansions and index theorems. See \cite{berline2003heat, getzler1986short} for further details.
		
	\end{rem}

	Consider the vector field
	\begin{equation} \label{vf-X}
		X := t\partial_{t} + \sum_{i=1}^{n} \bar{\tilde{z}}_{i} \partial_{\bar{\tilde{z}}_{i}}.
	\end{equation}
	Using the filtration $F_{*}A(N(U,r);p)$, we have the following characterization of analytic regular expressions:
	
	\begin{prop}\label{filtration and regularity}
		Let $p \in U \subset \mathbb{C}^n$ and $r > 0$. An analytic expression $a \in A(N(U,r);p)$ lies in $A_\sm(N(U,r);p)$ if and only if $a \in F_0 A(N(U,r);p)$ and
		\[
		\iota_X\left(\lim_{\epsilon \to 0} \delta_\epsilon(a)\right) = 0,
		\]
		where $\iota_X$ denotes the interior product with respect to the vector field $X$.
	\end{prop}
	
	\begin{proof}
		We first prove the “only if” direction. Suppose $a \in A_\sm(N(U,r);p)$. Then
		\[
		\lim_{\epsilon \to 0} \delta_\epsilon(a) = \sum_{k=0}^{m} a_k(\mathbf{y}, d\mathbf{y}) \cdot b_k,
		\]
		for some polynomials $a_k(\mathbf{y}, d\mathbf{y})$ in the variables
		\[
		\mathbf{y} = (y_1, \dots, y_n), \quad d\mathbf{y} = (dy_1, \dots, dy_n), \quad y_j = \frac{\bar{\tilde{z}}_j}{t} = \frac{\bar{z}_j - \bar{w}_j}{t},
		\]
		and some $b_k$, which are analytic sections of $\widetilde{E} \boxtimes \widetilde{E}$ over $N(U,r)$, independent of $\bar{\tilde{\bfz}}$ and $d\bar{\tilde{\bfz}}$.

		Since $\iota_X(y_j) = \iota_X(dy_j) = \iota_X(b_k) = 0$, it follows that
		\[
		\iota_X\big(\lim_{\epsilon \to 0} \delta_\epsilon(a)\big) = 0.
		\]
		Moreover, by definition of regular expressions, $a \in F_0 A(N(U,r);p)$.
		
		For the “if” direction, suppose $a \in F_0 A(N(U,r);p)$ and $\iota_X\left(\lim_{\epsilon \to 0} \delta_\epsilon(a)\right) = 0$. Then we can write
		\be\label{expan delta epsilon}
		\lim_{\epsilon \to 0} \delta_\epsilon(a) = \sum_{k=0}^{m} \left( a_k(\mathbf{y}, d\mathbf{y}) \cdot b_k + \frac{dt}{t} \cdot a_k'(\mathbf{y}, d\mathbf{y}) \cdot b_k' \right),
		\ee
		for some polynomials $a_k$, $a_k'$ in the variables $\mathbf{y}, d\mathbf{y}$, and some analytic sections $b_k$, $b_k'$ of $\widetilde{E} \boxtimes \widetilde{E}$ over $N(U,r)$, which are independent of $\bar{\tilde{\bfz}}$ and $d\bar{\tilde{\bfz}}$.
		Then
		\[
		 \sum_{k=0}^{m} a_k'(\mathbf{y}, d\mathbf{y}) \cdot b_k'=\iota_X\big(\lim_{\epsilon \to 0} \delta_\epsilon(a)\big)  = 0.
		\]
       Here the first equality follows from a straightforward computation, while the
second follows from the assumption.
		Then by \eqref{expan delta epsilon},\be\label{delta epsilon is regular}\lim_{\epsilon \to 0} \delta_\epsilon(a)=\sum_{k=0}^{m} a_k(\mathbf{y}, d\mathbf{y}) \cdot b_k \in A_{\sm}\big(N(U, r); p\big).\ee It is easily checked that if $a \in F_0 A(N(U,r);p)$, $$a - \lim_{\epsilon \to 0} \delta_\epsilon(a)\in F_1 A(N(U,r);p) \subset A_{\sm}\big(N(U, r); p\big),$$ we conclude by \eqref{delta epsilon is regular} that $a \in A_{\sm}\big(N(U, r); p\big)$.
		
	\end{proof}
	
	We now present an important fiberwise criterion for regularity.
	
	Let $\pi_2: N(U,r) \subset \mathbb{C}^n \times \mathbb{C}^n \to \mathbb{C}^n$ denote the projection to the second factor:
	\[
	\pi_2(\tilde{\bfz}, \bfw) = \bfw.
	\]
	
	\begin{defn}
		Let $p \in U \subset \mathbb{C}^n$ and $r > 0$. An analytic expression $a \in A(N(U,r);p)$ is said to be \textbf{regular along} $\pi_2^{-1}(p')$ for $p' \in U$ if there exists $a' \in A_\sm(N(U,r);p)$ such that
		\[
		a|_{\pi_2^{-1}(p')} = a'|_{\pi_2^{-1}(p')}.
		\]
	\end{defn}
	
	We write
	\[
	N(p',r) := \{\bfz \in \mathbb{C}^n : |\bfz - p'| < r\}
	\]
	and use $A_\sm(N(p',r);p)$ to denote the space of analytic expressions that are regular along $\pi_2^{-1}(p') \cong N(p',r)$.
	
	\begin{prop}\label{fiberwise regular}
		Let $p \in U \subset \mathbb{C}^n$ and $r > 0$. Then $a \in A_\sm(N(U,r);p)$ if and only if $a \in A_\sm(N(p',r);p)$ for every $p' \in U$.
	\end{prop}
	
	\begin{proof}
		By Proposition \ref{filtration and regularity}, $a$ is regular if and only if
		\[
		\begin{cases}
			\lim_{\epsilon \to 0} \epsilon \delta_\epsilon(a) = 0, \\
			\iota_X\left( \lim_{\epsilon \to 0} \delta_\epsilon(a) \right) = 0.
		\end{cases}
		\]
		Since the vector field $\sum \bar{\tilde{z}}_i \partial_{\bar{\tilde{z}}_i}$ is tangent to the fibers of $\pi_2$, both conditions are equivalent to holding along every fiber $\pi_2^{-1}(p')$. This proves the proposition.
	\end{proof}
	
	\begin{rem}
		Proposition \ref{fiberwise regular} reduces the verification of regularity to a fiberwise check. As we will see in later sections, this greatly simplifies the proof of the regularity of the propagator.
	\end{rem}
	
	\subsection{Regular expressions and divisorial type singularities}\label{subsection regular expressions and divisorial type singularities}
	\def\bfy{{\mathbf{y}}}
	
	In this subsection, we establish the connection between regular expressions of and divisorial type singularities. Let $U \subset \mathbb{C}^n$ be an open subset, and let $E$ be a trivial holomorphic vector bundle over $\mathbb{C}^n$. Set $\widetilde{E} = E \otimes \Lambda^\bullet T^{*}_{0,1}\mathbb{C}^n$. We further assume that $U \subset \mathbb{C}^n$ is equipped with a real analytic K\"ahler metric such that the squared distance function $\rho^2 : U \times U \to \mathbb{R}$ is smooth.
	
	Let $\eta \in C_c^\infty(\mathbb{R})$ be a cutoff function satisfying $\eta(t) \equiv 1$ for $|t| \leq 1$, and $\eta(t) = 0$ for $|t| \geq 2$.
	
	We begin with the following:
	
	\begin{prop}\label{thm41}
		Let $ C \in S_\sm(U \times U) $. Then there exist polynomials $ \alpha_k(\tilde{\bfy}, d\tilde{\bfy}) $, $ \tilde{\alpha}_k(\tilde{\bfy}, d\tilde{\bfy}) $ and smooth differential forms $ \beta_k, \tilde{\beta}_k \in \Gamma(\widetilde{E} \boxtimes \widetilde{E}),k=1,\cdots,l $ such that
		\begin{equation}\label{thm41eq0}
			\int_0^\infty \eta(t)\, e^{-\frac{\rho^2}{2t}}\, C 
			= \sum_{k=0}^{l} \alpha_k(\tilde{\bfy}, d\tilde{\bfy})\, \beta_k
			+ \sum_{k=0}^{l} \tilde{\alpha}_k(\tilde{\bfy}, d\tilde{\bfy})\, \tilde{\beta}_k \ln(\rho^2),
		\end{equation}
		where 
		\[
		\tilde{\bfy}=(\tilde{y}_1,\cdots,\tilde{y}_n)\quad\text{and}\quad \tilde{y}_i = \frac{2(\bar{z}_i - \bar{w}_i)}        {\rho^2(\bfz, \bar{\bfz}, \bfw, \bar{\bfw})},
		\]
		and the integral is taken over the $ t $-variable. That is, if 
		\[
		w = \alpha + dt \wedge \beta, \quad \text{with } \alpha, \beta \in C^\infty\big((0,\infty); \Gamma(\widetilde{E} \boxtimes \widetilde{E}) \big),
		\]
		then 
		\[
		\int_0^\infty w := \int_0^\infty \beta \, dt.
		\]
		
	\end{prop}
	
	\begin{proof}
		It suffices to consider the case where
		\begin{equation}\label{thm41eq1}
			\eta C_1 = a(\bfy, d\bfy) b,
		\end{equation}
		or
		\begin{equation}\label{thm41eq2}
			\eta C_2 = a(\bfy, d\bfy) b_2 dt,
		\end{equation}
		where $a$ is a monomial of degree $m$ with $m \geq 0$, and $b\in C_c^\infty([0, \infty), \Gamma(\widetilde{E} \boxtimes \widetilde{E}))$. 
		
		Let $\hat{b}(s)$ denote the Mellin transform of $b(t)$:
		\begin{equation}\label{mellin}
			\hat{b}(s) = \int_0^{\infty} \tau^{s-1} b(\tau)\, d\tau.
		\end{equation}
		Then, by Taylor expansion, it is easy to see that $\hat{b}(s)$ is a
$\Gamma(\widetilde{E} \boxtimes \widetilde{E})$-valued meromorphic
function of $s$, with simple poles at
$s=0,-1,-2,\dots$.
		
		By the inverse Mellin transform, we obtain
		\begin{equation}\label{inverse-mellin}
			b(t) = \frac{1}{2 \pi i} \int_{c-i \infty}^{c+i \infty} t^{-s} \hat{b}(s) ds,
		\end{equation}
		for some constant $c>0$.
		
		Applying Fubini's theorem, \eqref{mellin} and \eqref{inverse-mellin}, we derive, for $c>0$,
		\begin{equation}\label{thm41eq3}
			\begin{aligned}
				\int_0^{\infty} \eta(t) e^{-\frac{\rho^2}{2t}} C_1 &= \frac{1}{2 \pi i} \int_0^{\infty} e^{-\frac{\rho^2}{2t}} a(\bfy, d\bfy) \int_{c-i \infty}^{c+i \infty} t^{-s} \hat{b}(s) ds  \\
				&= \frac{1}{2 \pi i} \int_{c-i \infty}^{c+i \infty}  \left(\int_0^{\infty} e^{-\frac{\rho^2}{2t}} a(\bfy, d\bfy) t^{-s}\right) \hat{b}(s) ds .
			\end{aligned}
		\end{equation}
		
		Next, we evaluate the inner integral. Changing variables via $t' = \frac{\rho^2}{2t}$, we rewrite it as
		\begin{equation}\label{thm41eq4}
			\int_0^{\infty} e^{-\frac{\rho^2}{2t}} a(\bfy, d\bfy) t^{-s}  = \left(\frac{\rho}{\sqrt{2}}\right)^{-2s} \int_0^{\infty} e^{-t'} a\big( t' \tilde{\bfy}, d(t' \tilde{\bfy}) \big) (t')^{s} .
		\end{equation}
		
		Expanding $a(t' \tilde{\bfy}, d(t' \tilde{\bfy}))$ in powers of $t'$, we obtain
		\begin{equation}\label{thm41eq5}
			a\big( t' \tilde{\bfy}, d(t' \tilde{\bfy}) \big) = a_1(\tilde{\bfy}, d\tilde{\bfy}) (t')^{m} + a_2(\tilde{\bfy}, d\tilde{\bfy}) (t')^{m-1} dt',
		\end{equation}
		for some polynomials $a_1, a_2$.

		Substituting \eqref{thm41eq5} into \eqref{thm41eq4}, we obtain
		\begin{equation}\label{thm41eq6}
			\int_0^{\infty} e^{-\frac{\rho^2}{2t}} a(\bfy, d\bfy) t^{-s}  = \left(\frac{\rho}{\sqrt 2}\right)^{-2s}  a_2(\tilde{\bfy}, d\tilde{\bfy}) \Gamma(s+m).
		\end{equation}
		
		Combining \eqref{thm41eq3} and \eqref{thm41eq6}, we deduce that for $c>0$,
		\begin{equation}\label{thm41eq7}
			\int_0^{\infty} \eta(t) e^{-\frac{\rho^2}{2t}} C_1 = \frac{1}{2 \pi i} \int_{c-i \infty}^{c+i \infty} a_2(\tilde{\bfy}, d\tilde{\bfy})\hat{b}(s) \left(\frac{\rho}{\sqrt 2}\right)^{-2s} \Gamma(s+m) ds.
		\end{equation}

 Noting that $dt=d\Big(\frac{\rho^2}{2t'}\Big)=\frac{\rho d\rho}{t'}-\frac{\rho^2dt'}{2(t')^2}$, we compute similarly that, for $c>1$, we find
		\begin{equation}\label{thm41eq8}\ba
			\int_0^{\infty} \eta(t) e^{-\frac{\rho^2}{2t}} C_2 &= -\frac{1}{2 \pi i} \int_{c-i \infty}^{c+i \infty}  a_1(\tilde{\bfy}, d\tilde{\bfy})\hat{b}(s)\left(\frac{\rho}{\sqrt 2}\right)^{2-2s} \Gamma(s+m-1) ds\\
            &+\frac{1}{4 \pi i} \int_{c-i \infty}^{c+i \infty}  a_2(\tilde{\bfy}, d\tilde{\bfy})\hat{b}(s)d\rho\left(\frac{\rho}{\sqrt 2}\right)^{1-2s} \Gamma(s+m-1) ds\ea
		\end{equation}
		
Since \(\Gamma(s)\) has simple poles at \(s=0,-1,-2,\ldots\), it follows that,
if \(m\geq1\), then
\[
        \hat b(s)\Gamma(s+m),
        \qquad
        \hat b(s)\Gamma(s),
        \qquad
        \hat b(s)\Gamma(s+m-1)
\]
have at most double poles at these points. When \(m=0\), there is no second term
on the right-hand side of \eqref{thm41eq8}, while \(\hat b(s)\Gamma(s-1)\) has
double poles at \(s=0,-1,-2,\ldots\), and a simple pole at \(s=1\).

		From \eqref{decompostion ug} in \Cref{lem42}, \Cref{lem43}, along with \eqref{thm41eq7} and \eqref{thm41eq8}, we conclude that \eqref{thm41eq0} holds.
		
	\end{proof}
	
	\Cref{lem42} and \Cref{lem43} are instrumental in the proof of \Cref{thm41}. Although they follow from the standard theory of the Mellin transform, we could not find a reference that covers the precise versions we require, so we include a proof here for completeness.

	\def\Re{{\mathrm{Re}}}
	\def\Im{\mathrm{Im}}
	\def\half{{\frac{1}{2}}}
	\begin{lem}\label{lem42}
Let $c>1$ be a real number, and let $h(s)$ be a meromorphic function with
at most double poles at $s=0,-1,-2,\dots$, and a simple pole at $s=1$.

Furthermore, assume:
\begin{enumerate}[(a)]
\item Near $s=-k$ for each $k=0,1,2,\dots$, the function $h(s)$ has the
Laurent series expansion
\begin{equation}\label{laurent1}
h(s)=\frac{a_k}{(s+k)^2}+\frac{b_k}{s+k}+\cdots.
\end{equation}
Near $s=1$, $h(s)$ has the Laurent series expansion
\begin{equation}\label{laurent2}
h(s)=\frac{b_{-1}}{s-1}+\cdots.
\end{equation}

\item For any $N>2$, consider the region
\[
S_{-N,c}:=\big\{s\in\C:\Re(s)\in[-N,c],\ |\Im(s)|\ge1\big\}.
\]
For any integer $l\ge1$, there exists a constant $C_{N,l}>0$ such that for
all $s\in S_{-N,c}$,
\begin{equation}\label{lem42eq1}
|h(s)|\le\frac{C_{N,l}}{|\Im(s)|^l}.
\end{equation}

\item For any $t\in\R$, let
\[
L_t:=\{s\in\C:\Re(s)=t\}.
\]
Then for each integer $k\ge2$, there exists a constant $M_k>0$ such that
for all $s\in L_{-k-\frac{1}{2}}$,
\begin{equation}\label{lem42eq2}
|h(s)|\le M_k.
\end{equation}
\end{enumerate}

Then the function $g(u)$ defined by
\[
g(u):=\frac{1}{2\pi i}\int_{L_c} h(s)\,u^{-s}\,ds
\]
is smooth on $(0,\infty)$ and has the asymptotic expansion near $u=0$:
\begin{equation}\label{expangu}
g(u)\sim b_{-1}u^{-1}+\sum_{j=0}^{\infty}\big(b_j u^j-a_j\ln u\cdot u^j\big).
\end{equation}
Moreover, there exist $\delta>0$ and $\varphi,\psi\in C^\infty\big([0,\delta)\big)$
such that
\begin{equation}\label{decompostion ug}
g(u)=b_{-1}u^{-1}+\varphi(u)+\psi(u)\ln u,\qquad u\in(0,\delta).
\end{equation}
\end{lem}

\begin{proof}
By \eqref{lem42eq1}, it follows that $g(u)\in C^\infty((0,\infty))$.

Using \eqref{laurent1}, \eqref{laurent2}, \eqref{lem42eq1}, and the residue
theorem, we obtain for any integer $k\ge2$,
\begin{equation}\label{lem42eq3}
\frac{1}{2\pi i}\int_{L_c} h(s)u^{-s}\,ds-\frac{1}{2\pi i}\int_{L_{-k-\frac{1}{2}}} h(s)u^{-s}\,ds
=
b_{-1}u^{-1}+\sum_{j=0}^{k}\big(b_j u^j-a_j\ln u\cdot u^j\big).
\end{equation}

Applying \eqref{lem42eq1} and \eqref{lem42eq2}, we obtain the estimate
\begin{equation}\label{lem42eq4}
\bigg|\int_{L_{-k-\frac{1}{2}}} h(s)u^{-s}\,ds\bigg|
\le \big(2C_{k+\frac{1}{2},2}+2M_k\big)u^{k+\frac{1}{2}}.
\end{equation}
The asymptotic expansion \eqref{expangu} follows directly from
\eqref{lem42eq3} and \eqref{lem42eq4}.

By \eqref{lem42eq1}, one easily checks that
\be\label{g prime u}
g'(u)=-\frac{1}{2\pi i}\int_{L_c} sh(s)u^{-s-1}\,ds
=-\frac{u^{-1}}{2\pi i}\int_{L_c} sh(s)u^{-s}\,ds.
\ee
It is easy to check that \eqref{lem42eq1} also holds for $sh(s)$.
Next, we verify that \eqref{lem42eq2} holds for $sh(s)$.
By \eqref{lem42eq1}, it suffices to verify this for
\[
s\in U_k:=L_{-k-\half}\cap\{s'\in\C:|\Im(s')|\le1\}.
\]
In fact, on $U_k$, we have
\[
|sh(s)|\le |s|M_k\le \Big(k+\half+1\Big)M_k.
\]

The Laurent expansions of $sh(s)$ are determined by those of $h(s)$.
Therefore, using \eqref{g prime u} and applying the above argument to
$\int_{L_c} sh(s)u^{-s}\,ds$, we conclude that the asymptotic behavior of $g'(u)$ is exactly
the termwise derivative of the right-hand side of \eqref{expangu}.
Repeating the same argument shows that the same statement holds for
derivatives of any order.

Now we prove \eqref{decompostion ug}. By Borel's theorem, there exist $\varphi_0,\psi_0\in C^\infty([0,\delta))$ such that
\[
\varphi_0(u)\sim \sum_{j=0}^\infty b_j u^j,
\qquad
\psi_0(u)\sim -\sum_{j=0}^\infty a_j u^j
\qquad (u\to0^+).
\]
Moreover, their derivatives have the corresponding asymptotic behavior.

Set
\[
R(u):=g(u)-b_{-1}u^{-1}-\varphi_0(u)-\psi_0(u)\ln u,\qquad u>0.
\]
Since $g$ and all its derivatives admit the prescribed asymptotic expansions, and since the same is true for $b_{-1}u^{-1}+\varphi_0(u)+\psi_0(u)\ln u$, it follows that for every $m,N\ge0$,
\[
R^{(m)}(u)=O(u^N),\qquad u\to0^+.
\]
Thus extending $R$ by setting $R(0)=0$ gives a function still in $C^\infty([0,\delta))$. Therefore, setting
\[
\varphi:=\varphi_0+R,\qquad \psi:=\psi_0,
\]
we obtain $\varphi,\psi\in C^\infty([0,\delta))$ and
\[
g(u)=b_{-1}u^{-1}+\varphi(u)+\psi(u)\ln u.
\]
This proves the lemma.
\end{proof}

	\begin{lem}\label{lem43}
		Let $v \in C_c^{\infty}([0,\infty))$, and let $\hat{v}(s)$ be its Mellin transform:
		\[
		\hat{v}(s) = \int_0^{\infty} t^{s-1} v(t) dt.
		\]
		Then for any integer $m \geq -1$, the function $\hat{v}(s) \Gamma(s+m)$ satisfies the conditions of \Cref{lem42}.
	\end{lem}
	
	\begin{proof}
		By the standard properties of the Gamma function and the Mellin transform, $\hat{v}(s) \Gamma(s+m)$ has at most double poles at $s = 0, -1, -2, \dots$ and a simple pole at $s = 1$.
		
		For any integer $k \geq 2$, when $\Re(s) > -k$, we have
		\begin{equation}\label{lem43eq1}
			\hat{v}(s) = \frac{1}{s(s+1)\cdots(s-1+k)} \int_0^{\infty} v^{(k)}(t) t^{s-1+k} dt.
		\end{equation}
		Furthermore, from Stirling’s formula, for $s \in S_{N,c}$,
		\begin{equation}\label{lem43eq2}
			|\Gamma(s+m)| \leq C_{m,N} e^{-\frac{\pi}{4} |\Im(s)|}.
		\end{equation}
		The bounds \eqref{lem42eq1} and \eqref{lem42eq2} follow directly from \eqref{lem43eq1} and \eqref{lem43eq2}.
	\end{proof}
	\def\Bl{\mathrm{Bl}}
	\def\ta{{\tilde{\a}}}
	\def\tb{{\tilde{\b}}}
	\def\ty{{\tilde{y}}}
	
	Let $\triangle:=\{(\bfz,\bfz):\bfz\in U\}$ denotes the diagonal of $U\times U$. The blow up of $U\times U$ is denoted by $\Bl_{\triangle}(U \times U)$, and we have a canonical map $p:\Bl_{\triangle}(U \times U)\rightarrow U\times U$, such that $p^{-1}(\triangle)$ is the exceptional divisor on $\Bl_{\triangle}(U \times U)$.
	
	The following proposition is the main result of this subsection:

	\begin{prop}\label{lem44}
		If $C \in S_\sm(U\times U)$, then we have 
		\begin{enumerate}[(1)]
			\item $\int_0^\infty \eta(t)e^{-\frac{\rho^2}{2t}}C$ has divisorial type singularities along $p^{-1}(\triangle)$ (see \Cref{divisorial type singularities}). Moreover, any holomorphic derivatives of $\int_0^\infty \eta(t)e^{-\frac{\rho^2}{2t}}C$ have divisorial type singularities along $p^{-1}(\triangle)$. Here $\rho$ is the distance function on $M.$
			\item \[
			\triangle ^{*}\Big(\int_0^\infty \eta(t)e^{-\frac{\rho^2}{2t}}C\Big):=\lim_{\epsilon\rightarrow0}\triangle ^{*}\Big(\int_{\epsilon}^{\infty} \eta(t)e^{-\frac{\rho^2}{2t}}C\Big)\]
			is a smooth bundle-valued differential form on $U$, where $\triangle ^{*}$ is the pull back of bundle-valued differential forms along the diagonal embedding
			\[
			\triangle:U\rightarrow U\times U.
			\]
			Moreover, the pull backs of arbitrary holomorphic derivatives of $\int_0^\infty \eta(t)e^{-\frac{\rho^2}{2t}}C$ along $\triangle$ are smooth bundle-valued differential forms.
		\end{enumerate}
	\end{prop}
	\begin{proof}
		First, by $(\ref{distance function on arbitrary})$, we can show that if $C \in S_\sm(U\times U)$, then $e^{\frac{\rho^2}{2t}}\partial_{z_{i}}\big(e^{-\frac{\rho^2}{2t}}C\big)$ and $e^{\frac{\rho^2}{2t}}\partial_{w_{i}}\big(e^{-\frac{\rho^2}{2t}}C\big)$ are also regular expressions. Moreover, if $C \in S_\sm(U\times U)$, then $\triangle ^{*}\big(e^{-\frac{\rho^2}{2t}}C\big)=\triangle ^{*}(C)$ is a smooth bundle-valued differential form, so 
		\[
		\triangle ^{*}\left(\int_0^\infty \eta(t)e^{-\frac{\rho^2}{2t}}C\right):=\lim_{\epsilon\rightarrow0}\triangle ^{*}\left(\int_{\epsilon}^{\infty} \eta(t)e^{-\frac{\rho^2}{2t}}C\right)\]
		is a smooth bundle-valued differential form on $U$. 
		
		Now, let's prove $\int_0^\infty \eta(t)e^{-\frac{\rho^2}{2t}}C$ has divisorial type singularities. By Proposition \ref{thm41}, we only need to prove that $ \tilde{y}_i$, $d( \tilde{y}_i)$, and $\ln(\rho^{2})$ have divisorial type singularities along $p^{-1}(\triangle)$. 
		
		Without loss of generality, we assume $U=\C^n$. We only need to prove that there's a neighborhood $W\subset \C^n \times \C^n$ at $(0,0)$, such that our claims hold in $p^{-1}(W)\subset \Bl_{\triangle}(\C^n \times \C^n)$.

		One can identify $\Bl_{\triangle}(\C^n \times \C^n)$ with the closure of the following natural embedding:
		\[
		i:\C^n \times \C^n \setminus \triangle \hookrightarrow \left( \C^n \times \C^n \right) \times \left( \C^n \times \C^n \setminus \triangle \right)/\C^*, \quad (\bfz,\bfw) \mapsto (\bfz,\bfw,\bfz,\bfw),
		\]
		where the $\C^*$-action is given by $\lambda \cdot (\bfz,\bfw) = (\lambda \bfz, \lambda \bfw)$ for $\lambda \in \C^*$.
		
		\def\P{\mathbb{P}}
		It is easy to see that the map
		\[
		(\bfz,\bfw) \mapsto (\bfw,\; [z_1 - w_1 : \cdots : z_n - w_n])
		\]
		yields an identification of $(\C^n \times \C^n \setminus \triangle)/\C^*$ with $\C^n \times \P^{n-1}$. Under this identification, $\Bl_{\triangle}(\C^n \times \C^n)$ has the following description:
		\begin{align*}
			&\Bl_{\triangle}(\C^n \times \C^n)\\
			&\cong\big\{(\bfz,\bfw,[\lambda_1 : \cdots : \lambda_n])\in \left( \C^n \times \C^n \right) \times  \P^{n-1} :\bfz-\bfw=k\cdot \boldsymbol{\lambda} \text{ for some }k\in \mathbb{C}\big\},
		\end{align*}
		where $\boldsymbol{\lambda} =(\lambda_1,\cdots,\lambda_n)\in\C^n.$
		
		Let $\{U_i\}_{i=1}^n$ be an open cover of $\Bl_\triangle(\C^n \times \C^n)$, where each $U_i$ corresponds to the subset defined by $\lambda_i\neq 0$.

		On the chart $U_1$, we may use the following local coordinates:
		\[
		(z_1 - w_1, \lambda_2, \ldots, \lambda_n, w_1, \ldots, w_n).
		\]
		
		By the proof of Lemma \ref{fake distance function is defining function}, there exists a small neighborhood $W\subset \C^n \times \C^n$ of $(0,0)$, such that $\frac{\rho^2(\bfz,\bar\bfz,\bfw,\bar\bfw)}{|z_1-w_1|^{2}}$ is a positive smooth function on $p^{-1}(W)\cap U_1$. On the other hand, on $p^{-1}(W)\cap U_1$, we have
		\[
		\begin{cases}
			\tilde{y}_i&=\frac{1}{z_{1}-w_{1}}\cdot \frac{\bar{\lambda}^i|z_{1}-w_{1}|^{2}}{\rho^2(\bfz,\bar{\bfz},\bfw,\bar{\bfw})},\\
			d(\tilde{y}_i)&=\frac{d(z_{1}-w_{1})}{(z_{1}-w_{1})^2}\cdot \frac{\bar{\lambda}^i|z_{1}-w_{1}|^{2}}{\rho^2(\bfz,\bar{\bfz},\bfw,\bar{\bfw})}+\frac{1}{z_{1}-w_{1}}\cdot d\Big(\frac{\bar{\lambda}^i|z_{1}-w_{1}|^{2}}{\rho^2(\bfz,\bar{\bfz},\bfw,\bar{\bfw})}\Big),\\
			\ln(\rho^{2})&=\ln\Big(\frac{\rho^{2}}{|z_1-w_1|^2}\Big)+\ln(|z_1-w_1|^2).
		\end{cases}
		\]
		Therefore, they have divisorial type singularities on $p^{-1}(W)\cap U_1$. Similar results hold on other coordinate charts. This proves our assertion.
	\end{proof}
	
	Proposition \ref{lem44} reduces the proof of Theorem \ref{propagator has divisorial type singularities} to showing that propagator in the Schwinger space admits a regular expression in a neighborhood of the diagonal. We will prove this in next section.
	
	\section{Regularities of Propagator in Schwinger Space}\label{Section 4}

	This section establishes the regularity of the propagator in Schwinger space. In \Cref{heatexp}, we review the notion of heat kernel expansions. In \Cref{getzler}, we prove \Cref{main0} using Getzler’s rescaling technique.

	Let $(M,E,\omega)$ be a triple as in \Cref{propagator has divisorial type singularities}.
	
	\begin{defn}\label{global-regular-expression}
		The linear space of regular expressions $S_\sm \subset \Omega^*\Big((0,\infty); \Gamma\big(\tE \boxtimes \tE\big)\Big)$ is generated by differential forms 
		\[
		w \in \Omega^*\Big((0,\infty); \Gamma\big(\tE \boxtimes \tE\big)\Big),
		\]
		such that for each local normal frame $(U, \bfz, \bfe)$, we have $w|_{U \times U} \in S_\sm(U \times U)$. Here, we identify $U$ with an open subset of $\C^n$, and view $E|_U$ as a trivial bundle via the chosen frame.

	\end{defn}

	\begin{defn}We use $\Op\big(S_\sm\big)\subset \End_\C\Big(\Omega^*\big((0,\infty),\Gamma(\tE\boxtimes \tE)\big)\Big)$ to denote the operators that preserve $S_\sm$. 
	\end{defn}
	
	It can be seen easily that $\nabla_{\frac{\p}{\p z^i}},\bar\p\in\Op\big(S_\sm\big).$
	
	Recall that the propagator in Schwinger space is given by 
	$$
	P_t=-d t \wedge\left(\bar{\partial}_{\widetilde{E}}^* \otimes \mathrm{id}\right) H_t+H_t \in \Omega^*\left((0,+\infty),\Gamma( \widetilde{E} \boxtimes\widetilde{E})\right).
	$$

	The main theorem in this section is that
	\begin{thm}\label{main0}
		We have  $e^{\frac{\rho^2}{2t}}P_t\in S_\sm$, where $\rho$ is the distance function induced by the metric $g$.
	\end{thm}

	Fix a local K\"ahler normal frame $(U,\bfz, \bfe)$ as described in \Cref{local-coord} at some point $p\in M$. Let $(U, \mathbf{w},\bf{f})$ denote the same coordinate chart on another copy of $M$. 
	
	In this section, without loss of generality, we use this frame to identify $U$ with an open set in $\mathbb{C}^n$, equipped with an analytic K\"ahler metric $g_{i\bar{j}}$, and take $E|_U$ to be the trivial holomorphic vector bundle $E = U \times V$, endowed with an analytic Hermitian metric $h_{a\bar{b}}$. 
	
	\def\Chris{1}
	\if\Chris0

	The Christoffel symbols of the Chern connection on $TM \to M$ satisfy:
	\begin{equation}
		\begin{cases}
			\nabla_{\frac{\partial}{\partial z_i}}^{TM} \frac{\partial}{\partial z_j} = \Gamma_{ij}^k \frac{\partial}{\partial z_k}, \\
			\nabla_{\frac{\partial}{\partial z_i}}^{TM} \frac{\partial}{\partial \bar{z}_j} = 0, \\
			\nabla_{\frac{\partial}{\partial \bar{z}_i}}^{TM} \frac{\partial}{\partial \bar{z}_j} = \Gamma_{\bar{i}\bar{j}}^{\bar{k}} \frac{\partial}{\partial \bar{z}_k}, \\
			\nabla_{\frac{\partial}{\partial \bar{z}_i}}^{TM} \frac{\partial}{\partial z_j} = 0.
		\end{cases}
	\end{equation}
	
	Similarly, the Christoffel symbols of the Chern connection on $E \to M$ satisfy:
	\begin{equation}
		\begin{cases}
			\nabla_{\frac{\partial}{\partial z_i}}^{E} e_a = \Gamma_{ia}^{b} e_b, \\
			\nabla_{\frac{\partial}{\partial \bar{z}_j}}^{E} e_a = 0.
		\end{cases}
	\end{equation}
	\fi
	
	Then under this identification, the heat equation \eqref{heateq1} takes the more explicit form:
	\begin{align}\label{heateq2}
		\p_t H_t(\bfz,\bar\bfz,\bfw,\bar\bfw)&= -(\Delta_{\widetilde{E}} \otimes \mathrm{id}) H_t(\bfz,\bar\bfz,\bfw,\bar\bfw), \\
		H_0(\bfz,\bar\bfz,\bfw,\bar\bfw)&= \delta(\bfz - \bfw, \bar{\bfz} - \bar{\bfw}) \, \omega(\bfz, \bfw) \, d^n(\bar{\bfz} - \bar{\bfw}),\label{initial condition}
	\end{align}
	where $\omega(\bfz, \bfw)$ is some holomorphic section of $E \boxtimes E$, $\delta$ denotes the standard delta distribution, and $d^n(\bar{\bfz} - \bar{\bfw})=d(\bz_1-\bw_1)\wedge\cdots\wedge d(\bz_1-\bw_1).$
	
	Since being a regular expression is a local property, it suffices to prove \Cref{main0} on $U \times U$.
	
	
	\def\defnsreg{1}
	\if\defnsreg0
	\begin{defn}
		
		Let $S_{\sm} \subset \Omega^*\big((0,\infty); \Gamma(\widetilde{E} \boxtimes \widetilde{E})\big)$ denote the linear space generated by expressions of the form
		\[
		a(\mathbf{y}, d\mathbf{y})\, b,
		\]
		where $a(\mathbf{y}, d\mathbf{y})$ is a polynomial in the variables
		\[
		\mathbf{y} = (y_1, y_2, \ldots, y_n), \quad \text{and} \quad d\mathbf{y} = (dy_1, dy_2, \ldots, dy_n),
		\]
		with
		\[
		y_i = \frac{\overline{z_i} - \overline{w_i}}{t}.
		\]
		Here, $d$ denotes the de Rham differential on $(0,\infty) \times M \times M$, so
		\[
		dy_i = d\left( \frac{\overline{z_i} - \overline{w_i}}{t} \right) = \frac{d(\overline{z_i} - \overline{w_i})}{t} - \frac{(\overline{z_i} - \overline{w_i})\, dt}{t^2}.
		\]
		The factor $b \in \Omega^*\Big((0,\infty); \Gamma\big(\widetilde{E} \boxtimes \widetilde{E}\big)\Big)$ is assumed to extend smoothly to a differential form in $\Omega^*\Big([0,+\infty); \Gamma\big(\widetilde{E} \boxtimes \widetilde{E}\big)\Big)$.

		It can be seen easily that $S_\sm$ is an algebra.
		
	\end{defn}

	\begin{defn}
		We use $\Op(S_\sm)\subset \End_\C\Big(\Omega^*\big((0,\infty);\Gamma(\tE\boxtimes \tE)\big)\Big)$ to denote the operators that preserve $S_\sm$. 
	\end{defn}
	
	It could be seen easily that $\nabla^{\widetilde{E}}_{X},\alpha\wedge,{\bar\p}\in\Op(S_\sm),$ where $\alpha\in S_\sm$ and $X\in\Gamma\big(T^{1,0}(M\times M)\big)$ 
	
	\begin{rem}
		
		The space $S_{\sm}$ is independent of the choice of holomorphic normal frame. To see this, let $\mathbf{z}'$ be another holomorphic coordinate system, and consider a change of coordinates given by $z_i = f_i(\mathbf{z}')$. Then we can write
		\[
		f_i(\mathbf{z}') = f_i(\mathbf{w}') + \sum_{j=1}^n f_{ij}( \mathbf{w}') (z'_j - w'_j)+\sum_{1\leq j,k\leq n}f_{ijk}(\bfz',\bfw')(z_j'-w_j')(z_k'-w_k'),
		\]
		for some holomorphic functions $f_{ij}$ and $f_{ijk}$.

		Now let $y_i' := \frac{z_i' - w_i'}{t}$. Then for $\bfw=(f_1(\bfw'),\cdots, f_n(\bfw'))$,
		\[\ba
		y_i &= \frac{\bar{z}_i - \bar{w}_i}{t}
		= \frac{f_i(\bar{\mathbf{z}}') - f_i(\bar{\mathbf{w}}')}{t}
		\\&= \sum_j \overline{f_{ij}( \bfw')} \cdot \frac{\bar{z}'_j - \bar{w}'_j}{t}+t\sum_{j,k}\overline{f_{ijk}(\bfz',\bfw')}\frac{(\bfz_j'-\bfw_j')(\bfz_k'-\bfw_k')}{t^2}\\
		&= \sum_j \overline{f_{ij}( \bfw')} \cdot y_j'+t\sum_{j,k}\overline{f_{ijk}(\bfz',\bfw')}y_j'y_k'\\
		\ea\]
		and
		\[
		dy_i = \sum_j d\left( \overline{f_{ij}(\mathbf{z}' , \mathbf{w}')} \right) y'_j
		+ \overline{f_{ij}(\mathbf{z}' ,\mathbf{w}')} \, dy'_j,
		\]
		where we recall that $d$ denotes the de Rham differential on $(0,\infty) \times M \times M$.
		
	\end{rem}\fi

	\subsection{Heat kernel expansion}\label{heatexp}
	
	We aim to prove \Cref{main0} using heat kernel expansions and Getzler’s rescaling technique. In this subsection, we briefly review the relevant properties of heat kernel expansions.

	\def\tz{{\tilde{z}}}
	\def\tfz{{\tilde{\bfz}}}
	\def\btz{{\bar{\tilde{z}}}}
	\def\btfz{{\bar{\tilde{\bfz}}}}
	
	Fix a local K\"ahler normal frame $(U,\bfz, \bfe)$ as in \Cref{local-coord} at some point $p\in M$. Let $(U, \mathbf{w},\bf{f})$ denote the same coordinate chart on another copy of $M$. Let
	\be\label{tilde z coord}
	\begin{cases}
		\tz_i = z_i - w_i, \\
		\btz_i = \bar{z}_i - \bar{w}_i,
	\end{cases}
	\ee
	and let $\tfz = (\tz_1, \cdots, \tz_n)$. Then we will focus on the coordinates
	\be\label{shift-ccord}
	(\tfz, \bfw)
	\ee
	on $U \times U$.
	
	We have the following heat kernel expansion:
	
	\begin{lem}\label{HeatAsy}
		Let $H_t$ be the solution to the heat equation \eqref{heateq2} and \eqref{initial condition}. Then there exists a sequence of sections $u_{i} \in \Gamma(\tE \boxtimes \tE)$ for $i=0,1,2, \cdots$, such that:
		\begin{enumerate}[(1)]
			\item $u_i$ is analytic.
			\item For any nonnegative integers $k, l, m$, 
			\[
			\bigg\| \frac{\partial^{k}}{\partial t^{k}} \Big( t^n e^{\frac{\rho^{2}}{2t}} H_t - \sum_{i=0}^{m} t^i u_i \Big) \bigg\|_{C^l(K \times K)} \leq C(K,l) t^{m+1 - k},
			\]
			where $K \subseteq U$ is a compact subset and $C(K,l) > 0$ is a constant. Here\\ $\|\cdot\|_{C^l(K \times K)}$ denotes the standard $C^l$-norm over $K \times K$.
			
		\end{enumerate}
	\end{lem}
	
	\begin{proof}
		The second statement follows directly from \cite[Theorem 1.1]{ludewig2019strong}. Here we briefly explain the proof of the first statement.
		
		It follows from \eqref{initial condition} that 
		\be\label{u0 at diagonal}
		u_0(\bfz, \bar\bfz, \bfw, \bar\bfw)\big|_{\triangle} = \omega(\bfz,\bfw) \, d^n(\bar\bfz - \bar\bfw)|_{\triangle},
		\ee
		where $\triangle \subset U \times U$ denotes the diagonal.
		
		The heat kernel coefficients $u_i$ can be computed recursively by solving ODEs, as described in \cite[Chapter 2]{berline2003heat}. Since both $(M, g)$ and the bundle $E \to M$ are real analytic, and the initial condition is given by \eqref{u0 at diagonal}, it follows from standard theory of ODEs that each $u_i$ is analytic.
		
	\end{proof}
	
	It follows from \Cref{HeatAsy} that:
	\begin{cor}\label{cor32}
		We have $ H_t = e^{-\frac{\rho^2}{2t}} \Big(\big( \sum_{i=0}^n t^{-n+i} u_i \big)+ b_0 \Big) $ with $ t^{-1}b_0\in S_\sm $.
	\end{cor}
	
	\begin{proof}
		From \Cref{HeatAsy}, we have when $t\in(0,1)$
		\be\label{cor32eq1}
		\Big\| \frac{\p^{k}}{\p t^{k}} \Big( e^{\frac{\rho^2}{2t}} H_t - \sum_{i=0}^{n+k} t^{i-n} u_i \Big) \Big\|_{C^l(K \times K)} \leq C_{K,k,l} t.
		\ee
		This implies that the $ k $-th derivatives in the time variable and any spatial derivatives of
		\[
		e^{\frac{\rho^2}{2t}} H_t - \sum_{i=0}^{n+k} t^{i-n} u_i
		\]
		are well-defined.
		
		Since $ \sum_{i=n+1}^{n+k} t^{i-n} u_i $ is smooth, the $ k $-th time derivatives (for any $k$) and any spatial derivatives of
		\[
		b_0:=e^{\frac{\rho^2}{2t}} H_t - \sum_{i=0}^n t^{i-n} u_i
		\]
		also exist.
		
		Finally, by \eqref{cor32eq1}, we deduce that $t^{-1}b_0 \in S_\sm$.
	\end{proof}
	\subsection{Getzler's rescaling}\label{getzler}
	Our goal in this subsection is to study the regularity properties of the heat kernel $ H_t $ and the propagator $ P_t $ using Getzler’s rescaling technique.

	We introduce the following (Getzler's) rescaling $\delta_\epsilon$ associated with  K\"ahler normal frame $(U,\bfz,\bfe)$ and \eqref{tilde z coord} as in \eqref{rescaling}:
	\be
	\delta_\epsilon\big(a(\tilde{\bfz}, {\bar{\tilde{\bfz}}}, t, \bfw, \bar{\bfw}, d {\bar{\tilde{\bfz}}}, d \bar{\bfw}, dt)\big)=a (\tilde{\bfz}, \epsilon {\bar{\tilde{\bfz}}}, \epsilon t, \bfw, \bar{\bfw}, \epsilon d {\bar{\tilde{\bfz}}}, d \bar{\bfw}, \epsilon dt),
	\ee
	for any $a\in A\big(N(U,r);p\big)$.
	
	Note that by \Cref{HeatAsy}, $\sum_{i=0}^nt^{i-n}u_i\in A\big(N(U,r);p\big).$

	To analyze the action of $ \delta_{\epsilon} $ on the sum $\sum_{i=0}^n t^{i-n} u_i$, we consider the following decompositions associated with  $(U, \bfz, \bfe)$ and \eqref{tilde z coord}.
	
	First, it is clear that each $u_i$ admits the decomposition
	\[
	u_i = \sum_{\bfj, \bfl \in \{0,1\}^n} u_{i,\bfj,\bfl} \, d\btfz^{\bfj} \wedge d\bar{\bfw}^{\bfl},
	\]
	where $u_{i,\bfj,\bfl} \in \Gamma(E \boxtimes E)$. Here, for a multi-index $\bfi = (i_1, \dots, i_n) \in \{0,1\}^n$ such that $i_{j_1} = \cdots = i_{j_\ell} = 1$, we define 
	\[
	d\bfz^{\bfi} := dz_{j_1} \wedge \cdots \wedge dz_{j_\ell}.
	\]
	Further decomposing $u_{i,\bfj,\bfl}$, we obtain:
	\begin{align}
		\sum_{i=0}^n t^{i-n} u_i 
		&= \sum_{i=0}^n t^{i-n} \Big( 
		\sum_{m=0}^{n-i} \sum_{|\bfj| + |\bfk| = m} \sum_{\bfl} 
		b^i_{\bfk,\bfj,\bfl} \, \btfz^{\bfk} \, d\btfz^{\bfj} \wedge d\bar{\bfw}^{\bfl} \label{expansion1} \\
		&\qquad + \sum_{|\bfj| + |\bfk| = n - i + 1} \sum_{\bfl} 
		b^i_{\bfk,\bfj,\bfl} \, \btfz^{\bfk} \, d\btfz^{\bfj} \wedge d\bar{\bfw}^{\bfl} \Big), \label{expansion2}
	\end{align}
	with $\bfj, \bfl \in \{0,1\}^n$ and $\bfk \in \mathbb{Z}_{\geq 0}^n$, where
	\begin{enumerate}[(a)]
		\item if $|\bfk| + |\bfj| \leq n - i$, then $b^i_{\bfk,\bfj,\bfl} \in \Gamma(E \boxtimes E)$ is holomorphic in the $\tfz$ variable;
		\item if $|\bfk| + |\bfj| > n - i$, then $b^i_{\bfk,\bfj,\bfl} \in \Gamma(E \boxtimes E)$ (not necessarily holomorphic in $\tilde{\bfz}$).
	\end{enumerate}
	Here for $\bfj=(j_1,\cdots,j_n)\in\Z_{\geq0}^n\text{ or }\{0,1\}^n, |\bfj|:=j_1+\cdots+ j_n$.

	We can rearrange the above expansion as follows:
	\begin{prop}\label{prop33}
		We have  
		\[
		\sum_{i=0}^n t^{i-n} u_i = C_1+\sum_{l=-n}^0 C_l 
		\]
		where  
		\[
		C_l = \sum_{i=0}^n  t^{i-n}\sum_{|\bfj|+|\bfk|=l+n-i}\sum_{\bfl} b^i_{\bfk,\bfj,\bfl} \btfz^\bfk d\btfz^\bfj\wedge d{\bar\bfw}^\bfl,l\leq0
		\]  
		and 
		\[ C_1=\sum_{i=0}^n  t^{i-n}\sum_{|\bfj|+|\bfk|=n-i+1}\sum_{\bfl} b_{\bfk,\bfj,\bfl}^i\btfz^\bfk d \btfz^\bfj\wedge d{\bar{\bfw}}^\bfl.\]
	\end{prop}
	
	To prove \Cref{main0}, by \Cref{cor32}, we need to study $ \sum_{l=-n}^0C_l $ and $C_1$. By \cref{fiberwise regular}, we localize our study at $U\times\{p\}.$ We note that the following lemma, which would be proved in \cref{limit-L}.
	\begin{lem}\label{lem36}At $U\times\{p\}$, the limit \be\label{defn of L} L = \lim_{\epsilon \rightarrow 0} \epsilon \delta_\epsilon \left( \p_t+ \Delta_{\widetilde{E}} \otimes \mathrm{id} \right) \delta_\ep^{-1}\ee exists.

    More explicitly,
    \[
    L=\epsilon \delta_\epsilon \left( \p_t+ \Delta_{\widetilde{E}} \otimes \mathrm{id} \right) \delta_\ep^{-1}+O(\ep)
    \]
	\end{lem}
	We will show that
   \begin{prop}
       \label{cor39}
Let $q_t := e^{-|\tilde{\bfz}|^2/2t}$. Then, restricted to $U \times \{p\}$, 
\begin{enumerate}[(a)]
    \item Let $L$ denote the operator defined in \eqref{defn of L}. Then
    \[
        L\big(q_t C_l\big) = 0, \quad l \le 0.
    \]
    \item \label{thm48item2} We have 
    \[
        C_l = 0 \quad \text{for } l < 0.
    \]
\end{enumerate}
\end{prop}

	\begin{proof}
We first treat the case $l=-n$.  
By the construction above, $\delta_{\epsilon} C_l = C_l \epsilon^l$ and $\delta_{\epsilon} C_1 = O(\epsilon)$. Hence
\begin{equation}\label{rescaling Cl}
  \delta_{\epsilon}\!\left( \sum_{i=0}^n t^{i-n} u_i \right)
  = \sum_{l=-n}^0 C_l \epsilon^l + O(\epsilon).
\end{equation}

From the standard heat kernel expansion (see, e.g., \cite[Chap.~2]{berline2003heat}),
\begin{equation}\label{heat eq asym}
  \bigl( \partial_t + \Delta_{\widetilde{E}}\otimes\mathrm{id} \bigr)
  \!\left( e^{-\rho^2/2t}\!\sum_{i=0}^{n} t^{i-n} u_i \right)
  = e^{-\rho^2/2t}\bigl(\Delta_{\widetilde{E}}\otimes\mathrm{id}\bigr)u_n.
\end{equation}

By \eqref{Kahler distance}, on $U\times\{p\}$,
\begin{equation}\label{limit of metric}
  \delta_\epsilon\!\left(\frac{\rho^2}{2t}\right)
   = \frac{|\tilde{\bfz}|^2}{2t} + O(\epsilon).
\end{equation}

Restricting to $U\times\{p\}$ and using \Cref{lem36} together with
\eqref{rescaling Cl}–\eqref{limit of metric},
\begin{equation}\label{thm48eq1}
\begin{aligned}
 &\ \ \ \ L\!\left(q_t\!\sum_{l=-n}^0 C_l\epsilon^l\right)
 = L\!\left(\delta_{\epsilon} e^{-\rho^2/2t} + O(\epsilon)\right)
    \left(\delta_{\epsilon}\!\sum_{i=0}^n t^{i-n}u_i + O(\epsilon)\right) \\
 &= L(\delta_{\epsilon} e^{-\rho^2/2t})\,\delta_{\epsilon}\!\sum_{i=0}^n t^{i-n}u_i
    + O(\epsilon^{-n+1}) \\
 &= \epsilon\,\delta_{\epsilon}
    \bigl((\partial_t+\Delta_{\widetilde{E}}\otimes\mathrm{id})
    \delta_{\epsilon}^{-1}\bigr)
    \delta_{\epsilon}\!\left(e^{-\rho^2/2t}\!\sum_{i=0}^n t^{i-n}u_i\right)
    + O(\epsilon^{-n+1}) \\
 &= \epsilon\,\delta_{\epsilon}
    \!\left(e^{-\rho^2/2t}(\Delta_{\widetilde{E}}\otimes\mathrm{id})u_n\right)
    + O(\epsilon^{-n+1})
 = O(\epsilon^{-n+1}).
\end{aligned}
\end{equation}

Comparing powers of $\epsilon$ gives
\[
  L(q_t C_{-n})=0.
\]

To show \eqref{thm48item2} for $l=-n$, write
\[
  C_l=\sum_{i=0}^n t^{i-n}
   \sum_{|\bfj|+|\bfk|=l+n-i}\;
   \sum_{\bfl} b^i_{\bfk,\bfj,\bfl}\,
   \btfz^{\bfk} d\btfz^{\bfj}\wedge d\bar{\bfw}^{\bfl}.
\]
The coefficient of $t^{-n}$ is
\[
  C_l^{0}:=\sum_{|\bfj|+|\bfk|=l+n}\;\sum_{\bfl}
   b^0_{\bfk,\bfj,\bfl}\,
   \btfz^{\bfk} d\btfz^{\bfj}\wedge d\bar{\bfw}^{\bfl}.
\]

By \eqref{initial condition}, in the frame $(U,\bfz,\bfe)$,
\[
  u_0|_{\{p\}\times\{p\}}
  = \o(\bfz,\bfw)\, d^n(\bar{\bfz}-\bar{\bfw})|_{\{0\}\times\{0\}}.
\]

Thus $C_l^{0}|_{\{p\}\times\{p\}}=0$ for $l<0$, since then $|\bfj|\le n+l<n$.  
By \Cref{prop310} below, this implies $C_{-n}=0$.

With $C_{-n}=0$, \eqref{thm48eq1} improves to
\[
  L\!\left(q_t\sum_{l=-n+1}^0 C_l\epsilon^l\right)=O(\epsilon^{-n+2}).
\]

Repeating the argument yields $L(q_tC_{-n+1})=0$ and therefore $C_{-n+1}=0$.  
Proceeding inductively gives $C_l=0$ for all $l<0$ and $L(q_tC_l)=0$ for all $l\leq0$.

    \end{proof}
	
	The proposition below will be proved in \cref{proof of propositions}.
	\begin{prop}\label{prop310} Let $q_t:=e^{-\frac{|\tilde{\bfz}|^2}{2t}}$. Suppose $S\in\Omega^*\big((0,\infty);\Gamma(\tE\boxtimes\tE|_{U\times U})\big)$ admits an expansion of the form
		\[
		S = q_t \Big( \sum_{i=0}^{n} t^{i - n} v_i \Big), \quad v_i \in \Gamma(\widetilde{E} \boxtimes \widetilde{E}|_{U\times U}).
		\]
		Moreover, restricted to $U \times \{p\}$, $S$ is a solution to the equation $L S = 0$.
		Then, if $v_0 = 0$ on $\{p\}\times \{p\}$,  we must have $S \equiv 0$ on $U\times\{p\}$. Here, $L$ denotes the operator defined in \eqref{defn of L}. 
	\end{prop}

	By \Cref{cor39} and  \Cref{prop33}, restricted on $U\times\{p\}$,
	\be\label{no other Cl} \sum_{i=0}^{n}t^{i-n}u_{i}=C_{0}+C_1.\ee To prove \Cref{main0}, we need to deal with the terms $ C_0 $ and $ C_1 $. For this purpose, the following lemma is needed:
	\begin{lem}\label{lem311}
		For any multi-index $\bfk\in \Z_{\geq0}^n,\bfj\in\{0,1\}^n$, we have
		\begin{enumerate}[(1)]
			\item\label{itemlem3111} We have $\btfz^\bfk d\btfz^\bfj =t^{|\bfk|+|\bfj|}\bfy^\bfk d\bfy^\bfj+t^{|\bfk|+|\bfj|-1}\sum_{j=1}^ndt\wedge \iota_{\btz_j\frac{\p}{\p \bz_j}}\bfy^\bfk d\bfy^\bfj.$
			\item\label{itemlem3112} We have $\sum_{j=1}^{n}\iota_{\btz_j\frac{\p}{\p \bz_j}}\in \Op(S_{\sm}).$
		\end{enumerate}
	\end{lem}
	\begin{proof}
		We first prove (\ref{itemlem3112}). Note that $\sum_{i=1}^{n} \iota_{\btz_i \frac{\partial}{\partial \bz_i}}$ satisfies the Leibniz rule. It suffices to verify the following identities:
		\be\label{lem319eq1}
		\sum_{j=1}^{n} \iota_{\btz_j \frac{\partial}{\partial \bz_j}} dy_i = y_i, \quad 
		\sum_{j=1}^{n} \iota_{\btz_j \frac{\partial}{\partial \bz_j}} dt = 0,
		\ee
		\be\label{lem319eq2}
		\sum_{j=1}^{n} \iota_{\btz_j \frac{\partial}{\partial \bz_j}} t = 
		\sum_{j=1}^{n} \iota_{\btz_j \frac{\partial}{\partial \bz_j}} y_i = 0,
		\ee
		which are straightforward to verify.
		
		For (\ref{itemlem3111}), just note that
		\[
		d\bar{\tilde{z}}_i = y_i\, dt + t\, dy_i,
		\]  
		\[
		\sum_{j=1}^n \iota_{\btz_j\frac{\p}{\p \bz_j}}\, dy_i = y_i,
		\]
		and 
		\[
		y_i\, dt \wedge y_j\, dt = 0.
		\]

		\def\infact{1}
		\if\infact0
		In fact, (\ref{itemlem3111}) can be interpreted as the ``Taylor expansion" of  
		\[
		P(\bar{\bfz}, d\bar{\bfz}) = P(\bfy t, t d\bfy + \bfy dt)
		\]  
		around the ``point $ (t\bfy ,  t d\bfy) $." More precisely, we consider the expansion  
		\[
		P(t\bfy , t d \bfy + \bfy d t) = P_1(\bfy,t,d\bfy) + d t\wedge P_2(\bfy,t, d \bfy),
		\]
		where $P_1$ and $P_2$ are polynomials.
		
		Noting that  
		\[
		d t \wedge P_1(\bfy,t, d \bfy) = d t \wedge P( t\bfy, t d \bfy + \bfy d t) = d t \wedge P(t\bfy, t d \bfy),
		\]  
		we introduce the vector field  
		\[
		X = t \frac{\partial}{\partial t}+\sum_{i=1}^n \bar{\tilde{z}}_i \frac{\partial}{\partial \bar{\tilde{z}}_i}.
		\]  
		Since $ \iota_X d y_i = 0 $, we obtain  
		\[
		t P_1(\bfy,t, d \bfy) = \iota_X (d t \wedge P_1) = \iota_X (d t \wedge P) = t P(t\bfy , t d \bfy).
		\]  
		Thus,  
		\[
		P_1(\bfy, t, d \bfy) = P(t\bfy, t d \bfy).
		\]
		
		Next, using  
		\[
		0 = \iota_{\frac{\p}{\p t}} P(\bar{\tilde{\bfz}}, d \bar{\tilde{\bfz}}) = \iota_{\frac{\p}{\p t}}P_1(\bfy,t,d\bfy) + P_2(\bfy,t,d\bfy) - dt\wedge \iota_{\frac{\p}{\p t}}P_2(\bfy,t,d\bfy),
		\]  
		and multiplying by $ dt $, we obtain  
		\[
		d t\wedge \iota_{\frac{\partial}{\partial t}} P_1(\bfy,t,d\bfy) + d t \wedge P_2(\bfy,t,d\bfy) = 0.
		\]  
		That is,  
		\[
		d t\wedge \iota_{\frac{\partial}{\partial t}} P(t\bfy, t d\bfy) + d t \wedge P_2(\bfy,t,d\bfy) = 0.
		\]
		
		Consequently, we deduce  
		\[
		\begin{aligned}
			P(\bar{\tilde{\bfz}}, d \bar{\tilde{\bfz}}) &= P( t\bfy, t d \bfy) + d t \wedge P_2(\bfy,t,d\bfy) \\
			&= P(t\bfy , t d \bfy) - d t\wedge \iota_{\frac{\p}{\p t}} P(t\bfy , t d \bfy).
		\end{aligned}
		\]
		
		Noting that  
		\[
		\iota_X P(t\bfy , t d \bfy) = 0,
		\]  
		we conclude that  
		\[
		\iota_{\frac{\p}{\partial t}} P(t\bfy , t d \bfy) = -\frac{1}{t} \sum_{i=1}^n \iota_{\bar{\tilde{z}}_i \frac{\p}{\p \bar{\tilde{z}}_i}} P(t\bfy , t d \bfy).
		\]
		
		Thus, we have proved the first statement.
		\fi
	\end{proof}
	
	It follows from \Cref{lem311} that
	\begin{prop}\label{prop312}
		We have $$C_{0}=\Big(1+\sum_{j=1}^{n}\frac{dt}{t}\wedge \iota_{\btz_j\frac{\p}{\p \bz_j}}\Big)\tilde{C}_{0}=\Big(1+\sum_{j=1}^{n}{dt}\wedge y^j\iota_{\frac{\p}{\p \bz_j}}\Big)\tilde{C}_{0},$$
		where \be\label{tilde C0 smooth}\tilde{C}_{0}=\sum_{i=0}^n  \sum_{|\bfj|+|\bfk|=n-i}\sum_{\bfl} b^i_{\bfk,\bfj,\bfl}\bfy^\bfk d\bfy^\bfj\wedge d{\bar\bfw}^\bfl\in S_\sm.\ee
	\end{prop}
	\begin{proof}
		By \Cref{lem311} and note that $\iota_{\frac{\p}{\p \bz^i}}{d\bar\bfw}^\bfl=0$, we compute
		\begin{align*}
			\begin{split}
				C_0&=\sum_{i=0}^n  t^{i-n}\sum_{|\bfj|+|\bfk|=n-i}\sum_{\bfl} b^i_{\bfk,\bfj,\bfl} \btfz^\bfk d\btfz^\bfj\wedge d{\bar\bfw}^\bfl\\
				&=\sum_{i=0}^n  \sum_{|\bfj|+|\bfk|=n-i}\sum_{\bfl} b^i_{\bfk,\bfj,\bfl}\Big(\bfy^\bfk d\bfy^\bfj+\sum_{j=1}^n\frac{dt}{t}\wedge \iota_{\btz_j\frac{\p}{\p \bz_j}}\bfy^\bfk d\bfy^\bfj\Big)\wedge d{\bar\bfw}^\bfl\\
				&=\Big(1+\sum_{j=1}^n\frac{dt}{t}\wedge \iota_{\btz_j\frac{\p}{\p \bz_j}}\Big)\sum_{i=0}^n  \sum_{|\bfj|+|\bfk|=n-i}\sum_{\bfl} b^i_{\bfk,\bfj,\bfl}\bfy^\bfk d\bfy^\bfj\wedge d{\bar\bfw}^\bfl.
			\end{split}
		\end{align*}
		
	\end{proof}
	
	As a result,
	\begin{prop}\label{prop313}
		There exists $a_1,a_2\in S_\sm$, such that
		restricted on $U\times\{p\}$, 
		\be\ba\label{expan Ht C0} e^{\frac{\rho^{2}}{2t}}H_{t}&={C}_{0}+ta_{1}+dta_{2}\\
		&=\Big(1+\sum_{j=1}^{n}{dt}\wedge y^j\iota_{\frac{\p}{\p \bz_j}}\Big)\tilde{C}_{0}+ta_{1}+dta_{2}.\ea\ee
	\end{prop}
	\begin{proof}
		First, by \Cref{cor32} and \eqref{no other Cl}, on $U\times\{p\},$
		\[
		H_{t} = e^{-\frac{\rho^2}{2t}} \Big( \sum_{i=0}^{n} t^{i-n} u_{i} + b_{0} \Big) = e^{-\frac{\rho^2}{2t}} \left( C_{0} + C_1 + b_{0} \right).
		\]
		Next, proceeding as in \cref{prop312}, we have
		\begin{equation*}
			C_1=\Big(1+\sum_{j=1}^{n}\frac{dt}{t}\wedge \iota_{\btz_j\frac{\p}{\p \bz_i}}\Big)\tilde{C}_{1} \quad \text{with\quad $ t^{-1}\tilde{C}_1\in S_\sm.$}
		\end{equation*}
		By \eqref{itemlem3112} in \Cref{lem311} and the fact that $ \frac{b_0}{t} \in {S}_{\sm} $ (see \Cref{cor32}), we can set:
		\[
		a_1 = t^{-1}\tilde{C}_1+ t^{-1} b_0
		\]
		and
		\[
		a_2 = \sum_{j=1}^{n} \iota_{\btz_j \frac{\partial}{\partial \bz_j}} t^{-1}\tilde{C}_1.
		\]
		
		Together with \Cref{prop312}, the result follows.
	\end{proof}
	
	Next, we would like to deal with the  term $dt\wedge(\bar{\partial}_{\widetilde{E}}^*\otimes \mathrm{id})H_t$ in $P_t$.

	\begin{lem}\label{lem316}
		Let ${\mathcal{M}}=e^{\frac{\rho^2}{2t}}(\bar{\partial}_{\widetilde{E}}^*\otimes \mathrm{id})e^{-\frac{\rho^2}{2t}}$ and
		$${{\mathcal{N}}}:=\sum_{i=1}^n{y_i}\iota_{\frac{\p}{\p \bz_i}}-2\sum_{i=1}^n\iota_{\frac{\p}{\p \bz_i}}\frac{\p}{\p z_i}\in \Op(S_\sm).$$
		Then for any $S\in S_\sm$,  there exists $S'\in S_\sm$ such that restricted on $U\times\{p\}$,\be\label{difference}{{\mathcal{N}}}S-{\mathcal{M}}S=S'.\ee
	\end{lem}
	\begin{proof}
		Since $t\iota_{\frac{\p}{\p\bz_i}}dy_j=\delta_{ij}$, we have
		\be\label{lem315}
		t\,\iota_{\frac{\p}{\p\bz_i}}\in \Op(S_\sm).
		\ee
		Let ${\mathcal{M}}=e^{\frac{\rho^2}{2t}}(\bar{\partial}_{\widetilde{E}}^*\otimes \mathrm{id})e^{-\frac{\rho^2}{2t}}$,   note that $\bar{\partial}_{\widetilde{E}}^*\otimes \mathrm{id}=-(g^{-1})^{i\bar{j}}\iota_{\frac{\p}{\p \bz_j}}\nabla^{\tE}_{\frac{\p}{\p z_i}}$,
		
		\be\label{lem314}\ba{\mathcal{M}}&=\sum_{i,j}\Big(\frac{1}{2t}\frac{\p \rho^2}{\p z_i}(g^{-1})^{i\bar{j}}\iota_{\frac{\p}{\p \bz_j}}-(g^{-1})^{i\bar{j}}\iota_{\frac{\p}{\p \bz_j}}\nabla^{\tE}_{\frac{\p}{\p z_i}}\Big)\\
        \ea\ee 
		It follows from the properties of the K\"ahler normal frame $(U, \bfz, \bfe)$ (see \Cref{local-coord}), together with \eqref{lem314} and \eqref{lem315}, that for any $T \in S_\sm$, there exist $T_i, T_i', T_{ij} \in S_\sm$ such that, when restricted to $U \times \{p\}$, the following identities hold:
		\[
		t^{-2} \Big( \frac{\partial}{\partial z_i} \rho^2(\bfz, \bar{\bfz}, \bfw, \bar{\bfw}) - \frac{\partial}{\partial z_i} |\tilde{\bfz}|^2 \Big) T 
		= \Big(t^{-2}  \frac{\partial}{\partial z_i} \rho^2(\bfz, \bar{\bfz}, \bfw, \bar{\bfw}) - t^{-1}y_i\Big)  T 
		= T_i,
		\]
		\[
		t^{-1} \left( (g^{-1})^{i \bar{j}} - 2 \delta^{ij} \right) T = T_{ij},
		\]
		and
		\[
		t^{-1} \Big( \nabla^{\widetilde{E}}_{\frac{\partial}{\partial z_i}} - \frac{\partial}{\partial z_i} \Big) T 
		= t^{-1} \Big( h^{-1} \frac{\partial h}{\partial z_i} \Big) T = T_i'.
		\]

	\end{proof}
	\begin{prop}\label{prop325} There exists $a_1',a_2'\in S_\sm$ such that
		restricted on $U\times \{p\}$,
		\be\label{prop325eq2}P_t=e^{-\frac{\rho^2}{2t}}\bigg(\tilde{C}_{0}+2dt\wedge\Big(\sum_{i=1}^{n}\iota_{\frac{\p}{\p\bz_i}}\frac{\p}{\p z_i}\tilde{C}_{0}\Big)+ta_{1}'+dta_{2}'\bigg).\ee
	\end{prop}
	\begin{proof}

		We calculate, restricted on $U\times \{p\}$,
		\begin{align*}
			&\quad dt \wedge e^{\frac{\rho^2}{2t}}(\bar{\partial}_\tE^* \otimes \mathrm{id})H_t \\
			&= dt \wedge \Big(\sum_{i=1}^n {y_i} \iota_{\frac{\partial}{\partial \bar{z}_i}} -2 \sum_{i=1}^n \iota_{\frac{\partial}{\partial \bz_i}} \frac{\partial}{\partial z_i}\Big)\tilde{C}_0 \mod (dt \wedge {S}_{\sm}) \\
			&=C_0 - \tilde{C}_0 - 2dt \wedge \Big(\sum_{i=1}^n \iota_{\frac{\partial}{\partial \bz_i}} \frac{\partial}{\partial z_i}\Big)\tilde{C}_0 \mod (dt \wedge {S}_{\sm}),
		\end{align*}
		where the first equality follows from the fact that $ (dt)^2 = 0 $, along with \eqref{expan Ht C0} and \Cref{lem316}; the second equality follows from \Cref{prop312}.
		
		Note that $P_t=-dt\wedge(\bar{\p}_{\tE}^*\otimes\mathrm{id})H_t+H_t$, by \Cref{prop313} and the computation above, \eqref{prop325eq2} follows.
	\end{proof}
	Note that $\tilde{C}_0\in S_\sm$, to prove \Cref{main0}, we still need the following proposition, which will be proved in \cref{proof of propositions}.
	\begin{prop}\label{prop318}
		Restricted on $U\times \{p\}$,
		\[ \sum_{i=1}^n \iota_{\frac{\partial}{\partial\bz_i}} \frac{\partial}{\partial z_i}\tilde{C}_0 = 0. \]\end{prop} 
	
	\subsubsection{Proof of {\Cref{main0}} and {\Cref{propagator has divisorial type singularities}}}\label{proof of technical theorem}
	
	By \Cref{fiberwise regular}, to prove \Cref{main0}, it suffices to show that for any point $ p \in M $, and any K\"ahler normal frame $ (U, \bfz, \bfe) $ at $ p $, there exists $ C \in S_\sm $ such that on $ U \times \{p\} $,
	\begin{equation} \label{local-check}
		e^{\frac{\rho^2}{2t}} P_t = C.
	\end{equation}
	
	The identity \eqref{local-check} follows directly from \Cref{prop325} and \Cref{prop318}. Therefore, \Cref{main0} is proved.
	
	It follows from \Cref{main0} that $e^{\frac{\rho^2}{2t}}P_t\in S_\sm. $ Next, by $(\ref{distance function on arbitrary})$, we can show that if $C \in S_\sm(U\times U)$, then $e^{\frac{\rho^2}{2t}}\partial_{z_{i}}\big(e^{-\frac{\rho^2}{2t}}C\big)$ and $e^{\frac{\rho^2}{2t}}\partial_{w_{i}}\big(e^{-\frac{\rho^2}{2t}}C\big)$ are also regular expressions, thus \Cref{propagator has divisorial type singularities} follows from \Cref{propagator singularities and heat} and \Cref{lem44}.

	\def\proofof{1}
	\if\proofof0
	\begin{proof}[Proof of Theorem \ref{main0}]\ \\
		Recall that the operators ${\mathcal{M}},{\tilde{\mathcal{N}}}$
		and ${\tilde{\mathcal{M}}}$ are described in \Cref{lem316}.
		
		Since $S_\sm$ is independent of the choice of holomorphic normal frame, by examining \eqref{prop325eq1}, it suffices to prove that on $U \times U$
		\begin{align*}
			&C_l = 0, \quad l < -1,\\ 
			-dt\wedge & {\tilde{\mathcal{M}}}C_0+\frac{1}{2}C_0-\frac{1}{2}\tilde{C}_0=0
		\end{align*}
		which is equivalent to showing that for any $p' \in U$, when restricted to $U \times \{p'\}$,
		\begin{align}
			&C_l = 0, \quad l < -1, \label{Cl vanishes}\\ 
			-dt\wedge & {\tilde{\mathcal{M}}}C_0+\frac{1}{2}C_0-\frac{1}{2}\tilde{C}_0=0. \label{C0 vanishes}
		\end{align}
		
		Let $(U', \bfz', \bfe')$ be a K\"ahler normal frame at $p' \in U$, with $U' \subset U$. Applying \eqref{expan Ht C0} in the version corresponding to the frame $(U', \bfz', \bfe')$, there exists $S' \in S_\sm$ such that, when restricted to $U' \times \{p'\}$,
		\be\label{exp Ht U prime}
		e^{\frac{\rho^2}{2t}} H_t = C_0' + S'.
		\ee
		Here, $C_0'$ denotes the term associated with the frame $(U', \bfz', \bfe')$.
		
		Let $z_i' = f_i(\bfz)$ denote the coordinate transformation from $\bfz$ to $\bfz'$. Then, by Taylor expansion,
		\[
		z_i' = f_i(\bfw) + \sum_j f_{ij}(\bfw)(z_j - w_j) + \sum_{j,k} f_{ijk}(\bfz, \bfw)(z_j - w_j)(z_k - w_k)
		\]
		for some holomorphic $f_{ij},f_{ijk}.$
		\def\bff{{\mathbf{f}}}
		Let $\bff = (f_1, \dots, f_n)$. Then, for $\bfw' = \bff(\bfw)$,
		\begin{equation} \label{yi prime}
			y_i' := \frac{\bz_i' - \bw_i'}{t} = \sum_j \bar{f}_{ij}(\bar{\bfw}) y_j + t \sum_{j,k} \bar{f}_{ijk}(\bar{\bfz}, \bar{\bfw}) y_j y_k,
		\end{equation}
		\begin{equation}\ba \label{d tilde zi prime}
			d\tilde{z}_i' &:= \sum_j \left( \bar{f}_{ij}(\bar{\bfw}) d\btz_j + d\bar{f}_{ij}(\bar{\bfw}) \btz_j \right) \\
			&+ t\sum_{j,k} \left( d\bar{f}_{ijk}(\bar{\bfz}, \bar{\bfw}) y_j \btz_k + \bar{f}_{ijk}(\bar{\bfz}, \bar{\bfw}) y_k d\btz_j + \bar{f}_{ijk}(\bar{\bfz}, \bar{\bfw}) y_j d\btz_k \right).
			\ea\end{equation}
		
		and
		\begin{equation} \label{d yi prime}
			\begin{aligned}
				dy_i' &= \sum_j \left( \bar{f}_{ij}(\bar{\bfw}) dy_j + d\bar{f}_{ij}(\bar{\bfw}) y_j \right) \\
				&\quad + t \sum_{j,k} \left( d\bar{f}_{ijk}(\bar{\bfz}, \bar{\bfw}) y_j y_k + \bar{f}_{ijk}(\bar{\bfz}, \bar{\bfw}) y_k dy_j + \bar{f}_{ijk}(\bar{\bfz}, \bar{\bfw}) y_j dy_k \right).
			\end{aligned}
		\end{equation}
		
		Thus, using \eqref{yi prime}-\eqref{d yi prime}, uniqueness of Taylor expansions, one can show that on $U' \times U'$,
		\begin{equation} \label{C0 minus C0 prime}
			\tilde{C}_0 - \tilde{C}_0' \in t S_\sm \quad \text{and}\quad {C}_0 - {C}_0' \in t S_\sm \oplus dt\wedge S_\sm
		\end{equation}
		
		Thus, on $U' \times \{p'\}$, by \eqref{C0 minus C0 prime} and \eqref{exp Ht U prime},
		\begin{equation} \label{Cn vanishing1}
			\lim_{\ep \to 0} \ep^{n} \delta_\ep e^{\frac{\rho^2}{2t}} H_t = 0,
		\end{equation}
		where $\delta_\ep$ is associated with $(U, \bfz, \bfe)$.
		
		On the other hand, by \eqref{expan Ht Cl}, on $U' \times \{p'\}$,
		\begin{equation} \label{Cn vanishing2}
			\lim_{\ep \to 0} \ep^{n} \delta_\ep e^{\frac{\rho^2}{2t}} H_t = C_{-n},
		\end{equation}
		where $\delta_\ep$ is again the one associated with $(U, \bfz, \bfe)$.
		
		It then follows from \eqref{Cn vanishing1} and \eqref{Cn vanishing2} that $C_{-n} = 0$ on $U' \times \{p'\}$. Since $C_{-n}$ is analytic, we conclude that $C_{-n} = 0$ on $U \times \{p'\}$.
		
		This establishes \eqref{Cl vanishes} for $l = -n$. By repeating the same argument inductively, we eventually establish \eqref{Cl vanishes} for all $l \leq -1$.
		
		The proof of \eqref{C0 vanishes} is more subtle.

		By applying the version of \Cref{prop318} and \eqref{prop325eq2} associated with the frame $(U', \bfz', \bfe')$—noting also that the operator ${\mathcal{M}}$ is coordinate-independent—and using \eqref{C0 minus C0 prime}, we obtain the existence of $T', T_1, T_2 \in S_\sm$ such that, when restricted to $U' \times \{p'\}$, the following holds:
		\[
		e^{\frac{\rho^2}{2}} P_t - \frac{1}{2}\tilde{C}_0 - t T' = e^{\frac{\rho^2}{2}} P_t - \frac{1}{2}\tilde{C}_0' = t T_1 + dt \wedge T_2.
		\]
		
		Therefore, restricted to $U' \times \{p'\}$,
		\begin{equation} \label{C0 vanishes1}
			\delta_\ep \left( e^{\frac{\rho^2}{2}} P_t - \frac{1}{2}\tilde{C}_0 \right) = \delta_\ep \left( e^{\frac{\rho^2}{2}} P_t - \frac{1}{2}\tilde{C}_0' \right) = 0,
		\end{equation}
		where $\delta_\ep$ is associated with $(U, \bfz, \bfe)$.
		
		It can be checked that $\delta_\ep (dt\wedge {\mathcal{M}}C_0)=dt\wedge {\tilde{\mathcal{M}}} C_0$, $\delta_\ep(C_0)=C_0$ and $\delta_\ep \tilde{C}_0 = \tilde{C}_0$. Thus, by \eqref{Cl vanishes} and \eqref{prop325eq1}, restricted to $U' \times \{p'\}$,
		\begin{equation} \label{C0 vanishes2}
			\delta_\ep \left( e^{\frac{\rho^2}{2}} P_t - \frac{1}{2}\tilde{C}_0 \right) = -dt\wedge {\tilde{\mathcal{M}}}C_0+\frac{1}{2}C_0-\frac{1}{2}\tilde{C}_0==0.
		\end{equation}

		It then follows from \eqref{C0 vanishes1} and \eqref{C0 vanishes2} that \eqref{C0 vanishes} holds on $U' \times \{p'\}$, which implies that \eqref{C0 vanishes} holds on $U \times \{p'\}$, since all objects in \eqref{C0 vanishes} are analytic. 
	\end{proof}\fi
	\subsubsection{Proof of \Cref{lem36}}\label{limit-L}
	\def\bj{{\bar{j}}}
	The proof of \Cref{lem36} follows from the standard argument of Getzler's rescaling \cite{getzler1986short,berline2003heat}; we include it here for completeness.
	
	Recall that we have fixed a local K\"ahler normal chart $(U, \mathbf{z}, \bfe)$ for the bundle $E \to M$, with $\mathbf{z} = (z_i)_{i=1}^n$ centered at some point $p \in M$. Let $(U, \mathbf{w}, \mathbf{f})$ denote the same coordinate chart on another copy of $M$.

	We use the following coordinates:
	\be
	\begin{cases}
		\tz_i = z_i - w_i, \\
		\btz_i = \bar{z}_i - \bar{w}_i.
	\end{cases}
	\ee
	and introduce the rescaling $\delta_\ep$ associated with these coordinates.
	
	In this subsection, we study the limit over the slice $U \times \{p\}$:
	\[
	\lim_{\ep \to 0} \ep \deep \left( \p_t+ \Delta_{\tE} \otimes \mathrm{id} \right) \deep^{-1}.
	\]
	We adopt the following abbreviations:
	\begin{itemize}
		\item For vector fields: 
		\[
		\p_i := \frac{\p}{\p z_i},\text{ and }\p_\bj := \frac{\p}{\p \bz_j}.
		\]
		\item For any connection $\nabla$ on a complex vector bundle $H \to M$:
		\[
		\nabla_i : = \nabla_{\p_i}, \quad \nabla_{\bj} := \nabla_{\p_{\bj}}.
		\]
		\item If $F$ is the curvature of $\nabla$, then:
		\[
		F_{i\bj} := F \left(\p_i,\p_{\bj} \right) \in \End(H).
		\]
	\end{itemize}

	First, we recall the local expression for the Laplacian:
	\[
	\Delta_{\tE} = -\sum_{i,j}\Big(g^{i\bj} \nabla^{\tE}_{\p_i} \nabla^{\tE}_{\p_{\bj}} + g^{i\bj} \nabla^{\tE}_{\nabla^{TM}_{\p_i} \p_{\bj}} - \sum_ld\bz_i \, \iota_{\p_{\bar{l}}} F^{E \otimes K}_{j\bar{i}} g^{j\bar{l}}\Big),
	\]
	where $K$ denotes the canonical bundle, and $F^{E \otimes K}$ is the curvature of $E \otimes K$, induced by the connections $\nabla^{TM}$ and $\nabla^E$.
	
	\def\wework{1}
	\if\wework0
	We work in the following local coordinates near the diagonal ($i = 1, 2, \dots, n$):
	\begin{equation*}
		\begin{cases}
			\tz_i = z_i - w_i, \\
			\btz_i = \bz_i - \bw_i, \\
			\tw_i = w_i, \\
			\btw_i = \bw_i.
		\end{cases}
	\end{equation*}
	
	\def\tfw{{\tilde{\bfw}}}
	\def\btfw{{\bar{\tilde{\bfw}}}}
	These coordinates satisfy:
	\be
	\begin{cases}
		d\tz_i = dz_i - dw_i, \\
		d\btz_i = d{\bz_i} - d\bw_i, \\
		d\tw_i = dw_i, \\
		d\btw_i = d\bw_i;
	\end{cases}
	\ee
	and
	\be
	\begin{cases}
		\frac{\p}{\p \tz_i} = \frac{\p}{\p z_i}, \\
		\frac{\p}{\p \btz_i} = \frac{\p}{\p \bz_i}, \\
		\frac{\p}{\p \tw_i} - \frac{\p}{\p \tz_i} = \frac{\p}{\p w_i}, \\
		\frac{\p}{\p \btw_i} - \frac{\p}{\p \btz_i} = \frac{\p}{\p \bw_i}.
	\end{cases}
	\ee
	\fi

	The following identities from K\"ahler geometry will be useful:
	\begin{align*}
		\nabla^{TM}_i\p_{\bj} &= 0, \quad \nabla^{TM}_\bj\p_i = 0, \\
		\nabla^{TM}_i \p_j &= \sum_kw_{ij}^k \p_k, \quad \nabla^{TM}_{\bar{i}}\p_{\bj} = \sum_kw_{\bar{i}\bj}^{\bar{k}} \p_{\bar{k}}, \\
		\nabla^{E}_i e_a &= \sum_b\Gamma_{ia}^b e_b, \quad \nabla_{\bj}^E e_a = 0,
	\end{align*}
	where
	\be
	\ba
	w_{ij}^k = \sum_lg^{k\bar{l}}\p_i g_{j\bar{l}}, \quad 
	w_{\bar{i}\bj}^{\bar{k}} = \sum_l g^{l\bar{k}} \p_{\bar{i}} g_{l\bar{j}}, \quad 
	\Gamma_{ia}^b = \sum_c h^{b\bar{c}}\p_i h_{a\bar{c}}.
	\ea
	\ee
	Note that $\nabla^E_{\bj} =\p_{\bj}$ and $\nabla^{TM}_{i} d\bz_j = 0$. Then we obtain:
	\begin{align*}
		\nabla^{\tE}_{\bar{j}} &=\p_{\bj} + \sum_{i,k}w_{\bar{i}\bar{j}}^{\bar{k}} d\bz_i \wedge \iota_{\p_{\bar{k}}} \\
		&=\p_{\bj} + \sum_{i,k}w_{\bar{i}\bar{j}}^{\bar{k}} (d\btz_i + d\bar{w}_i) \wedge \iota_{\p_{\bar{k}}}, \\
		\nabla^{\tE}_i &=\p_i + \Gamma_i, \\
		\sum_{i,j,l} d\bz_i \, \iota_{\p_{\bar{l}}} F_{j\bar{i}}^{E \otimes K} g^{j\bar{l}} &= \sum_{i,j,l} (d\btz_i + d\bw_i) \, \iota_{\p_{\bar{l}}} F_{j\bar{i}}^{E \otimes K} g^{j\bar{l}},
	\end{align*}
	where $\Gamma_i := \sum_{a,b} \Gamma_{ia}^b e^a \otimes e_b \in \End(E)$, and $\{e^a\}$ is the dual frame of $\{e_a\}$.
	
	Note that on $U \times \{p\}$, where $\tilde{\bfz} = \bfz$, $\bar{\tilde{\bfz}} = \bar{\bfz}$, we have:
	\[
	\Gamma_{ia}^b(\bfz, \bar\bfz) = \Gamma_{ia}^b(\tfz, \bar\tfz) = \sum_j f_{ija}^b(\tfz, \bar\tfz) \btz_j,
	\]
	for some analytic $f_{ija}^b$. As a result,
	\[
	\delta_\ep \Gamma_{ia}^b = O(\ep).
	\]
	Thus, at $U \times \{p\}$,
	\begin{align*}
		\ep \deep \nabla_{\bj}^{\tE} \deep^{-1} &= \bp_\bj + \sum_{i,k}w_{\bar{i}\bar{j}}^{\bar{k}}(\tilde\bfz,0) d\bw_i \iota_{\bp_{\bar{k}}} + O(\ep),\\ 
		\deep \nabla^{\tE}_i \deep^{-1} &= \p_i + \deep \Gamma_i = \p_i + O(\ep),\\
		\deep g^{i\bj} &= 2\delta^{ij} + O(\ep),\\
		\ep \deep \Big(\sum_{i,j,l}(d\btz_i + d\bw_i)\iota_{\bp_{\bar{l}}}F_{j\bar{i}}^{E\otimes K} g^{j\bar{l}} \Big) &= \sum_{i,j} d\bw_i \iota_{\bp_{\bj}} F_{j\bar{i}}^{E\otimes K} + O(\ep).
	\end{align*}
	Here, for an analytic function $f$ on $C^\infty(\C^n)$, $f(\bfz, 0)$ denotes its holomorphic component.
	
	Therefore, at $U \times \{p\}$,
	\begin{equation}
		\ep \deep \Delta_{\tE} \deep^{-1} = -2\sum_{i}\p_i\Big(\bp_{\bar{i}} + \sum_{j,k}w_{\bar{j}\bar{i}}^{\bar{k}}({\bfz}, 0) d\bw_j \iota_{\bp_{\bar{k}}} \Big) + 2\sum_{i,j} d\bw_i \iota_{\bp_\bj} F_{j\bar{i}}^{E\otimes K}({\bfz}, 0) + O(\ep).
	\end{equation}
	\def\md{{\mathrm{md}}}
	We define
	\[
	\Delta^{\md} := -2\sum_{i}\p_i\Big(\bp_{\bar{i}} +\sum_{j,k} w_{\bar{j}\bar{i}}^{\bar{k}}({\bfz}, 0) d\bw_j \iota_{\bp_{\bar{k}}} \Big) + 2\sum_{i,j} d\bw_i \iota_{\bp_\bj} F_{j\bar{i}}^{E\otimes K}({\bfz}, 0).
	\]
	
	Then the operator $L$ appearing in \Cref{lem36} is given by
	\[
	L = \p_t + \Delta^{\md} \otimes \mathrm{id}.
	\]
	
	
	
	
	\subsubsection{Proof of \Cref{prop310} and \Cref{prop318}}\label{proof of propositions}

	Let $\nabla^\md$ be the operator given by
	\[
	\nabla^\md_{\bar{j}} = \p_{\bar{j}} +\sum_{l,k} w_{\bar{l} \bar{j}}^{\bar{k}}(\tilde \bfz, 0) d\bw_l \iota_{\p_{\bar{k}}}, \quad \text{and} \quad \nabla^{\md}_j = \p_j.
	\]
	Then the Laplacian $\Delta^{\md}$ is the connection Laplacian associated with $\nabla^{\md}$:
	\[
	\Delta^{\md} = -2\sum_i \nabla^{\md}_i \nabla^{\md}_{\bar{i}} +2 \sum_{i,j} d\bar{w}_i \iota_{\bar{\p}_{\bar j}} F_{j\bar{i}}^{E\otimes K}(\bfz, 0).
	\]
	
	Suppose $S \in \Omega^*((0,\infty); \Gamma(\tE \times \tE))$ satisfies the condition in \Cref{prop310}. That is, it solves the heat equation at $U \times \{p\}$:
	\begin{equation}
		\big(\p_t + \Delta^\md\big) S = 0,
	\end{equation}
	and admits the following asymptotic expansion at $U\times \{p\}$:
	\be
	S = e^{-\frac{|\tilde{\bfz}|^2} {2t}} \Big(\sum_{k=0}^n t^{k-n} v_k \Big), \quad v_k \in \Gamma(\widetilde{E} \boxtimes \widetilde{E}).
	\ee
	We compute, at $U \times \{p\}$,
	\be\ba\label{sec5eq1}
	\Big(\p_t +\Delta^{\md}  \Big) S
	=t^{-n} e^{-\frac{|\tilde{\bfz}|^2} {2t}} \bigg( \p_t +\Delta^{\md}  + \sum_i \Big( \frac{\bar{\tilde{z}}_i}{t} \nabla^{\md}_{\bar{i}} + \frac{\tilde{z}_i}{t} \nabla^{\md}_i \Big) \bigg) \sum_{k=0}^n t^{k} v_k.
	\ea\ee
	Let $r = |\tfz|$. Then \eqref{sec5eq1} implies, at $U \times \{p\}$,
	\be\ba
	\Big(\p_t +\Delta^{\md} \Big) S
	=t^{-n} e^{-\frac{|\tilde{\bfz}|^2} {2t}} \left(\p_t + \Delta^{\md} + \frac{\nabla_{r \nabla r}}{t} \right) \sum_{k=0}^n t^{k} v_k.
	\ea\ee
	Hence, we obtain the recursive system:
	\be\label{sec3eq13}
	\begin{cases}
		\nabla^{\md}_{r \nabla r} v_0 = 0, \\
		\nabla^{\md}_{r \nabla r} v_{k+1} + (k+1) v_{k+1} + \Delta^{\md} v_k = 0, \quad k \geq 0.
	\end{cases}
	\ee
	
	\begin{proof}[Proof of \Cref{prop310}]
		We focus on the analysis at $U \times \{p\}$.
		
		It follows from standard ODE theory that if $v_0 = 0$ at $\tfz = 0$ and satisfies \eqref{sec3eq13}, then $v_0 \equiv 0$.
		
		Next, since $v_0$ already vanishes, the second equation in \eqref{sec3eq13} yields
		\[
		\nabla^\md_{\nabla r} (r v_1) = 0.
		\]
		By the same reasoning, since $rv_1=0$ at $\tilde{\bfz}=0$, we have $r v_1 \equiv 0$, and hence $v_1 \equiv 0$.
		
		Proceeding inductively, suppose $v_k \equiv 0$. Then from \eqref{sec3eq13}, we obtain
		\[
		\nabla^\md_{\nabla r} (r^{k+1} v_{k+1}) = 0.
		\]
		Since $r^{k+1} v_{k+1}$ vanishes at $\tfz = 0$, it follows that $v_{k+1} \equiv 0$. Therefore, $v_k \equiv 0$ for all $k \geq 0$, completing the proof.
		
	\end{proof}
	
	\begin{lem}\label{lem51}
		At $U \times \{p\}$, the following commutation relations hold:
		\begin{enumerate}[(a)]
			\item $\left[\sum_l \iota_{\p_{\bar{l}}} \partial_l, \nabla^\md_{{j}} \right] = 0.$
			\item $\left[\sum_l \iota_{\p_{\bar{l}}} \partial_l, \nabla^\md_{\bar{j}} \right] = 0.$
			\item $\left[\sum_l \iota_{\partial_{\bar{l}}} \partial_l, \Delta^\md \right] = 0.$
		\end{enumerate}
	\end{lem}
	\begin{proof}
		We again work at $U \times \{p\}$. 
		
		Part (a) is straightforward. 
		
		To prove (b), observe that the curvature term satisfies
		\[
		R_{\bar{i} l \bar{j}}^{\bar{k}} = \partial_l w_{\bar{i} \bar{j}}^{\bar{k}}, \quad \text{and} \quad R_{\bar{i} l \bar{j}}^{\bar{k}} = R_{\bar{i} k \bar{j}}^{\bar{l}}.
		\]
		Then
		\begin{equation}\label{lem51eq1}
			\begin{aligned}
				&\quad\Big[\sum_l \iota_{\partial_{\bar{l}}} \partial_l, \sum_{i,k}w_{\bar{i} \bar{j}}^{\bar{k}}(\tfz, 0) d \bar{w}_i \iota_{{\partial}_{\bar{k}}} \Big]
				= -\sum_{i,k,l}\partial_l w_{\bar{i} \bar{j}}^{\bar{k}}(\tfz,0) d \bar{w}_i \iota_{\partial_{\bar{l}}} \iota_{{\partial}_{\bar{k}}} \\
				&= -\sum_{i,k,l}R_{\bar{i} l \bar{j}}^{\bar{k}}(\tfz,0) d \bar{w}_i \iota_{\p_{\bar{l}}} \iota_{\partial_{\bar{k}}} = -\frac{1}{2} \sum_{i,k,l}R_{\bar{i} l \bar{j}}^{\bar{k}}(\tfz,0) d \bar{w}_i [\iota_{\p_{\bar{l}}}, \iota_{\partial_{\bar{k}}}] = 0.
			\end{aligned}
		\end{equation}
		
		Since
		\[
		\nabla^\md_{\bar{j}} = \bar{\partial}_{\bar{j}} + \sum_{i,k}w_{\bar{i} \bar{j}}^{\bar{k}}(\tfz,0) d \bar{w}_i \iota_{{\partial}_{\bar{k}}},
		\]
		equation \eqref{lem51eq1} implies part (b).
		
		For (c), recall that
		\[
		\sum_l \iota_{\p_{\bar{l}}} \partial_l = -\bar{\partial}^{\md,*}, \quad \text{and} \quad \Delta^\md = [\bar{\partial}^{\md}, \bar{\partial}^{\md,*}].
		\]
		Thus,
		\[
		[\bar{\partial}^{\md,*}, \Delta^\md] = 0,
		\]
		which proves (c).
	\end{proof}
	
	Consider
	\[
	\tilde{v}_k := \sum_l \iota_{\p_{\bar{l}}} \p_l v_k, \quad k \geq 0.
	\]
	
	By \Cref{lem51}, the $\tilde{v}_k$ satisfy:
	\begin{equation}\label{prop52eq1}
		\begin{cases}
			\nabla^\md_{r \nabla r} \tilde{v}_0 + \tilde{v}_0 = 0, \\
			\nabla^\md_{r \nabla r} \tilde{v}_{k+1} + (k+2) \tilde{v}_{k+1} + \Delta^\md \tilde{v}_k = 0, \quad k \geq 0,
		\end{cases}
	\end{equation}
	where $r = |\tilde{\bfz}|$.
	
	\begin{proof}[Proof of \Cref{prop318}]
		As in the proof of \Cref{prop310}, equations \eqref{prop52eq1} imply that $\tilde{v}_k = 0$ for all $k \geq 0$ at $U \times \{p\}$.
	\end{proof}

	\appendix
	\def\bfi{{\mathbf{i}}}
	\def\bfj{{\mathbf{j}}}
	\def\bfk{{\mathbf{k}}}
	\section{Existence of Canonical Local Coordinates}\label{local-coord}
	
	We introduce suitable coordinate systems to simplify our discussion, specifically K\"ahler normal coordinates on the manifold and holomorphic local frames on the vector bundle.

	It is standard that
	\begin{prop}\label{Kahler normal}
		Let $M$ be an $n$-dimensional real analytic K\"ahler manifold, and let $p \in M$. Then there exists a holomorphic coordinate system $\mathbf{z} = (z_1, \ldots, z_n)$ in a neighborhood $U$ of $p$ such that $\mathbf{z}(p) = 0$ and 
		\[
		\left(\frac{\p^kg_{i\bar{j}}}{\p z_{i_1}\cdots\p z_{i_k}} \right)(p) = 0, \quad \text{for } k \geq 1,
		\]
		and $g_{i\bar{j}}(p) = \half \delta_{ij}$,  
		where $g$ denotes the Riemannian metric on $M$ and $g_{i\bar{j}} := g\left( \frac{\partial}{\partial z_i}, \frac{\partial}{\partial \bar{z}_j} \right)$. The pair $(U, \mathbf{z})$ is called the \textbf{K\"ahler normal chart at $p$}, and $\mathbf{z}$ is called a K\"ahler coordinate system at $p$, which is unique up to unitary transformations.
		
	\end{prop}
	\begin{proof}
		The proof can be found in \cite{Kontsevich1994-1995,ruan1996canonical}.
	\end{proof}
	
	We have the following corollary for K\"ahler normal coordinates:
	
	\begin{cor}
		Let $M$ be an $n$-dimensional real analytic K\"ahler manifold. Let $p \in M$, and let $(U, \mathbf{z})$ be a K\"ahler normal chart at $p$. Consider another copy of $M$, and let $(U, \mathbf{w})$ denote the corresponding K\"ahler normal chart at $p$ on this second copy.
		
		Let $\rho^2 : M \times M \to \mathbb{R}$ denote the square of the distance function. With respect to the charts above, we write it as $\rho^2(\mathbf{z}, \bar{\mathbf{z}}, \mathbf{w}, \bar{\mathbf{w}})$. Then:
		
		\begin{enumerate}[(1)]
			\item On $U \times \{p\}$, we have
			\begin{equation} \label{Kahler distance}
				\rho^2(\mathbf{z}, \bar{\mathbf{z}}, 0, 0)
				= \sum_{i=1}^{n} z_i \bar{z}_i 
				+ \sum_{i,j,k,l=1}^{n} z_i z_j \bar{z}_k \bar{z}_l\, f_{ij\bar{k}\bar{l}}(\mathbf{z}, \bar{\mathbf{z}}),
			\end{equation}
			where $f_{ij\bar{k}\bar{l}}$ are real analytic functions on $U$.
			
			\item In general, we have
			\begin{equation} \label{distance function on arbitrary}
				\rho^2(\mathbf{z}, \bar{\mathbf{z}}, \mathbf{w}, \bar{\mathbf{w}})
				= \sum_{i=1}^{n} (z_i - w_i)(\bar{z}_i - \bar{w}_i)
				+ \sum_{i,j=1}^{n} (z_i - w_i)(\bar{z}_j - \bar{w}_j)\, f_{i\bar{j}}(\mathbf{z}, \bar{\mathbf{z}}, \mathbf{w}, \bar{\mathbf{w}}),
			\end{equation}
			where $f_{i\bar{j}}$ are real analytic functions on $U \times U$ satisfying $f_{i\bar{j}}(0,0,0,0) = 0$.
		\end{enumerate}
		
	\end{cor}
	\begin{proof}
		We recall the Eikonal equation on K\"ahler manifolds (see \cite[Example 3.2]{doi:10.1142/S0129055X24500478}):
		\be
		\label{Eikonal}
		\sum_{i,j=1}^n\partial_{i}\big(\rho^2(\bfz,\bar\bfz,0,0)\big)g^{i\bar{j}}(\bfz,\bar\bfz)\partial_{\bar{j}}\big(\rho^2(\bfz,\bar\bfz,0,0)\big)=2\rho^2(\bfz,\bar\bfz,0,0).
		\ee
		By inserting the Taylor expansion of $\rho^2(\bfz,\bar\bfz,0,0)$ into $(\ref{Eikonal})$ and using the properties of K\"ahler normal coordinates, we can prove the first assertion inductively.
		
		For the second assertion, we notice that $(\ref{Kahler distance})$ implies the following identities:
		\be\label{vanishing derivatives}
		\begin{cases}
			\frac{\partial}{\partial z_{i_k}} \cdots \frac{\partial}{\partial z_{i_1}} \rho^2(0,0,0,0) = 0, \quad\text{for } k \geq 1 \\
			\frac{\partial}{\partial \bz_{i_k}} \cdots \frac{\partial}{\partial \bz_{i_1}} \rho^2(0,0,0,0) = 0, \quad\text{for } k\geq 1
		\end{cases} 
		\ee
		Since $(\ref{vanishing derivatives})$ is invariant under arbitrary holomorphic coordinate transformations, it is true at any point of the diagonal in any holomorphic coordinate chart. Therefore, by using the natural coordinates on $U\times U$, we can get the expansion $(\ref{distance function on arbitrary})$.
	\end{proof}
	
	\begin{prop}
		Let $M$ be an $n$-dimensional real analytic complex manifold, and let $E \to M$ be a real analytic Hermitian holomorphic vector bundle. Let $p \in M$, and let $(U, \mathbf{z})$ be a holomorphic chart centered at $p$, i.e., $\mathbf{z}(p) = 0$, such that $E|_U \to U$ is trivial.
		
		Then there exists a local holomorphic frame $\mathbf{e} = (e_a)_{a=1}^{\mathrm{rank}(E)}$ of $E$ over $U$ such that
		\[
		\left( \frac{\partial^k h_{a\bar{b}}}{\partial z_{i_1} \cdots \partial z_{i_k}} \right)(p) = 0 \quad \text{for all } k \geq 1,
		\]
		and
		\[
		h_{a\bar{b}}(p) = \delta_{ab},
		\]
		where $h_{a\bar{b}} := h(e_a, e_b)$.
		
		We refer to $(U, \mathbf{z}, \mathbf{e})$ as a \textbf{holomorphic normal frame} at $p$. For each $(U,\bfz)$, the frame $\bfe$ is unique up to unitary transformations. Moreover, if $(U, \mathbf{z})$ is a K\"ahler normal chart at $p$, then we say that $(U, \mathbf{z}, \mathbf{e})$ is a \textbf{K\"ahler normal frame} at $p$.

	\end{prop}
	
	\begin{proof}
		First, for each holomorphic coordinate system $\bfz$ such that $\bfz(p) = 0$, choose an arbitrary holomorphic frame $\bfe'=(e_a')$ such that $h_{a\bar{b}}(p) = \delta_{ab}$. To establish the proposition, we aim to construct a holomorphic matrix $G_a^b(\mathbf{z})$ such that  
		\begin{equation} \label{eq1}
			G_a^b(\mathbf{z})\, h_{b\bar{c}}(\mathbf{z}, \bar{\mathbf{z}})\, \overline{G_d^c(\mathbf{z})} = \delta_{ad} + \sum_{\mathbf{j} \in \mathbb{Z}_{>0}^n} \bar{\mathbf{z}}^{\mathbf{j}} f_{ad, \mathbf{j}}(\mathbf{z}, \bar{\mathbf{z}})
		\end{equation}
		for some smooth $f_{ab,\bfj}$,
		where for any $\mathbf{j} = (j_1, \ldots, j_n) \in \mathbb{Z}_{>0}^n$, we set $\bar{\mathbf{z}}^{\mathbf{j}} := \bar{z}_1^{j_1} \cdots \bar{z}_n^{j_n}$.  
		With this matrix, setting $e_a = G_a^b e'_b$ and $\bfe=(e_a)$, it follows from \eqref{eq1} that $(\bfz, \bfe)$ forms a holomorphic normal frame.
		
		To determine $G_a^b$, we analytically continue \eqref{eq1} to a neighborhood of $(p, \bar{p}) \in M \times \bar{M}$ and evaluate it at $(\mathbf{z}, 0)$, obtaining
		\[
		G_a^b(\mathbf{z})\, h_{b\bar{c}}(\mathbf{z}, 0)\, \overline{G_d^c(0)} = \delta_{ad}.
		\]
		By setting 
		\[
		G_a^b(\mathbf{z}) = \left(h^{-1}(\mathbf{z}, 0)\right)^{b\bar{c}} \delta_{ca},
		\]
		we verify that $G_a^b(\mathbf{z})$ satisfies equation \eqref{eq1}.
	\end{proof}

	\section{Cauchy Principal Value}
	\label{section Cauchy principal value}
	In this appendix, we will introduce the Cauchy principal value for the reader's convenience. We will mainly follow \cite{herrera1971residues}, with extra attention to logarithmic terms.
	
	Let's first consider the local case.
	
	\subsection*{Local case}
	
	For any $\bfz=(z_1,\cdots,z_{m+n})\in \C^{m+n}$, $\bfi=(i_1,\cdots,i_m)\in\Z^m$, we set $$\bfz^\bfi=z_1^{i_1}\cdots z_m^{i_m}$$
	and $$ 
	\ln(\bfz,\bfi)=\big(\ln(|z_1|^2)\big)^{i_1}\cdots \big(\ln(|z_m|^2)\big)^{i_m}.
	$$

	\begin{defn}
		Given a differential form $\alpha \in \Omega^{*, *}\left(\left(\mathbb{C}^*\right)^m \times \mathbb{C}^n\right)$. We say $\alpha$ is a differential form on $\mathbb{C}^{m+n}$ with \textbf{divisorial type singularities of $(\bfi,\bfk)$-type} along the principal divisor: $z_{1} z_2 \cdots z_m$, if there exist $\beta \in \Omega^{*,*}\left(\mathbb{C}^{m+n}\right)$, $\bfi \in\Z_{\geq0}^m$ and $\bfk\in\Z_{\geq 0}^m$,
		such that
		$$
		\alpha=\bfz^{-\bfi}\ln(\bfz,\bfk)\beta.
		$$
		If $\bfk=0$, we say $\a$ is logarithm-free.
	\end{defn}
	
	We will prove in this appendix that
	\begin{thm}\label{thm-a2}
		Given the following data:
		\begin{enumerate}[(a)]
			\item   A compactly supported $(m+n,m+n)$ differential form $\alpha$ on $\mathbb{C}^{m+n}$ with divisorial type singularities of $(\bfi,\bfk  )$ along $z_1 \cdots z_m$.
			\item A holomorphic function $\bfz^{\bfj}$, where $\bfj \in\mathbb{Z}_{>0}^m$.
			\item A nowhere vanishing smooth function $f$ on a neighborhood of support of $\a$.
		\end{enumerate}
		Then the following limit exists:
		\be\label{thm-a-eq}\lim_{\delta\to0}\int_{|f\bfz^\bfj|>\delta}\a.\ee
		Moreover, the limit is independent of the choice of $\bfj$ and smooth function $f$.
	\end{thm}
	\begin{defn}
		With the same assumptions as Theorem $\ref{thm-a2}$, we call the limit
		\[
		\lim_{\delta\to0}\int_{|f\bfz^\bfj|>\delta}\a
		\]
		the integral of $\alpha$ over $\mathbb{C}^{m+n}$ in the sense of Cauchy principal value. We will use \[
		\dashint_{\mathbb{C}^{m+n}}\alpha
		\]
		to denote it.
	\end{defn}
	By Theorem $\ref{thm-a2}$ $(c)$, we see that the Cauchy principal value only depends on the support of the divisor, and not on the multiplicity or the choice of the holomorphic function $\bfz^\bfj$ representing it. This feature allows us to define the Cauchy principal value on general complex manifolds.
	
	Before we get into the global case, let us prove Theorem $\ref{thm-a2}$.
	
	Let \be\label{defn of B} B:=\{\bfz\in \C^{m+n}:\sup_{1\leq i\leq m+n}|z_i|<1\}.\ee

	Let $\mathbf{k} \in \Z_{\ge 0}^m$, with components $(k_1, \dots, k_m)$,  
and denote 
\[
|\mathbf{k}| := k_1 + \dots + k_m.
\]

	For each $\bfj\in \Z_{>0}^m$, let $S^\bfj_\delta:=\big\{\bfz\in B:|\bfz^\bfj|=\delta\big\}$ and $B_\delta^\bfj:=\big\{\bfz\in B:|\bfz^\bfj|\in(\delta,1)\big\}$.
	
	\def\hi{{\hat{i}}}
	For each $l\in \{1,2,\cdots,m+n\}$, set $\bfz_{\hi}=(z_1,\cdots,z_{l-1},z_{l+1},\cdots,z_{m+n})\in \C^{m+n-1}$ and for each ${l}\in \{1,\cdots,m\}$, $\bfj_{\hat{l}}:=(j_1,\cdots,j_{l-1},j_{l+1},\cdots,j_m).$
	
	Similar to \cite[Lemma 6.4]{herrera1971residues}, we have
	\begin{lem}\label{lema3}
		Let $\phi \in C\big((\C^*)^{m} \times \C^n\big)$ be bounded.  
Then for any $\bfk \in \Z_{\ge 0}^m$, $\bfj \in \Z_{>0}^m$, and any 
$i \in \{1,2,\dots,m\}$, 
		\be\ba\label{lema3eq1}
		\lim_{\delta\to0}\int_{S^\bfj_\delta} \phi(\bfz,\bar\bfz)\ln(\bfz,\bfk) &dz_i\wedge d\bfz_{\hi}\wedge d\bar\bfz_{\hi}=0,\\
		\lim_{\delta\to0}\int_{S^\bfj_\delta} \phi(\bfz,\bar\bfz) \ln(\bfz,\bfk)&d\bar{z}_i\wedge d\bfz_{\hi}\wedge d\bar\bfz_{\hi}=0.
		\ea\ee
		Here $d\bfz_{\hi}:=dz_1\wedge\cdots \wedge dz_{i-1}\wedge dz_{i+1}\cdots \wedge dz_{m+n}$ and $d\bar{\bfz}_{\hi}:=\overline{d\bfz_{\hi}}.$
	\end{lem} 
	\begin{proof}

		Without loss of generality, we may assume that $i = 1$. Recall that by the definition of $S_\delta^{\mathbf{j}}$, we have $S_\delta^{\mathbf{j}} \subset B$. Consider the following parametrization of $S_\delta^{\mathbf{j}}$:
		\[
		(\theta, \bfz_{\hat{1}}) \mapsto \Bigg( \bigg( \frac{\delta}{\prod_{s=2}^{m} \rho_s^{j_s}} \bigg)^{1/j_1} e^{i\theta}, \bfz_{\hat{1}} \Bigg),
		\]
		where $z_s = \rho_s e^{i\theta_s}$. With this parametrization, one verifies that the integrals appearing in \eqref{lema3eq1} are bounded, up to a constant, by
		\begin{equation} \label{lemaeq2}
			\delta^{1/j_1} \left( \ln \delta \right)^{|\mathbf{k}|} \cdot \int_{E^\bfj_\delta} \prod_{s=2}^{m} \rho_s^{1 - j_s / j_1} \, d\rho_s,
		\end{equation}
		where for each $\mathbf{r}=(r_1,\cdots,r_m)\in\R_{>0}^{m},$
		\[
		E^{\mathbf{r}}_\delta := \Big\{ (\rho_2, \dots, \rho_m) \in\R_{\geq0}^m: \delta \leq \prod_{s=2}^{m} \rho_s^{r_s} \leq 1,\quad \sup_{s\in\{2,\cdots,m\}} \rho_s \leq 1 \Big\}.
		\]
		
		We now prove by induction that for any $r \geq 0$ and $\mathbf{r}=(r_1, \dots, r_m) \in \mathbb{R}_{>0}^m$, the following limit holds:
		\begin{equation} \label{lemaeq3}
			\lim_{\delta \to 0} \delta^{1/r_1} (\ln \delta)^r \cdot \int_{E^{\mathbf{r}}_\delta} \prod_{s=2}^{m} \rho_s^{1 - r_s / r_1} \, d\rho_s = 0.
		\end{equation} 
		The base case $m=2$ follows from a direct computation.
		
		Let
		\[
		E'_\delta := \Big\{ (\rho_3, \dots, \rho_m)\in\R_{\geq0}^m : \delta \leq \prod_{s=3}^{m} \rho_s^{r_s} \leq 1,\quad \sup_{s\in\{3,\cdots,m\}} \rho_s \leq 1 \Big\}.
		\]
		We distinguish two cases:
		\begin{enumerate}[(a)]
			\item Case $r_2 \leq 2r_1$:\\
			By Fubini's theorem,
			\[
			\begin{aligned}
				&\delta^{1/r_1} (\ln \delta)^r \cdot \int_{E^{\mathbf{r}}_\delta} \prod_{s=2}^{m} \rho_s^{1 - r_s / r_1} \, d\rho_s \\
				&\leq \delta^{1/r_1} (\ln \delta)^r \cdot \int_{E'_\delta} \prod_{s=3}^{m} \rho_s^{1 - r_s / r_1} \, d\rho_s \int_{\delta^{1/r_2}}^1 \rho_2^{1 - r_2 / r_1} \, d\rho_2 \\
				&\leq C \delta^{1/r_1} (\ln \delta)^{r+1} \cdot \int_{E'_\delta} \prod_{s=3}^{m} \rho_s^{1 - r_s / r_1} \, d\rho_s.
			\end{aligned}
			\]
			By the induction hypothesis, the right-hand side tends to zero as $\delta \to 0$, hence the limit in \eqref{lemaeq3} holds in this case.
			
			\item Case $r_2 > 2r_1$:\\
			Again applying Fubini's theorem,
			\[
			\begin{aligned}
				&\delta^{1/r_1} (\ln \delta)^r \cdot \int_{E^{\mathrm{r}}_\delta} \prod_{s=2}^{m} \rho_s^{1 - r_s / r_1} \, d\rho_s \\
				&\leq \delta^{1/r_1} (\ln \delta)^r \cdot \int_{E'_\delta} \prod_{s=3}^{m} \rho_s^{1 - r_s / r_1} \, d\rho_s \int_{\left( \delta / \prod_{s=3}^{m} \rho_s^{r_s} \right)^{1/r_2}}^{+\infty} \rho_2^{1 - r_2 / r_1} \, d\rho_2 \\
				&\leq C_1 \delta^{2/r_2} (\ln \delta)^r \cdot \int_{E'_\delta} \prod_{s=3}^{m} \rho_s^{1 - 2r_s / r_2} \, d\rho_s 
			\end{aligned}
			\]
			By the induction hypothesis, the right-hand side vanish as $\delta \to 0$. Therefore, the desired limit in \eqref{lemaeq3} holds in this case as well.
		\end{enumerate}
		Now we complete the induction and establish \eqref{lemaeq3}.
		
	\end{proof}
	Similar to \cite[Lemma 6.3]{herrera1971residues}, we have
	\begin{lem}\label{lema4}
		Let $g \in C^\infty\left(\mathbb{C}^{m+n}\right)$ be independent of $z_l$, and $\bfi,\bfk \in \Z_{\geq0}^m,\bfj\in\Z_{>0}^m$. Then
		for all $r, s \in \Z_{\geq0}$ such that $r+s<i_l$, and each $\delta>0$\begin{enumerate}[(1)]
			\item $\int_{S_\delta^\bfj} z_l^r \bar{z}_l^s \bfz^{-\bfi}\ln(\bfz,\bfk) g d \bar{z}_l \wedge d \bfz_{\hat{l}} \wedge d \bar{\bfz}_{\hat{l}}=0$;
            \item If $s\geq1$, then  $\int_{S_\delta^\bfj} z_l^r \bar{z}_l^s \bfz^{-\bfi}\ln(\bfz,\bfk) g d {z}_l \wedge d \bfz_{\hat{l}} \wedge d \bar{\bfz}_{\hat{l}}=0$.
		\end{enumerate}
		Here $d\bfz:=dz_1\wedge\cdots \wedge dz_{m+n}$ and $d\bar\bfz:=\overline{d\bfz}.$

	\end{lem}
	\begin{proof}
		Let $B_{\delta,\hat{l}}^{\bfj}:=\Big\{\bfz_{\hat{l}}\in \C^{m+n-1}:\sup\{|z_{a}|:a\neq l\}<1, \delta<|\bfz_{\hat{l}}^{\bfj_{\hat{l}}}|<1\Big\}.$

		

		To compute (1), one integrates over $B_{\delta,\hat{l}}^\bfj$ the differential form $\bfz_{\hat{l}}^{-\bfi_{\hat{l}}}\ln(\bfz_{\hat{l}},\bfk_{\hat{l}}) \cdot g d \bfz_{\hat{l}}$ $\wedge d \bar{\bfz}_{\hat{l}}$ times the integrals
		$$
		\int_{\left\{z_l:|z_l^{j_l}|=\delta /|\bfz_{\hat{l}}^\bfj|\right\}} z_l^{r-i_l} \bar{z}_l^s\big(\ln(|z_l|^2)\big)^{k_l} d \bar{z}_l, \quad \bfz_{\hat{l}}\in B_{\delta,\hat{l}}^\bfj,
		$$
		which are all zero if $s\geq0$, since in such case $r-i_l-s-1\leq r-i_l+s-1 <0$.
		We deduce that the integral in (1) is zero.

        For (2), just note that if $s\geq1$, $r-i_l-s+1\leq r-i_1+s-2+1<0$.
	\end{proof}

	\begin{lem}[Lemma 6.2 in \cite{herrera1971residues} ]\label{lem-a5}
		Let $k\in C^\infty(\C^{m+n})$, and $\bfi\in\Z_{\geq0}^m$. There exists a decomposition
		$$
		\begin{aligned}
			& k(\bfz,\bar{\bfz})=\sum_{l=1}^m\sum_{r+s<i_l} z_{l}^r \bar{z}_{l}^s g_{r, s}^l(\bfz_{\hat{l}}, \bar{\bfz}_{\hat{l}})+K(\bfz,\bar{\bfz}), \\
			& K(\bfz,\bar{\bfz})=\sum_{\bfk+\bfj=\bfi}\bfz^\bfj \cdot \bar{\bfz}^\bfk \cdot K_{\bfj, \bfk}(\bfz,\bar{\bfz})
		\end{aligned}
		$$
		such that $g_{r, s}^l$ and $K_{\bfj, \bfk}$ are smooth functions.
	\end{lem}

	Proceeding as in \cite[Proposition 6.5]{herrera1971residues}, using \Cref{lema3}-\Cref{lem-a5}, we have:
	\begin{prop}\label{prop-a5}
		For any $\bfi,\bfk \in \Z_{\geq0}^m,\bfj\in\Z_{>0}^m$ and $g \in C^\infty\left(\mathbb{C}^{m+n}\right)$, 
	we have
        \be\label{prop-a-eq3}
		\lim _{\delta \rightarrow 0} \int_{S_\delta^\bfj} \bfz^{-\bfi} \ln(\bfz,\bfk)g d \bar{z}_l \wedge d \bfz_{\hat{l}} \wedge d \bar{\bfz}_{\hat{l}}=0
		\ee
and
        \be\label{prop-a-eq4}
		\lim _{\delta \rightarrow 0} \int_{S_\delta^\bfj} \bfz^{-\bfi} \ln(\bfz,\bfk)\bar{z}_lg d {z}_l \wedge d \bfz_{\hat{l}} \wedge d \bar{\bfz}_{\hat{l}}=0.
		\ee

	\end{prop} 
Thus, we obtain the following.

\begin{cor}\label{cor-a5}
	For any $\bfi,\bfk \in \Z_{\geq 0}^m$, $\bfj \in \Z_{>0}^m$, and $g \in C_c^\infty(\C^{m+n})$, we have
	\be\label{cor-a-eq3}
		\lim_{\delta \to 0} 
		\int_{|\bfz^\bfj|=\delta} 
		\bfz^{-\bfi}\,\ln(\bfz,\bfk)\, g\;
		d\bar z_l \wedge d\bfz_{\hat l} \wedge d\bar{\bfz}_{\hat l}
		=0,
	\ee
	and
	\be\label{cor-a-eq4}
		\lim_{\delta \to 0} 
		\int_{|\bfz^\bfj|=\delta} 
		\bfz^{-\bfi}\,\ln(\bfz,\bfk)\,\bar z_l g\;
		d z_l \wedge d\bfz_{\hat l} \wedge d\bar{\bfz}_{\hat l}
		=0.
	\ee
\end{cor}

\begin{proof}
	We may assume $g \in C_c^\infty(B)$.  
	Indeed, if not, there exists $N>1$ such that $\mathrm{supp}(g) \subset B_N:=\{\bfz\in\C^{m+n}:\sup_{1\le j\le m+n} |z_j|<N\}$.  
	In this case, one verifies that \Cref{prop-a5} continues to hold when the set $B$  is replaced by $B_N$.

	Now assume $g \in C_c^\infty(B)$. Then
	\[
		\int_{|\bfz^\bfj|=\delta} 
		\bfz^{-\bfi}\,\ln(\bfz,\bfk)\, g\;
		d\bar z_l \wedge d\bfz_{\hat l} \wedge d\bar{\bfz}_{\hat l}
		=
		\int_{S_\delta} 
		\bfz^{-\bfi}\,\ln(\bfz,\bfk)\, g\;
		d\bar z_l \wedge d\bfz_{\hat l} \wedge d\bar{\bfz}_{\hat l},
	\]
	so \eqref{cor-a-eq3} follows directly.  
	The argument for \eqref{cor-a-eq4} is analogous.
\end{proof}

	\begin{prop}\label{prop-a7}
		Let $f$ be a real nowhere vanishing smooth function on $\C^{m+n}$, and let 
$\bfj \in \Z_{>0}^m$, $\bfi,\bfk \in \Z_{\geq 0}^m$. 
Then, for any $g \in C_c^\infty(\C^{m+n})$, we have
\be\label{prop-a6-eq2}
	\lim_{\delta \to 0}
	\int_{|f\bfz^\bfj|=\delta}
	\bfz^{-\bfi}\,\ln(\bfz,\bfk)\, g\;
	d\bar z_1 \wedge d\bfz_{\hat 1} \wedge d\bar{\bfz}_{\hat 1}
	=0.
\ee

	\end{prop}
	\begin{proof}
		
		\def\supp{\mathrm{supp}}
		
		We may assume that $|f|>1$ on a precompact open neighborhood $K$ containing the support of $g$.  

Let $h$ be a real smooth function such that $h^{j_1}=f$ on $K$, and set
\[
\Lambda := 3\sup_{\bfz\in K} |\nabla h|(\bfz).
\]
For $1\le j\le m$, consider
\[
U_j := \{\bfz\in\C^{m+n} : |z_j| < \Lambda^{-1}\}.
\]

Consider a map $\lambda_j : U_j \cap K \to \C^{m+n}$ by
\[
\zeta_j = z_j\, h(\bfz,\bar\bfz), \qquad 
\zeta_l = z_l \quad (l\neq j).
\]
Then the map 
\[
    \lambda_j : U_j \cap K \;\longrightarrow\; \lambda_j(U_j \cap K)
\]
is a diffeomorphism. (Indeed, by our construction of \(U_j\), the Jacobian determinant of \(\lambda_j\) is nonzero on \(U_j \).) Suppose its inverse $\mu_j : \lambda_j(U_j\cap K) \to U_j\cap K$ is given by
\[
z_j = \zeta_j\, \tilde{h}_j(\boldsymbol\zeta,\bar{\boldsymbol\zeta}),
\qquad 
z_l = \zeta_l \quad (l\neq j),
\]
where $\tilde{h}_j$ is a smooth nowhere vanishing function on $\lambda_j(U_j\cap K)$.

Let
\[
\delta_0 := \Lambda^{-|\bfj|}, \qquad
K_{\delta_0} := K \cap \{\bfz\in\C^{m+n} : |f\bfz^\bfj| \le \delta_0\}.
\]
Then
\be\label{contain}
	K_{\delta_0} \subset \bigcup_{j=1}^m U_j.
\ee

Let $\{\phi_j\}_{j=1}^m$ be a partition of unity subordinate to $\{U_j\}_{j=1}^m$.  
By \eqref{contain}, if $\delta < \delta_0$ then
\[
\int_{|f\bfz^\bfj|=\delta} 
	\bfz^{-\bfi}\ln(\bfz,\bfk) g\;
	d\bar z_1 \wedge d\bfz_{\hat 1} \wedge d\bar{\bfz}_{\hat 1}
=
\sum_{j=1}^m
\int_{|f\bfz^\bfj|=\delta} 
	\bfz^{-\bfi}\ln(\bfz,\bfk)\, \phi_j g\;
	d\bar z_1 \wedge d\bfz_{\hat 1} \wedge d\bar{\bfz}_{\hat 1}.
\]

It therefore suffices to show that
\[
\lim_{\delta\to 0}
\int_{|f\bfz^\bfj|=\delta}
	\bfz^{-\bfi}\ln(\bfz,\bfk) (\phi_j g)\;
	d\bar z_1 \wedge d\bfz_{\hat 1} \wedge d\bar{\bfz}_{\hat 1}
=0.
\]
We prove only the case $j=1$.

Changing variables via $\bfz=\mu_1(\boldsymbol{\zeta},\bar{\boldsymbol{\zeta}})$ yields
\be\label{change-coord}
\begin{aligned}
&\int_{|f\bfz^\bfj|=\delta}
	\bfz^{-\bfi}\ln(\bfz,\bfk)\, (\phi_1 g)(\bfz,\bar\bfz)\;
	d\bar z_1 \wedge d\bfz_{\hat 1} \wedge d\bar{\bfz}_{\hat 1}
\\
&= 
\int_{|\boldsymbol\zeta^\bfj|=\delta}
	{\boldsymbol\zeta}^{-\bfi}
	\sum_{l=0}^{i_1}
	\binom{i_1}{l}
	(\ln|\tilde h|^2)^l
	(\ln|\zeta_1|^2)^{i_1-l}
	\ln({\boldsymbol\zeta}_{\hat 1},\bfk_{\hat 1})
	\,\tilde{k}_1(\boldsymbol\zeta,\bar{\boldsymbol\zeta})
	\ d\bar\zeta_1 \wedge d\boldsymbol\zeta_{\hat 1} \wedge d\bar{\boldsymbol\zeta}_{\hat 1}
\\
&\quad+
\int_{|\boldsymbol\zeta^\bfj|=\delta}
	{\boldsymbol\zeta}^{-\bfi}
	\sum_{l=0}^{i_1}
	\binom{i_1}{l}
	(\ln|\tilde h|^2)^l
	(\ln|\zeta_1|^2)^{i_1-l}
	\ln({\boldsymbol\zeta}_{\hat 1},\bfk_{\hat 1})
	\,\bar\zeta_1\,\tilde{k}_2(\boldsymbol\zeta,\bar{\boldsymbol\zeta})
	\ d\zeta_1 \wedge d\boldsymbol\zeta_{\hat 1} \wedge d\bar{\boldsymbol\zeta}_{\hat 1},
\end{aligned}
\ee
where
\[
\tilde{k}_1(\boldsymbol\zeta,\bar{\boldsymbol\zeta})
=
\tilde{h}_1^{-i_1}(\phi_1 g)(\mu,\bar\mu)
\left(
	\bar{\tilde h}_1 + \bar\zeta_1 \frac{\partial\bar{\tilde h}_1}{\partial \bar\zeta_1}
\right),
\qquad
\tilde{k}_2(\boldsymbol\zeta,\bar{\boldsymbol\zeta})
=
\tilde{h}_1^{-i_1}(\phi_1 g)(\mu,\bar\mu)
\left(
	\frac{\partial\bar{\tilde h}_1}{\partial \zeta_1}
\right).
\]

By \Cref{cor-a5} and \eqref{change-coord}, we conclude that
\[
\lim_{\delta\to 0}
\int_{|f\bfz^\bfj|=\delta}
	\bfz^{-\bfi}\ln(\bfz,\bfk)\,\phi_1 g\;
	d\bar z_1 \wedge d\bfz_{\hat 1} \wedge d\bar{\bfz}_{\hat 1}
=0.
\]

	\end{proof}
	\def\bfl{{\mathbf{l}}}
	\begin{prop}[\Cref{thm-a2}]\label{prop-a8}
Let $\bfi,\bfk \in \Z_{\ge 0}^m$ and $\bfj \in \Z_{>0}^m$.  
Let $f \in C^\infty(\C^{m+n})$ be a real nowhere–vanishing function.  
Then the limit
\be\label{prop-a8-eq0}
	\lim_{\delta \to 0}
	\int_{|f\bfz^\bfj|>\delta}
		\bfz^{-\bfi}\, \ln(\bfz,\bfk)\, g\;
		d\bfz \wedge d\bar{\bfz}
\ee
exists and is independent of the choice of $f$ and $\bfj$, for every
$g \in C_c^\infty(\C^{m+n})$.
\end{prop}

\begin{proof}
If $0\le i_l\le 1$ for $1\leq l \leq m$, then 
\[
\bfz^{-\bfi}\,\ln(\bfz,\bfk)\,g
\]
is Lebesgue integrable on $\C^{m+n}$, so the limit \eqref{prop-a8-eq0} exists and is independent of $f$ and $\bfj$.

We introduce a lexicographic order on $\Z_{\ge0}\times\Z_{\ge0}$ by
\[
(p_1,p_2)\le (q_1,q_2)
\quad\Longleftrightarrow\quad
\big(p_1<q_1\big)\ \text{ or }\ 
\big(p_1=q_1 \text{ and } p_2\le q_2\big).
\]
Assume inductively that the proposition holds for all pairs
\((|\bfi|,|\bfk|)<(p_1,p_2)\).  
We prove it for \((|\bfi|,|\bfk|)=(p_1,p_2)\).

The case $0\le i_l\le1$ is already handled above, so we may assume
that $i_1\ge2$.

Consider
\[
b \;=\; \frac{1}{1-i_1}\,\bfz^{-\bfi} z_1\,
d\bar z_1\wedge d\bfz_{\hat1}\wedge d\bar{\bfz}_{\hat1}.
\]

We compute:
\[
\begin{aligned}
&\quad\int_{|f\bfz^\bfj|>\delta}
\bfz^{-\bfi}\ln(\bfz,\bfk)\,g\;d\bfz\wedge d\bar{\bfz}
\\[0.3em]
&=\int_{|f\bfz^\bfj|>\delta}
d\left(\ln(\bfz,\bfk)\,g\, b\right)
-
\int_{|f\bfz^\bfj|>\delta}
(\partial_{z_1}\ln(\bfz,\bfk))\,g\;dz_1\wedge b
-
\int_{|f\bfz^\bfj|>\delta}
\ln(\bfz,\bfk)\,\partial_{z_1}g\;dz_1\wedge b
\\
&=\int_{|f\bfz^\bfj|=\delta}
\ln(\bfz,\bfk)\,g\, b
-\int_{|f\bfz^\bfj|>\delta}
(\partial_{z_1}\ln(\bfz,\bfk))\,g\;dz_1\wedge b
\;-\;
\int_{|f\bfz^\bfj|>\delta}
\ln(\bfz,\bfk)\,\partial_{z_1}g\;dz_1\wedge b
\\[0.3em]
&=: I_1(\delta)+I_2(\delta)+I_3(\delta).
\end{aligned}
\]

By \Cref{prop-a7}, 
\[
\lim_{\delta\to0}I_1(\delta)=0.
\]

For $I_2(\delta)$, the integrand has $(\bfi,\bfk')$–type divisorial 
singularities with $|\bfk'|\le|\bfk|-1$.  
By the induction hypothesis, $\lim_{\delta\to0}I_2(\delta)$ exists and
is independent of $f$ and $\bfj$.

For $I_3(\delta)$, the integrand has $(\bfi',\bfk)$–type divisorial 
singularities with $|\bfi'|\le|\bfi|-1$.  
Again by the induction hypothesis,  
$\lim_{\delta\to0}I_3(\delta)$ exists and is independent of $f$ and $\bfj$.

Since $\Z_{\ge0}\times\Z_{\ge0}$ is well–ordered by the above 
lexicographic order, the claim follows by transfinite induction.
\end{proof}

	\subsection*{Global case}
	Now, we generalize Cauchy principal values to complex manifolds.
	
	\begin{defn}
		Let $M$ be a complex manifold with a simple normal crossing divisor $D$. A coordinate chart \textbf{compatible with} $D$ is an open subset $U\subset M$ together with a biholomorphic map \[\varphi:U\rightarrow V\subset\mathbb{C}^{m+n},\] where $V$ is an open subset of $\mathbb{C}^{m+n}$, and $m,n$ are non-negative integers, such that $\varphi(D\cap U)=(\{z_1\cdots z_m=0\}\times \mathbb{C}^{n})\cap V$. We say that a property \textbf{holds locally}, if it holds in any coordinate chart compatible with $D$.
	\end{defn}
	
	\begin{defn}\label{divisorial type singularities}
		Let $M$ be a complex manifold with a simple normal crossing divisor $D$. Let $\alpha$ be a differential form on $M\backslash D$. We say $\alpha$ is a differential form with \textbf{divisorial type singularities} along $D$, if $\alpha$ is a differential form with divisorial type singularities locally. Likewise, we say $\alpha$ is \textbf{logarithm-free}, if it is logarithm-free locally.
	\end{defn}
	\begin{rem}
		For holomorphic vector bundles over complex manifolds with simple normal crossing divisors, we can also define sections with divisorial type singularities in a similar manner.
	\end{rem}
	One nice property of divisorial type singularities is that they are stable under pull back:
	\begin{prop}\label{pull back of divisorial type singularities}
		Let $M$ and $N$ be complex manifolds with simple normal crossing divisors $D$ and $D'$ respectively. Assume $F:M\rightarrow N$ is a holomorphic map which satisfies $F^{-1}(D')=D$ as sets, and $\alpha$ is a differential form on $N$ with divisorial type singularities along $D'$. Then $F^{*}(\alpha)$ is a differential form on $M$ with divisorial type singularities along $D$. Moreover, if $\alpha$ is logarithm-free, $F^{*}(\alpha)$ is also logarithm-free.
	\end{prop}
	\begin{proof}
		We only need to check that $F^{*}(\alpha)$ has divisorial type singularities (possible logarithm-free) locally. The local case can be derived from the following fact: for a open subset $U\subset \mathbb{C}^{m+n}$ and a holomorphic function $g:U\rightarrow\mathbb{C}$, which satisfies 
		\[
		g^{-1}(0)=\{(z_{1},z_{2},...,z_{m+n})\in U:z_{i}=0\text{ for some }i\in\{1,2,...,m\}\},
		\] there exist $\bfj \in \Z_{>0}^m$ and a nowhere vanishing holomorphic function $\tilde{g}$ on $U$, such that $g= \tilde{g}\, \bfz^{\bfj}$.
	\end{proof}
	Now, we can define the Cauchy principal value.
	
	\begin{defn}
		Let $M$ be a complex manifold with a simple normal crossing divisor $D$. Assume $\{U_{i}\}_{i\in I}$ is a coordinate cover compatible with $D$, and $\{\rho_{i}\}_{i\in I}$ is a partition of unity for $\{U_{i}\}_{i\in I}$. Given a compactly supported differential form with divisorial type singularities along $D$, which we denote by $\alpha$. The integral of $\alpha$ over $M$, \textbf{in the sense of Cauchy principal value}, is defined by
		\[
		\dashint_{M}\alpha=\sum_{i\in I}\dashint_{U_{i}}\rho_{i}\alpha.
		\]
	\end{defn}
	It is standard to prove the Cauchy principal value does not depend on the choice of $\{U_{i}\}_{i\in I}$ and $\{\rho_{i}\}_{i\in I}$.
	
	Finally, we introduce another way to characterize the Cauchy principal value on complex manifolds, which will be used in the main content.
	
	\begin{defn}\label{continuous defining function}
		Let $M$ be a complex manifold with simple normal crossing divisor $D$. A \textbf{smooth defining function of $D$} is a smooth function $h:M\rightarrow \mathbb{R}$ that satisfies the following conditions:
		\begin{enumerate}[(1)]
			\item $h$ is non-negative, and $h^{-1}(0)=D$.
			\item For each point $p\in D$, there exists an open neighborhood $U\subset M$ of $p$, a smooth positive function $f:U\rightarrow\mathbb{R}$, and a holomorphic function $g:U\rightarrow \mathbb{C}$, such that $h^{-1}(0)=D\cap U$ and $h|_{U}=fg\bar{g}$.
		\end{enumerate}
	\end{defn}
	
	The concept of smooth defining function is stable under pull back:
	\begin{prop}\label{pull back of defining function}Let $M$ and $N$ be complex manifolds with simple normal crossing divisors $D$ and $D'$ respectively. Assume $F:M\rightarrow N$ is a holomorphic map which satisfies $F^{-1}(D')=D$ as sets, and $h$ is a smooth defining function of $D'$. Then $h\circ F$ is a smooth defining function of $D$.
	\end{prop}
	\begin{proof}
		This can be verified directly. 
	\end{proof}
	The following proposition provides another way to characterize the Cauchy principal value:
	\begin{prop}\label{prop-a18}
		Let $M$ be a complex manifold with a simple normal crossing divisor $D$. Let $h$ be a smooth defining function of $D$, and let $\alpha$ be a compactly supported differential form with divisorial-type singularities along $D$.
		Then we have
		\[
		\lim_{\delta\rightarrow 0}\int_{h>\delta}\alpha=\dashint_{M}\alpha.
		\]
	\end{prop}
	\begin{proof}
		\def\supp{\mathrm{supp}}
		For each $p\in D$, we choose a coordinate neighborhood $U_{p}$ compatible with $D$, on which $h|_{U_{p}}=f_{p}g_{p}\bar{g}_{p}$, where $f_{p}:U_{p}\rightarrow \mathbb{R}$ is a smooth positive function and $g_{p}:U_{p}\rightarrow\mathbb{C}$ is holomorphic. Let $U_{0}=M\backslash D$, then $\{U_{0}\}\cup\{U_{p}\}_{p\in D}$ is a open cover of $M$. We denote a corresponding partition of unity by $\{\rho_0\}\cup\{\rho_{p}\}_{p\in D}$.
		
		We notice that $\supp (\rho_{0}\alpha)$ is compact, so $h|_{\supp(\rho_{0}\alpha)}$ has a nonzero minimum. As a consequence,
		\be
		\label{one tech}
		\lim_{\delta\rightarrow 0}\int_{h>\delta}\rho_{0}\alpha=\int_{M}\rho_{0}\alpha=\dashint_{M}\rho_{0}\alpha.
		\ee
		For $\rho_{p}\alpha$, we notice the following fact: for a open subset $U\subset \mathbb{C}^{m+n}$ and a holomorphic function $g:U\rightarrow\mathbb{C}$, which satisfies 
		\[
		g^{-1}(0)=\big\{(z_{1},z_{2},\cdots,z_{m+n})\in U:z_{i}=0\text{ for some }i\in\{1,2,...,m\}\big\},
		\]
		there exist $\bfj \in \Z_{>0}^m$ and a nowhere vanishing holomorphic function $\tilde{g}$ on $U$, such that $g=\tilde{g}\, \bfz^{\bfj}$. By using this fact and Theorem \ref{thm-a2}, we have
		\be
		\label{two tech}
		\lim_{\delta\rightarrow 0}\int_{h>\delta}\rho_{p}\alpha=\dashint_{M}\rho_{p}\alpha
		\ee
		for $p\in D$. Combining (\ref{one tech}) and (\ref{two tech}), we proved this proposition.
	\end{proof}
	
	\bibliography{lib}
	\bibliographystyle{plain}
\end{document}